\documentclass[11pt]{article}
\usepackage[margin=1in]{geometry}
\usepackage{amsmath,amssymb,amsthm,booktabs,array}
\usepackage{setspace,placeins}
\usepackage[round]{natbib}
\usepackage{xurl}
\usepackage[hidelinks]{hyperref}
\usepackage{pgfplots}
\pgfplotsset{compat=1.18}
\usepackage{siunitx}
\newcommand{\rz}[1]{\num[round-mode=places,round-precision=0]{#1}}
\hypersetup{pdftitle={Safety Research and the Allocation of Frontier Compute},pdfauthor={Li Gan}}
\newtheorem{proposition}{Proposition}
\newtheorem{corollary}{Corollary}
\newtheorem{lemma}{Lemma}
\newcommand{\alUnit}{31.317}
\newcommand{\alBudgetPct}{11.5}
\newcommand{\alBudgetOne}{2.3}
\newcommand{\alBudgetLow}{5.8}

\newcommand{\alBudgetMid}{34.5}
\newcommand{\alBudgetHigh}{69.1}
\newcommand{\alBudgetMonths}{2.1}
\newcommand{\alBoundaryS}{14.77}
\newcommand{\alBoundaryRisk}{57.68}
\newcommand{\alHx}{0.819}
\newcommand{\alHs}{-3.285}
\newcommand{\alKnowledgeTerm}{1.002}
\newcommand{\alExposureTerm}{0.183}

\newcommand{\alHazardAdvantage}{25}
\newcommand{\alCostCustomers}{39.33}
\newcommand{\alCostCustomersPoint}{0.39}
\newcommand{\alCostCustomersPointPct}{1.26}
\newcommand{\alBudgetFiveIndex}{3.61}
\newcommand{\alQuality}{29.66}

\newcommand{\alCostResearchLow}{37.05}
\newcommand{\alCostResearchHigh}{66.13}
\newcommand{\alPerCostResearchLow}{0.111}
\newcommand{\alPerCostResearchHigh}{0.062}
\newcommand{\alPerCostCustomers}{0.084}
\newcommand{\alThreshold}{63.5}
\newcommand{\alDelayPerPoint}{0.11}
\newcommand{\alClockZOne}{1.83}
\newcommand{\alClockHalfEnd}{22.1}
\newcommand{\alClockHalfDelay}{4.1}
\newcommand{\alClockHalfNoGrowth}{7.5}

\newcommand{\alRedeployedShare}{95}

\newcommand{\alSatConstX}{40.7}
\newcommand{\alSatConstS}{14.3}

\newcommand{\alSatConstRisk}{54.83}

\newcommand{\alSatConstGain}{2.85}

\newcommand{\alSatStageSwitch}{0.2}

\newcommand{\alSatStageLateS}{7.8}

\newcommand{\alResStageRisk}{38.60}
\newcommand{\alResStageSwitch}{0.5}

\newcommand{\alResStageEarlyS}{34.0}

\newcommand{\alResStageLateI}{93.0}

\newcommand{\alResStageEnd}{27.60}

\newcommand{\alResStageGainBoundary}{19.08}

\newcommand{\alResEarlyLoss}{21.76}
\newcommand{\alResRepay}{10.25}
\newcommand{\alResFlowRatio}{1.44}
\newcommand{\alResCrossF}{0.68}
\newcommand{\alRelaxedRisk}{35.89}

\newcommand{\alRelaxedDelay}{12.0}
\newcommand{\alKappaPct}{5.5}
\newcommand{\alMuHalf}{1.15}
\newcommand{\alSqrtTen}{3.16}
\newcommand{\alRiseLog}{1.10}
\newcommand{\alRiseTransferCoef}{1.28}

\newcommand{\alThresholdOne}{65.8}
\newcommand{\alThresholdLimit}{66.2}

\newcommand{\alAfterShareOne}{82}
\newcommand{\alTakeoffRatio}{1.25}
\newcommand{\alKnowledgeRatio}{1.29}
\newcommand{\alKnowledgeResearch}{5.82}
\newcommand{\alKnowledgeCustomer}{4.51}
\newcommand{\alPremiumThreshold}{34}
\newcommand{\alPremiumRedeployed}{64}
\newcommand{\alPremiumSaturated}{13}
\newcommand{\alSweepSixty}{48.3}
\newcommand{\alSweepSwitchZero}{65}
\newcommand{\alSweepSwitchOne}{70}
\newcommand{\alCapNoCut}{85}
\newcommand{\alCapStageGainMin}{3.77}
\newcommand{\alCapStageGainMax}{6.66}
\newcommand{\alFullRisk}{75.00}
\newcommand{\alCFStageRisk}{52.14}
\newcommand{\alCFSixRisk}{51.14}
\newcommand{\alCFTwelveRisk}{51.01}
\newcommand{\alCFTimingGain}{6.67}

\newcommand{\alSatTwelveDelay}{0.49}
\newcommand{\alRedTwelveAdd}{0.002}
\newcommand{\alSatTwelveAdd}{0.91}
\newcommand{\alSatStageAdd}{1.08}
\newcommand{\alCFTwelveFloorF}{0.5}
\newcommand{\alSatTwelveFullFrom}{0.7}
\newcommand{\alRedTwelveCutFrom}{0.7}

\newcommand{\alRedTwelveAddFifteen}{2.03}

\newcommand{\alRedTwelveDelayFifteen}{1.04}

\newcommand{\alSatTwelveAddFifteen}{7.78}

\newcommand{\alRedConstAddFifteen}{0.07}
\newcommand{\alRedTwelveBackFifteen}{0.15}

\newcommand{\alRedTwelveAddThirty}{8.92}

\newcommand{\alSatTwelveAddThirty}{15.53}

\newcommand{\alThresholdFifteen}{106.5}

\newcommand{\alCrossBudget}{30.2}
\newcommand{\alBoundaryIFive}{35.2}
\newcommand{\alBoundaryIThirty}{4.3}
\newcommand{\alRedLastNoCut}{28.8}
\newcommand{\alRedFirstCut}{34.5}
\newcommand{\alRedXFifteen}{48.3}
\newcommand{\alRedXThirty}{37.7}
\newcommand{\alSatXTwoHalf}{45.0}
\newcommand{\alSatXThirty}{20.1}
\newcommand{\alThetaC}{0.244}
\newcommand{\alOneMinusThetaQ}{0.478}
\newcommand{\alRdI}{0.352}
\newcommand{\alRdJpre}{27.7}
\newcommand{\alRdSqrtI}{0.594}
\newcommand{\alRdOnePlusPi}{1.343}
\newcommand{\alRdPremRed}{0.64}
\newcommand{\alRdPremSat}{0.13}
\newcommand{\alRdThreshold}{0.635}
\newcommand{\alRdTermQ}{0.284}
\newcommand{\alRdTermC}{0.513}
\newcommand{\alRdSqrtThreshold}{0.797}
\newcommand{\alRhoLow}{1.13}
\newcommand{\alRhoHigh}{1.69}
\newcommand{\alCrossShare}{22.0}
\newcommand{\alWeakRhoLow}{0.81}
\newcommand{\alWeakRhoHigh}{0.92}
\newcommand{\alWeakThreshFive}{38.2}
\newcommand{\alWeakThreshThirty}{139.9}
\newcommand{\alYieldHxZero}{16.3}
\newcommand{\alYieldSatCross}{13.6}
\newcommand{\alThresholdYieldQuarter}{54.3}
\newcommand{\alThresholdYieldTenth}{38.2}
\newcommand{\alThresholdYieldDouble}{65.2}
\newcommand{\alHxTenth}{-0.345}

\newcommand{\alKnowledgeTenth}{0.555}
\newcommand{\alExposureTenth}{0.900}
\newcommand{\alYieldTenthSatX}{50.0}

\newcommand{\alMaxSectionDelay}{5.50}
\newcommand{\alMeshRiskChange}{0.06}
\newcommand{\alMeshGainChange}{0.06}
\newcommand{\alMeshDelayChange}{0.01}
\newcommand{\alExpandPremiumSix}{73}
\newcommand{\alExpandCrossSix}{34.2}
\newcommand{\alExpandPremiumTwelve}{84}
\newcommand{\alExpandCrossTwelve}{38.6}
\newcommand{\alExpandPremiumTwentyFour}{116}
\newcommand{\alExpandCrossTwentyFour}{48.4}

\newcommand{\alAbsorbXSixty}{49.1}
\newcommand{\alAbsorbGainSixty}{0.02}

\newcommand{\alAbsorbXSeventyTwo}{47.7}
\newcommand{\alAbsorbGainSeventyTwo}{0.12}
\newcommand{\alCriticalAbsorb}{53.5}
\newcommand{\alAbsorbXFifteenMin}{38.0}
\newcommand{\alAbsorbXFifteenMax}{43.4}
\newcommand{\alHzThrMin}{61.8}
\newcommand{\alHzThrMax}{65.3}
\newcommand{\alHzAdvMin}{22}
\newcommand{\alHzAdvMax}{28}
\newcommand{\alHzCrossMin}{28.5}
\newcommand{\alHzCrossMax}{31.8}
\newcommand{\alHzSatAddMin}{0.07}
\newcommand{\alHzSatAddMax}{1.25}
\newcommand{\alHzRedAddMax}{0.30}
\newcommand{\alHzTimingMin}{0.53}
\newcommand{\alHzTimingMax}{8.54}
\newcommand{\alHzSigmaMin}{0.060}
\newcommand{\alHzSigmaMax}{0.117}
\newcommand{\alHzCFMin}{43.16}
\newcommand{\alHzCFMax}{58.12}
\newcommand{\alHrRiskLo}{25}
\newcommand{\alHrRegLo}{10}
\newcommand{\alHrThr}{63.5}
\newcommand{\alHrCFTwelveHi}{51.01}
\newcommand{\alHrCFTwelveLo}{13.76}
\newcommand{\alHrSatAddHi}{0.91}
\newcommand{\alHrSatAddLo}{0.33}

\newcommand{\alHrCFTwelveTen}{5.28}
\newcommand{\alHrSatAddTen}{0.13}

\newcommand{\alHrCFTwelveFive}{2.61}
\newcommand{\alHrSatAddFive}{0.07}
\newcommand{\alJlong}{984.785}
\newcommand{\alRatio}{31.45}
\newcommand{\alEquivalentLong}{0.366}
\newcommand{\alSigma}{0.086116}
\newcommand{\alGridCount}{11804}
\newcommand{\alStateCount}{13}
\newcommand{\alMaxNoSpill}{0.28}

\newcommand{\alMaxNoSpillFive}{0.15}

\newcommand{\alRedRegret}{5.90}

\newcommand{\alRedYieldSignal}{5.90}
\newcommand{\alRedTailRisk}{99.25}
\newcommand{\alRedAdverseRisk}{98.06}

\newcommand{\alSatRuleX}{41.0}

\newcommand{\alSatRegret}{4.86}
\newcommand{\alSatBoundaryRegret}{8.92}

\newcommand{\alIntegrationGap}{0.0166}
\newcommand{\alStartsAgree}{6}
\newcommand{\alStageRefinement}{0.0004}
\newcommand{\alPathCheck}{2\times10^{-6}}
\newcommand{\alDerivativeCheck}{2\times10^{-7}}
\newcommand{\alSqrtGap}{3\times10^{-7}}
\newcommand{\alRedRuleXThirty}{36.5}
\newcommand{\alSegmentCheck}{6\times10^{-8}}
\newcommand{\alSegmentAgree}{2}
\newcommand{\alCurvThrQuarter}{239.9}
\newcommand{\alCurvRedXQuarter}{37.9}

\newcommand{\alCurvThrThreeQuarter}{42.9}

\newcommand{\alCurvSatXThreeQuarter}{50.0}

\newcommand{\tkTakeoffBase}{18}
\newcommand{\tkComputeFoldBase}{8}
\newcommand{\tkFullRiskBase}{75.00}
\newcommand{\tkSigmaBase}{0.086}
\newcommand{\tkUnitBase}{31.317}
\newcommand{\tkBoundarySBase}{14.77}
\newcommand{\tkDelayPerPointBase}{0.11}
\newcommand{\tkThresholdBase}{63.5}
\newcommand{\tkCFTwelveRiskBase}{51.01}
\newcommand{\tkRedTwelveAddBase}{0.002}
\newcommand{\tkSatTwelveAddBase}{0.91}
\newcommand{\tkYieldDoubleRiseBase}{1.7}
\newcommand{\tkTailRiseBase}{2.7}
\newcommand{\tkCrossBudgetBase}{30.2}
\newcommand{\tkThresholdThirtyBase}{436.4}
\newcommand{\tkRedTwelveAddFifteenBase}{2.03}
\newcommand{\tkRedTwelveAddThirtyBase}{8.92}
\newcommand{\tkTakeoffLong}{36}
\newcommand{\tkComputeFoldLong}{64}
\newcommand{\tkFullRiskLong}{75.00}
\newcommand{\tkSigmaLong}{0.084}
\newcommand{\tkUnitLong}{115.687}
\newcommand{\tkBoundarySLong}{14.77}
\newcommand{\tkDelayPerPointLong}{0.15}
\newcommand{\tkThresholdLong}{60.9}
\newcommand{\tkCFTwelveRiskLong}{25.38}
\newcommand{\tkRedTwelveAddLong}{0.348}
\newcommand{\tkSatTwelveAddLong}{1.41}
\newcommand{\tkYieldDoubleRiseLong}{1.4}
\newcommand{\tkTailRiseLong}{-0.2}
\newcommand{\tkCrossBudgetLong}{34.0}
\newcommand{\tkThresholdThirtyLong}{371.7}
\newcommand{\tkRedTwelveAddFifteenLong}{2.70}
\newcommand{\tkRedTwelveAddThirtyLong}{5.00}

\newcommand{\alTdBase}{63.5}
\newcommand{\alTdTNine}{60.6}

\newcommand{\alTdTThirtySix}{60.9}
\newcommand{\alTdXHigh}{97.4}
\newcommand{\alTdXLow}{46.6}
\newcommand{\alTdEpsLow}{98.2}
\newcommand{\alTdEpsHigh}{41.1}
\newcommand{\alTdPsiMid}{75.2}
\newcommand{\alTdPsiHigh}{116.0}
\newcommand{\alTdAlphaLow}{95.0}
\newcommand{\alTdSafetyHigh}{91.1}
\newcommand{\alTdEtaTenth}{59989}
\newcommand{\alTdEtaQuarter}{239.9}

\newcommand{\alTdSpeedLow}{78.6}
\newcommand{\alTdSpeedHigh}{53.9}
\newcommand{\alTdFlatMin}{58.9}
\newcommand{\alTdFlatMax}{67.3}
\newcommand{\alTdDelayBase}{0.11}
\newcommand{\alTdDelayTNine}{0.07}
\newcommand{\alTdDelayTThirtySix}{0.15}
\newcommand{\alTdDelayXHigh}{0.09}

\newcommand{\alTdDelayEpsLow}{0.08}

\newcommand{\alTdGainBase}{0.91}

\newcommand{\alTdGainXHigh}{8.5}
\newcommand{\alTdGainEtaQuarter}{13.1}

\newcommand{\alTdCrossEpsRed}{0.31}

\newcommand{\alTdCrossEpsSat}{0.85}
\newcommand{\alTdCrossXRedPct}{64}
\newcommand{\alTdCrossXSatPct}{42}

\newcommand{\alTdCommitWorst}{1.13}

\newcommand{\alTdCommitSpillMin}{7.2}
\newcommand{\alTdCommitSpillMax}{10.7}
\newcommand{\alTdCommitCount}{12}
\newcommand{\alTimingMinPct}{10.46}
\title{Safety Research and the Allocation of Frontier Compute}
\author{Li Gan\thanks{Department of Economics, Texas A\&M University. Email: ganli@tamu.edu. AI software (Anthropic's Claude) was used under the author's direction to derive and check the algebra, write the calibration and simulation code, and edit the draft. All output was checked by the author, who is responsible for the content. An interactive companion page, which runs the paper's model in the browser, is at \url{https://ligan98.github.io/frontier-pacing/en/allocation.html}.}}
\date{November 2026}

\begin{document}
\maketitle
\begin{abstract}
\noindent During a frontier takeoff period, an AI lab splits a committed stock of compute among capability research, safety research, and customers, and a regulator limits both the services customers lose and how long any capability can be delayed. Extra safety compute can then come from customers or from capability research. Funding from capability research delays the end of the takeoff period, so the cheaper source depends on how much finishing raises customer services: customers above a threshold, capability research as well below it. Customer funding leaves capability dates unchanged, while research funding gives safety knowledge more time at the cost of later capability and customer services. The threshold is evaluated at a given allocation under a common continuation. Varying the calibration shows how delay costs, additional safety learning, and service valuation move it; in the baseline, it rises with the service budget. Customer-funded safety is front-loaded relative to customer use under the baseline timing conditions. Within the tested constant-share rules, commitment before the safety technology is known costs little unless safety work also advances capability. At the baseline budget, front-loaded customer-funded safety cuts catastrophe probability from \rz{\alFullRisk} to \rz{\alCFTwelveRisk} percent when full-speed catastrophe has a 75 percent probability, and from \alHrRiskLo{} to \rz{\alHrCFTwelveLo} percent when it has 25 percent. Either way it comes within one point of the best tested rule that also cuts capability research, and the funding threshold does not move.
\end{abstract}

\section{Introduction}

Proposals to pace the AI frontier would move compute into safety research, and they differ in where that compute comes from. The AI Futures Project proposes minimum compute shares for customers and for transparent safety research, limits on the model used for research, and a monthly risk ceiling \citep{aifutures2026}. \citet{amodei2026} argues for spending the time that slower development buys on safety research. These proposals raise a common question: which source should fund extra safety compute, customers or capability research?

This paper compares both funding sources in the customer services they cost. A lab divides growing compute among capability research, which advances the frontier, safety research, which builds knowledge that reduces later hazard, and customers, who value inference capacity and model quality. Below, research without a qualifier means capability research. A regulator minimizes cumulative catastrophe hazard subject to a \emph{service budget}, a cap on the discounted services customers give up, and a deadline on every capability milestone. Models of AI risk have studied the trade-off between speed and safety \citep{bdmdp2026} and between safety and consumption \citep{jones2016,aschenbrenner2020}. Here customers' inference compute is a second source of safety compute, and the common service cost makes the two sources comparable.

The first finding is that the value of finishing the frontier decides which source is cheaper at a given allocation. The comparison has a hazard side and a service side. On the hazard side, both sources buy the same direct increase in safety, and taking compute from research also gives knowledge more time before each capability arrives, at the price of longer exposure at each level (Proposition~\ref{prop:funding}). On the service side, research funding delays everything a finished frontier provides, which costs customers the \emph{value of finishing}: what completing a month earlier adds to their services (Proposition~\ref{prop:source}). So slowing capability research costs customers in two ways. Better models arrive later, and so does the compute that capability research releases at completion. The second cost depends on how much of that compute can serve customers, and it decides the comparison. Capability research is the cheaper marginal source when the \emph{finishing premium}, the proportional step up in customer services at completion, lies below a threshold, \alPremiumThreshold{} percent in the baseline (Corollary~\ref{cor:threshold}). Equivalently, it is cheaper when the customer share after completion lies below \alThreshold{} percent, the funding threshold used below. With the baseline valuation, if the freed compute can serve customers at once, finishing is valuable and customers should pay. If it cannot, delay costs little and capability research should pay too.

The second finding is that three quantities move the threshold, each by changing what one source costs or buys. The first is the delay a research cut causes. A large research share or a research pace that responds little to compute makes a research cut less costly in months, so research becomes a more attractive funding source. The length of the takeoff period matters little: a longer period makes each cut cost more months, but customer funding then gives up services over a longer period as well. The second is how much extra safety compute buys. When much safety progress occurs without dedicated compute, slowing research gives that progress more time while extra safety compute buys relatively little. The third is how customer services are valued. Faster diminishing returns to customer services favor research funding, because taking more compute from customers becomes costly sooner. One frontier lab's 2026 accounts put at least 73 percent of its compute spending outside paying customers, above the \alTdCrossXRedPct{} percent research share at which the answer flips. On those accounts capability research is the cheaper source even if customers could use all the compute it frees (Section~\ref{sec:drivers}).

The third finding is that a larger service budget moves funding toward capability research (Proposition~\ref{prop:budget}). A larger budget lets more safety compute come from customers, and each further unit costs them more, because the less compute customers keep, the more the next unit is worth to them. Slowing capability research does not become costlier in the same way. So once customers keep little enough compute, capability research becomes the cheaper source, provided that slowing it still cuts at least a fixed fraction of the hazard that customer funding cuts. In the calibration the funding threshold rises at every budget from \alBudgetOne{} to \alBudgetHigh{} percent of takeoff-period services, from \alThreshold{} percent at the baseline budget to \alThresholdFifteen{} percent at \alBudgetMid{} percent, where capability research is the cheaper source whatever customers can use. Outside the calibration the rise can stall if the hazard advantage of funding from capability research falls quickly as the customer share shrinks.

The fourth finding is that customer-funded safety should be front-loaded relative to customer compute (Proposition~\ref{prop:frontload}). Under the baseline timing conditions, this ratio falls along an interior optimum unless safety productivity rises fast. Knowledge bought early protects more of the remaining hazard, while a unit of customer compute grows more valuable as compute and model quality grow.

At the baseline budget, customer-funded timing captures nearly all the gains available from the tested joint rules. With a service budget of \alBudgetPct{} percent of the reference's takeoff-period services, front-loaded customer-funded safety is within \alSatTwelveAdd{} percentage points of the best tested rule that also allows research cuts. This comparison uses the AI Futures Project's worked example and the paper's square-root valuation. Allowing research cuts matters more when a larger service budget makes customer compute costly to remove, or when a higher research share makes delay cheap (Sections~\ref{sec:larger} and~\ref{sec:drivers}).

The fifth finding is that committing to a constant rule costs little within the tested technologies when safety work does not also advance capability. The comparison uses the same share grid for the chosen rule and its informed benchmark, and the finding persists across the parameter variants that move the funding threshold (Section~\ref{sec:uncertainty}). Capability spillovers matter more because they change what a rule costs, so learning about them changes which allocations the regulator can afford.

These numerical results describe a conditional scenario built from risk judgments, rather than estimated catastrophe probabilities. The funding answer is sensitive to the three drivers above, while the tested changes in hazard construction and compute growth move its threshold much less (Sections~\ref{sec:drivers} and~\ref{sec:robust}).

The funding comparison assumes that a finished frontier gives customers all the compute they can use, in the reference and in every policy. After completion hazard no longer depends on how compute is split, so with the reference held fixed, choosing this continuation for any policy loses nothing (Lemma~\ref{lem:continuation}). A model in which a finished frontier serves customers less than a slowed one credits safety research with gains that belong to the continuation (Appendix~\ref{app:reset}). When redeployment begins at completion, earlier completion brings the whole expansion forward, so under the baseline growth and valuation assumptions gradual redeployment makes finishing more valuable (Corollary~\ref{cor:gradual}).

This paper builds on \citet{gancommitments2026}, which uses the same technology to characterize when a risk ceiling and a development deadline can be promised together. Companion papers apply it to competing labs \citep{gan2026racing} and to agreements between countries \citep{gan2026agreements}. Here the question is who should pay for safety compute, and when. The closest work is \citet{bdmdp2026}, a static race in which firms, or a planner, choose once how to divide resources between speed and safety: the research-funded margin. Growth models of existential risk trade safety against consumption instead \citep{jones2016,jones2025,aschenbrenner2020}. This paper adds customer-serving inference to the speed--safety allocation, in one planner's problem for a lab's compute, and follows how safety knowledge and customer services accumulate over a takeoff period. Customers play the part consumption plays in the growth models. The resulting funding condition includes the services lost when research delays completion.

Timing is where the paper differs from the growth models. \citet{jones2016} divides research between ideas that raise consumption and ideas that save lives, and \citet{aschenbrenner2020} divides scientists between consumption and a stock of safety technology. When marginal utility declines fast, both shift resources toward safety as society grows richer. Proposition~\ref{prop:frontload} runs the other way, for three reasons made explicit in Section~\ref{sec:timing}. The service budget prices services the same at every date, danger ends with the takeoff period, and safety learning depends on the share of compute, not its amount. \citet{trammell2026} show that faster technological development can lower cumulative risk when stagnation is not safe; they treat technology as one-dimensional, studying its speed. The allocation here chooses a direction, moving compute between safety and capability, a quantitative form of differential technological development \citep{bostrom2002,bostrom2014}. \citet{jones2024} asks when AI progress should continue and when it should stop, and \citet{jones2025} how much to spend to reduce AI risk. His last page asks whether a tax on GPUs should fund safety research.

Models of AI races study why developers may neglect safety and how policy can change their choices \citep{armstrong2016,naude2020,han2020}. \citet{jensen2023} let racing developers split purchased compute between safety and performance, and \citet{fudenberg2026} study pacing in a continuous-time race against an exogenous safety threshold. Here a single lab is regulated and the deadline stands in for the rival. Work on the timing of AI deployment compares the safety gained by waiting with the hazard of waiting \citep{binder2021,bostrom2026}. \citet{fli2023} asked AI labs to pause for at least six months the training of systems more powerful than GPT-4. In \citet{acemoglu2024}, adoption is gradual when society learns about a technology's risks over time; in \citet{gans2025adoption}, learning from adoption itself generally favors faster adoption. \citet{guerreiro2023} show that taxes fall short when developers' beliefs about risk differ. \citet{gans2025direction} studies how to regulate the direction of innovation when one research path may cause harm. \citet{koh2026} derive an adaptive speed limit when a firm learns by doing and by waiting. Here learning is a dated stock that dedicated compute builds, so the regulator chooses how much to learn, when, and who pays for it. Section~\ref{sec:uncertainty} chooses a fixed rule by minimax regret \citep{savage1951,manski2004}. The hazard structure follows catastrophe models in environmental economics, where the probability of a catastrophe depends on a stock the planner controls \citep{cropper1976,clarke1994,tsur1996,polasky2011}. Here a second stock, safety knowledge, scales hazard down.

Commitments of compute to safety predate the pacing proposals. \citet{openai2023} pledged 20 percent of the compute it had secured to date, over four years, to superintelligence alignment. \citet{bengio2024} call for at least a third of AI R\&D budgets to go to safety and ethical use. Enforcing an allocation requires seeing how compute is used. Compute is detectable and quantifiable \citep{sastry2024}, chip monitoring can verify rules on large training runs \citep{shavit2023}, and cloud providers can see whether a workload is training or inference \citep{heim2024}. Telling safety research from capability research is harder, and legitimate safety research can advance capability \citep{aifutures2026}; the spillover states of Section~\ref{sec:uncertainty} model the second difficulty.

\section{The allocation problem}\label{sec:model}

Safety research competes with capability research and customer services for the same compute. So the funding choice changes both the frontier's exposure to danger and the services customers receive. Section~\ref{sec:background} describes the three uses of compute in words, the next five subsections formalize the regulator's choice, and Section~\ref{sec:scenario} supplies the baseline numbers.

\subsection{The three uses of compute}\label{sec:background}

A frontier lab uses its compute to serve customers, develop stronger models, and make those models safer. Serving customers is inference: models answer questions, write code, and do other paid work. Capability research trains stronger models and tests better methods, while safety research makes those models less likely to cause harm. The three uses compete for the same chips, so a unit of compute in safety research is a unit that neither serves customers nor advances the frontier. The paper's reference allocation puts 50 percent of compute into capability research, 5 percent into safety research, and 45 percent into customers. Call it full speed.

The frontier advances during a takeoff period. Capability runs from an Automated Coder milestone to a top-expert-dominating endpoint. Installed compute grows throughout the takeoff period, and a more capable research model does research faster, so research speed rises as the takeoff period proceeds. At full speed the reference reaches the endpoint in 18 months.

The danger is an absorbing loss of human control over AI. The paper reads this danger as misalignment. A more capable model is more likely to pursue goals its developers did not intend and better able to escape control, so the danger rises with capability. The hazard has two parts. State risk accrues per month at a capability level because a model in use has repeated opportunities to be misused or act against its developers' intentions. Step risk accrues per unit of capability gained because training a stronger model is dangerous in itself, however fast or slowly the frontier moves. Slowing research spreads the same capability gains over more months, which adds state risk and leaves step risk unchanged. Safety research builds a stock of knowledge that scales every later hazard down, so knowledge bought early protects more of what follows. Adding hazard over the takeoff period gives cumulative hazard, which determines the probability of catastrophe. At full speed, in the scenario, it is \rz{\alFullRisk} percent.

More safety compute must come from one of the other two uses, and the paper calls these the two funding sources. When customers pay, capability research keeps its share, every capability arrives on its reference date, and customers give up services. When research pays, customers keep their share and the frontier slows. Slower research gives knowledge more time to accumulate before each capability level is reached, which lowers hazard, but it also prolongs exposure to state risk at that level. It also delays everything a finished frontier provides. Both sources buy the same direct increase in safety compute; they differ only through these side effects. Who should pay is the question of which source removes more hazard per unit of services customers give up.

What a finished frontier provides depends on what happens after the takeoff period, which the paper calls the continuation. Capability stays at its endpoint, and safety research continues at its reference share. Capability research is no longer needed, so the question is how much of the compute it frees can serve customers right away. Two cases anchor the numerical work. Under full redeployment, the freed compute serves customers at once, so a finished frontier hands them all the compute that capability research no longer needs. Anthropic describes its compute as ``a single, fungible pool'' that it directs among research, model development, serving customers and internal use, and it states that demand for compute has historically exceeded available capacity \citep{anthropic2026s1}. Under no redeployment, the freed compute cannot serve customers, so they keep the reference share after completion. The same filing warns that if the lab's compute needs change, it may hold leased or dedicated infrastructure poorly suited to its actual needs \citep{anthropic2026s1}, as compute built for training might be once training ends. The value of finishing is what completing a month earlier adds to customer services. Under full redeployment it is large, because finishing hands customers a large step up in compute, so delaying completion is costly. Under no redeployment it is small, because the finished frontier serves them no more compute than the slowed one, so delay costs only the quality that arrives later.

The regulator makes two commitments. The service budget caps the discounted services customers give up relative to the reference. It is stated as a percentage of the services the reference delivers during its takeoff period, the period in which customer-funded safety costs customers anything; the main comparisons use \alBudgetPct{} percent, about \alBudgetMonths{} months of the reference's average takeoff-period services. The deadline caps the delay of every capability milestone, at 12 months in the scenario; it stands for a rival who could overtake the lab at any level. The service budget is what makes the two sources comparable: both are priced in the services customers lose. The regulator minimizes the probability of catastrophe within these two commitments.

Four terms recur. The hazard is the probability per month that a catastrophe happens if it has not happened yet. Safety knowledge is what safety research builds, and it scales down every later hazard. The funding threshold is the customer share after completion above which customers are the cheaper source of extra safety compute; below it, research should also pay. Front-loading means spending more of the safety compute early in the takeoff period, when compute and model quality are low.

The paper asks three questions. Who should pay for safety compute, customers or research? When should the safety compute be spent, early or late in the takeoff period? And what can a regulator commit to before it knows how safety research works? Section~\ref{sec:model} writes the model and calibrates it to the AI Futures Project's worked example. Section~\ref{sec:funding} answers the first question and Section~\ref{sec:timing} the second. Section~\ref{sec:quantitative} puts numbers on both with known technology and shows what moves the funding threshold. Section~\ref{sec:uncertainty} answers the third question.

In the formal model, the regulator chooses how the lab divides compute to reduce catastrophe risk while limiting lost customer services and delay. Capability and knowledge determine the hazard, customer services price each allocation, and the service budget and deadline restrict the feasible paths.

\subsection{Capability and the research clock}

Time is measured in months; Appendix~\ref{app:notation} lists the notation. Capability $F$ starts at zero and reaches a fixed endpoint at one; completion occurs at $T=\inf\{t:F(t)=1\}$. Installed compute grows at the rate $G>0$: $K(t)=e^{Gt}$, with $K(0)=1$. Before completion, shares $x(t)$, $s(t)$, and $i(t)$ go to capability research, safety research, and customer inference:
\begin{equation}
x+s+i=1,\qquad x,s,i\geq0.
\label{eq:resource}
\end{equation}
The constraint says that every unit of compute goes to exactly one use. Moving one point of share into safety means taking it from customers or from research; that choice is the paper's subject. Shares are fractions of a growing stock, so a constant share is a growing amount of compute. In the scenario, a 5 percent safety share commands twice as much compute after six months as at the start.

Capability research uses the current frontier model. Capability advances at the pace
\begin{equation}
\dot F=u(t,F,x)=\phi(F)\,\bigl(xK(t)\bigr)^{\varepsilon},
\qquad 0<\varepsilon<1.
\label{eq:capability}
\end{equation}
Compute and the capability of the research model both raise the pace. The term $xK$ is the compute used for capability research. The elasticity $\varepsilon$ says how the pace responds to it: with $\varepsilon=1/2$, halving the research share cuts the pace to $1/\sqrt2$ of its value, not to half. The function $\phi(F)>0$ is the \emph{research productivity} of a model of capability $F$. It captures recursive improvement: a stronger research model does research faster. The formal results hold for any continuous positive $\phi$. In the scenario, research at the endpoint is ten times as fast as at the start for the same compute. At a fixed share, \emph{research capacity} $(xK)^\varepsilon$ grows at the rate $\gamma=\varepsilon G$, so in the scenario compute doubling every six months gives research capacity doubling every twelve when $\varepsilon=1/2$.

The reference allocation, or full speed, uses the shares $(x_0,s_0,i_0)$ throughout; the scenario sets them to 50, 5, and 45 percent (Section~\ref{sec:scenario}). The reference arrival date at capability $F$ is $t_0(F)$, the reference pace there is $u_0(F)$, and the reference completes at $T_0$, 18 months in the scenario. Every policy is compared with these dates: a policy's delay at a capability level is how many months after $t_0(F)$ it arrives there.

The reference dates contain all that a policy needs to know about $\phi$. Write $m(t)=(x(t)/x_0)^\varepsilon$ for the \emph{research multiplier}: the pace at a given capability and date, relative to the pace at the reference share. Divide \eqref{eq:capability} by $\phi(F)$ and integrate from the start. The left side becomes $\int_0^F dv/\phi(v)$, which depends only on the capability reached. The right side is the research capacity accumulated since the start. The reference reaches $F$ at $t_0(F)$, so its accumulated capacity gives the left side's value at every level. A policy reaches capability $F$ when its own accumulated capacity matches that value, at the first date $t$ with
\begin{equation}
\gamma\int_0^t m(v)e^{\gamma v}\,dv=z(F),\qquad z(F)=e^{\gamma t_0(F)}-1.
\label{eq:clock}
\end{equation}
The left side is the research capacity the policy has accumulated by date $t$, scaled so that the reference's is $e^{\gamma t}-1$. The right side, $z(F)$, is what the reference had accumulated when it reached $F$. Call \eqref{eq:clock} the \emph{research clock}. It keeps time in accumulated research capacity instead of months. The research share sets how fast the clock runs, and the reference dates say which capability each reading reaches. Section~\ref{sec:funding} uses it for constant shares, and Appendix~\ref{app:numerics} for shares that change at milestones.

A constant share is the simplest case. With constant $m$, the left side of \eqref{eq:clock} is $m(e^{\gamma t}-1)$, so a constant share $x$ gives $e^{\gamma t}=1+z(F)/m$. Measured by $e^{\gamma t}-1$, every level arrives $1/m$ times as late as in the reference. In the scenario research capacity doubles every twelve months, so $z(1)=2^{18/12}-1\approx\alClockZOne$ at completion. A lab that halves its research share has $m=1/\sqrt2$, so it completes when $e^{\gamma t}=1+\alClockZOne\sqrt2$, at month \alClockHalfEnd. So halving research delays completion by \alClockHalfDelay{} months. Without compute growth, every level would arrive $\sqrt2$ times as late, and completion \alClockHalfNoGrowth{} months late. Growth shortens the delay, because the slowed lab reaches later levels when compute is more plentiful.

\subsection{Knowledge and risk}

Safety research builds a stock of knowledge that scales hazard down. Knowledge grows at
\begin{equation}
\dot S=\sigma(F)y(s),
\qquad S(0)=0.
\label{eq:knowledge}
\end{equation}
Here $\sigma(F)$ is \emph{safety productivity}: the knowledge a month of safety work at the reference share adds when capability is $F$. Knowledge is measured in units of log-hazard reduction: adding $\log2$ to $S$ halves every later hazard, as the hazard equation below makes precise. The \emph{learning multiplier} $y(s)$ scales that growth with the safety share, as the research multiplier scales the pace with the research share. It is continuous, increasing and differentiable, with $y(s_0)=1$; results that need more say so. In the scenario, half of reference-share learning continues with no dedicated compute. The other half has square-root returns, so quadrupling the safety share doubles that half. Section~\ref{sec:scenario} gives this form and fits $\sigma=\alSigma$ per month to a risk judgment. In this baseline, knowledge neither depreciates nor loses value at stronger capabilities; the uncertainty exercise relaxes both.

Conditional on no catastrophe so far, hazard is
\begin{equation}
\lambda(t)=[h(F)+k(F)\dot F]e^{-S}.
\label{eq:hazard}
\end{equation}
State risk $h(F)$ accrues per unit of time spent at capability $F$, and step risk $k(F)$ per unit of capability gained; both are nonnegative. The split follows the state risk and transition risk of \citet{trammell2026}. Cumulative hazard and the probability of catastrophe are
\begin{equation}
H=\int_0^T\lambda(t)\,dt,\qquad P=1-e^{-H}.
\label{eq:risk}
\end{equation}
Cumulative hazard adds the hazard rate over the takeoff period, and $e^{-H}$ is the probability of surviving it, so $P$ is one minus the survival probability. The catastrophe is an absorbing loss of human control over AI. The two parts of \eqref{eq:hazard} respond to speed differently. Slower research spends more months at each capability, so it adds state risk; step risk depends on how far the frontier moves, not on how fast. Knowledge scales both parts down in the same proportion. In the scenario the reference carries a \rz{\alFullRisk} percent probability of catastrophe, a cumulative hazard of $\log4$. Half of it is state risk, and three quarters of it falls in the second half of capability development. The baseline counts hazard only during the takeoff period, as the AI Futures Project's 75 percent does. So it assumes that the danger ends at completion, for example because a top-expert-dominating AI still under control can help align its successors \citep{aifutures2026}. Continuing danger is a separate numerical case.

The three equations already contain the trade-off that Section~\ref{sec:funding} prices. A smaller research share slows \eqref{eq:capability}, which moves every arrival date later. Knowledge in \eqref{eq:knowledge} then has more months to grow before each level is reached, while state risk in \eqref{eq:hazard} accrues over more months at each level. Which effect is larger is a question about the hazard functions and safety productivity, not about preferences.

\subsection{Customers and services}

Customers value inference capacity and model quality. Service quantity is $Q=e^{\mu F}\min\{i,\bar\imath\}K$, where $\bar\imath\leq1$ is the largest share customers can use, and its discounted value flow is
\begin{equation}
b(t,F,i)=e^{-rt}\sqrt{Q}
           =e^{(G/2-r)t+\mu F/2}\sqrt{\min\{i,\bar\imath\}}.
\label{eq:flow}
\end{equation}
The flow is a service index, measured in index units, not revenue or GDP. Model quality $e^{\mu F}$ multiplies the compute customers can use, $\min\{i,\bar\imath\}K$, and $e^{-rt}$ discounts the resulting service value. In the scenario, quality is tenfold higher at the endpoint and the annual discount rate is 3 percent. The square root gives diminishing marginal value: the marginal unit of customer compute is worth half the average unit. So a unit taken from customers costs them more the smaller their share. It also makes the allocation formula of Section~\ref{sec:timing} explicit. At fixed capability and share, the discounted flow grows at the rate $\kappa=G/2-r$: the square root halves the growth of compute, and discounting subtracts $r$. In the scenario $\kappa$ is \alKappaPct{} percent a month, positive because compute grows much faster than the discount rate. Serving capacity matters only where the cap binds, because while $i<\bar\imath$ the minimum drops out and the flow is $\sqrt i$ times the rest.

After completion, capability stays at one and safety research continues at the reference share. What happens next is the \emph{continuation}, which has two primitives. \emph{Serving capacity} $\bar\imath$ is the largest share of installed compute that can serve customers, before and after completion. \emph{Continuing danger} $\xi\geq0$ is the share of the endpoint state hazard that persists after completion, decaying as knowledge keeps growing. In the baseline $\xi=0$, and danger ends with the takeoff period. A finished frontier gives customers the \emph{continuation share} $i_c$, the same for the reference and every policy. Because the two share it, no comparison between a policy and the reference turns on a difference in what happens after the takeoff period. It turns on when the takeoff period ends and what happens during it. The paper assumes that customers receive all the compute that can serve them, $i_c=\min\{\bar\imath,1-s_0\}$. Section~\ref{sec:funding} shows that, with the reference held fixed, choosing it for any policy loses nothing. Two values of serving capacity anchor the numerical work:
\begin{itemize}
\item \emph{Full redeployment}: $\bar\imath=1$. Any share can serve customers, so a finished frontier gives them the compute that capability research no longer needs, $i_c=1-s_0$.
\item \emph{No redeployment}: $\bar\imath=i_0$. At most the reference share can serve customers, before and after completion.
\end{itemize}
The cap applies to a share of a growing stock, so even without redeployment customers receive ever more compute. Full redeployment is the lead case. Proposals for large inference allocations during restrained development presuppose that customers can absorb compute moved to inference \citep{aifutures2026}. That a finished frontier hands them the compute capability research no longer needs is this paper's assumption; Section~\ref{sec:funding} lets that compute arrive gradually. No redeployment is the alternative that reverses the funding answer at small budgets. In the scenario the reference customer share is 45 percent, so without redeployment customers stay at 45 percent after completion, while full redeployment gives them \alRedeployedShare{} percent. A third continuation, the \emph{reference reset}, returns customers to the reference share after completion. Section~\ref{sec:funding} shows that it provides fewer services without reducing hazard, and Appendix~\ref{app:reset} shows what it does to the results.

Rules are compared over the takeoff period and the 12 months after it. The deadline makes every rule finish within that window, by month 30 in the scenario. Once a rule and the reference have both finished, they give customers the same services, and in the baseline danger has ended, so every difference in services and hazard between rules falls inside the window. Installed compute over the window is taken as given. Compute is bought well before it is used. One frontier lab reports that most of its compute is procured under multi-year commitments made well in advance of deployment \citep{anthropic2026s1}. A lab that cannot put freed compute to use would eventually build less, but such a response comes after both paths have finished, so it leaves the comparison unchanged. Service costs are counted over 60 months, and any horizon from month 30 on gives the same cost (Section~\ref{sec:scenario}). With continuing danger, hazard extends past the window, and Section~\ref{sec:funding} shows that this barely moves the funding threshold.

\subsection{The budget and the objective}

For a service horizon $\tau$, let
\begin{equation}
J_\tau=\int_0^\tau b(t,F(t),i(t))\,dt,\qquad
C_\tau=J_{0,\tau}-J_\tau,
\label{eq:cost}
\end{equation}
where $J_{0,\tau}$ is the reference value over the same horizon and both paths follow the continuation after completion. $J_\tau$ adds up the discounted service flow over the horizon, and $C_\tau$ is the shortfall against the reference. It is the net discounted loss of services, from reallocated inference and from postponed quality; extra services at one date offset losses at another. A policy that serves customers less during the takeoff period and more after it nets the two. The reference value $J_{0,\tau}$ is the same number for every policy, so minimizing the cost and maximizing the policy's own services are the same thing. Both funding sources are priced in $C_\tau$: customer funding reduces inference during the takeoff period, while research funding postpones quality and completion. Both lower $J_\tau$, so customers pay through forgone services in either case. The model contains no prices, so it does not assign the monetary cost.

For known technology, the regulator chooses an allocation path to
\begin{equation}
\min_{x(\cdot),s(\cdot)} H
\quad\text{subject to}\quad
C_\tau\leq B,\qquad
\sup_{F\in[0,1]}[t(F)-t_0(F)]\leq L,
\quad F(T)=1.
\label{eq:planner}
\end{equation}
The regulator minimizes cumulative hazard while keeping lost services below $B$, limiting each milestone's delay to $L$ months, and completing the frontier. The customer share follows from the other two through \eqref{eq:resource}. A policy that falls more than $L$ months behind at an early milestone and then catches up still breaks the deadline. Minimizing $H$ and minimizing $P$ give the same choice. The deadline applies at every intermediate capability, not only at completion, and implies $T\leq T_0+L$; it stands for a lead over a rival whose response is not modeled. In the scenario $L$ is 12 months. The horizon satisfies $\tau\geq T_0+L$, so the service budget covers the whole transition.

The budget needs an absolute unit. Once a policy and the reference have both completed and caught up, the policy's cost stops growing, while the reference's services keep growing (Appendix~\ref{app:proof} proves this). So the same percentage permits a larger sacrifice at a longer horizon. Budgets here are fractions $b$ of a fixed unit $\bar J$: the reference's discounted services over its takeoff period, from month 0 to $T_0$, with reference shares throughout. That is the period in which customer funding costs anything. A customer-funded rule keeps every capability date, so after $T_0$ it serves customers exactly as the reference does. A constant customer share $i$ then costs $(1-\sqrt{i/i_0})\bar J$, and a budget $b$ buys the share with $\sqrt{i/i_0}=1-b$, whatever the length of the takeoff period or the growth of compute. No continuation changes $\bar J$, so every case faces the same absolute budget $B=b\bar J$. In the scenario $\bar J=\alUnit$ index units, so the \alBudgetPct{} percent budget of the main comparisons is \alBudgetFiveIndex{} index units, about \alBudgetMonths{} months of the reference's average takeoff-period services. That budget buys a definite amount of safety. With research held at its reference share, the largest safety share customers can fund at a constant rate is \alBoundaryS{} percent, up from the reference 5 percent.

\subsection{Information and policy classes}

The regulator can condition on the current state. The numerical work uses three smaller classes: constant shares, two-stage rules, and segment rules. Constant shares make the two funding sources explicit, because moving a point of share from one use to another is a single number. The two richer classes let the allocation change over the takeoff period, which Section~\ref{sec:timing} shows it should. Two-stage rules switch shares once, at a capability milestone $F_c$; segment rules hold shares constant on each of six or twelve capability segments. A two-stage rule can move safety compute toward the start of the takeoff period while holding research fixed; segment rules allow more such adjustments. Under uncertainty, constant shares are fixed before the technology is revealed (Section~\ref{sec:uncertainty}). One frontier lab implements the regulated allocation; compute uses are observable, allocations are enforceable, and the lab has no private objective in the model.

\subsection{A conditional baseline}\label{sec:scenario}

The baseline gives every result below a number. It takes its risk judgments from the AI Futures Project's worked example \citep{aifutures2026}, so its numbers show what that example's judgments imply; they are not estimates. The reference starts at an Automated Coder milestone, reaches a top-expert-dominating endpoint in 18 months, and carries a 75 percent probability of loss of control. Safety productivity is fitted so that a 12-month pause at the initial capability would cut that probability to 40 percent. The paper treats 75 percent as the upper end of a range and also reports 25 percent, with the pause judgment cutting hazard in the same proportion, to \alHrRegLo{} percent (Appendix~\ref{app:hazard}). During the pause the compute clock stops and safety keeps its reference share. Assigning the 40 percent judgment to that pause is a modeling choice, and it makes safety research productive. Half of reference hazard is state risk, and three quarters of it falls in the second half of capability development. Compute doubles every six months, and model quality raises both research speed and services per unit of inference tenfold by the endpoint. Reference shares are 50 percent research, 5 percent safety, and 45 percent customers; the annual discount rate is 3 percent. Appendix~\ref{app:calibration} gives the formulas. The hazard functions are built once and held fixed throughout; Section~\ref{sec:robust} varies their construction. Table~\ref{tab:baselinesources} lists the main numbers and where each comes from.

\begin{table}[htbp]\centering\small
\caption{Baseline numbers and their sources}\label{tab:baselinesources}
\begin{tabular}{@{}>{\raggedright\arraybackslash}p{5.2cm}>{\raggedright\arraybackslash}p{3.6cm}>{\raggedright\arraybackslash}p{6.6cm}@{}}\toprule
Number & Value & Source\\\midrule
Full-speed takeoff period, Automated Coder to endpoint & 18 months & AI Futures Project worked example\\
Probability of catastrophe at full speed & 75\% (25--75\%) & AI Futures Project worked example; 25\% is the lower case of Appendix~\ref{app:hazard}\\
Probability under a safety-case regime & 40\% (10--40\%) & AI Futures Project; read as a 12-month initial pause, a modeling choice; 10\% at the lower case\\
Longest milestone delay allowed & 12 months & AI Futures Project: a rival needs about a year to catch up\\
Growth of installed compute & doubles every 6 months & Illustrative path; Anthropic's compute spending doubled about every 6 months from mid-2025 to mid-2026 \citep{anthropic2026s1}, a fast case\\
Reference compute shares & 50\% research, 5\% safety, 45\% customers & Assumed. Anthropic: safety about 6\% of research compute in one week of 2026; paying customers at most 27--38\% of compute spending \citep{anthropic2026measuring,anthropic2026s1}\\
Research speed and service quality at the endpoint & 10 times the start & Assumed\\
Service budget, horizon, discount rate & \alBudgetPct\% of takeoff-period services $\bar J$, 60 months, 3\% a year & Policy settings of this paper\\\bottomrule
\end{tabular}\end{table}

The scenario also fixes the two functions that the model above leaves general. Research productivity rises tenfold over the takeoff period, $\phi(F)=\phi_0\cdot10^F$, and the constant $\phi_0$ is fitted to the 18-month takeoff period. The learning multiplier is
\begin{equation}
y(s)=\psi+(1-\psi)\left(\frac{s}{s_0}\right)^{\alpha},
\qquad \psi=\alpha=\frac12.
\label{eq:multiplier}
\end{equation}
The share $\psi$ of reference-share learning needs no dedicated compute, for example what the lab learns from studying its deployed models. The rest uses compute, with diminishing returns. Doubling safety compute from the reference 5 percent to 10 percent raises learning by $\tfrac12(\sqrt2-1)$, about a fifth.

Policies satisfy $0.02\leq x\leq0.50$, $0.05\leq s\leq0.70$, and $i\geq0.01$: research cannot exceed its reference share, and safety cannot fall below its reference share. The budget unit is $\bar J=\alUnit$ index units. Unless stated otherwise, every rule reported below meets a service budget of \alBudgetPct{} percent of $\bar J$ and the 12-month deadline. In the two consistent continuations the deadline is slack. No rule reported for the known technology uses more than \alMaxSectionDelay{} months of delay there, so the service budget does the work.

Service cost counts 60 months. That covers every transitional difference, because the deadline forces completion by month 30, and any later horizon gives the same absolute cost. With reference shares throughout, the reference's 60-month value is \alJlong, \alRatio{} times $\bar J$, so the \alBudgetPct{} percent budget equals \alEquivalentLong{} percent of five-year reference services. Stating the same \alBudgetPct{} percent at five years would multiply the permitted sacrifice by \alRatio.

\section{Who should pay for safety compute}\label{sec:funding}

The value of finishing determines the cheaper funding source for a given allocation and technology. Slowing research buys more learning time but delays customer services, so the comparison weighs the hazard removed against the services forgone. Propositions~\ref{prop:funding} and~\ref{prop:source} derive the two sides, and Corollary~\ref{cor:threshold} expresses their balance as the funding threshold. Proposition~\ref{prop:budget} shows how that threshold changes with the service budget.

Table~\ref{tab:twosources} sets the two sources side by side. Both put the same compute into safety, so they differ only in their side effects. A customer-funded unit costs services and delays nothing. A research-funded unit keeps customers whole during the takeoff period, gives knowledge more time and adds exposure, and delays the finished frontier and the services it brings.

\begin{table}[htbp]\centering\small
\caption{One more unit of safety compute, by funding source}\label{tab:twosources}
\begin{tabular}{@{}p{1.9cm}p{6.2cm}p{7.3cm}@{}}\toprule
 & Taken from customers & Taken from capability research\\\midrule
Costs & the services customers give up during the takeoff period & quality that arrives later, plus the value of finishing times the delay\\
Buys & safety knowledge & the same knowledge, more time for it before each capability, and longer exposure at each level\\
Capability & every milestone on its reference date & every later milestone delayed\\\bottomrule
\end{tabular}\end{table}

The numbers show where the comparison lands. At the \alBudgetPct{} percent customer-funded optimum, each percentage point taken from research delays completion by \alDelayPerPoint{} months, and a unit of safety share funded by research cuts \alHazardAdvantage{} percent more hazard than one funded by customers. Whether that is worth the services depends on how much freed compute can serve customers. If all the compute a finished frontier frees can serve them, research funding cuts \alPerCostResearchHigh{} of hazard per index unit of service cost and customer funding cuts \alPerCostCustomers, so customers should pay. If none of it can, research funding cuts \alPerCostResearchLow{} per index unit, so research should pay too. The break-even customer share after completion is \alThreshold{} percent. The rest of the section derives these numbers.

Constant shares make both funding sources explicit. Take a constant research share $x>0$, a constant safety share $s$, flat $\sigma$, and no capability spillover from safety work. The research multiplier $m=(x/x_0)^\varepsilon$ is then constant. Write $a(s)=\sigma y(s)$ for the knowledge a month adds. The research clock \eqref{eq:clock} gives the path in closed form:
\begin{equation}
t_x(F)=\frac1\gamma\log\left(1+\frac{z(F)}m\right),\qquad
u_x(F)=u_0(F)\,\frac{m+z(F)}{1+z(F)}.
\label{eq:constantpath}
\end{equation}
Here $t_x(F)$ is the date at which capability $F$ arrives, and $u_x(F)$ is the pace of progress there. A research share below $x_0$ makes $m<1$ and delays every level after the first. At the first level, where $z=0$, the pace is $m$ times the reference pace. At later levels the ratio rises toward one, because the slowed lab arrives later and has more compute. Since $S(F)=a(s)t_x(F)$, cumulative hazard can be written over capability:
\begin{equation}
H(x,s)=\int_0^1\left[\frac{h(F)}{u_x(F)}+k(F)\right]
                  e^{-a(s)t_x(F)}\,dF.
\label{eq:constantH}
\end{equation}
The first factor is the hazard met at capability $F$; the second is the protection that knowledge accumulated by the arrival date provides. Growing compute enters through $t_x$: slower research does not shift the reference path by a constant number of months.

The funding comparisons evaluate marginal changes at the baseline constant customer-funded optimum for an \alBudgetPct{} percent service budget. Research keeps 50 percent of compute, safety receives \alBoundaryS{} percent, and customers keep \alBoundaryIFive{} percent, so this is a policy allocation rather than the reference allocation with its 5 percent safety share.

Four steps turn \eqref{eq:constantH} into the hazard effect of each funding source.

\emph{Step 1: hazard over capability.} Equation \eqref{eq:constantH} is the first step. With constant shares and flat productivity, knowledge at the arrival of $F$ is $a(s)t_x(F)$, so hazard can be summed over capability levels instead of months.

\emph{Step 2: how a lower research share moves the path.} Taking one unit of share from research moves the arrival of capability $F$ later by $\ell(F)=-\partial t_x(F)/\partial x$ months and slows the pace there by the fraction $\vartheta(F)=\partial\log u_x(F)/\partial x$. Both follow from \eqref{eq:constantpath} by differentiating through $m$; the proposition gives their closed forms. The first level is not delayed at all, since $z(0)=0$, and later levels are delayed more.

\emph{Step 3: differentiate the integrand in $x$.} The protection factor $e^{-a(s)t_x(F)}$ contributes $a(s)\ell(F)$ times the integrand: faster research brings each level sooner, when less knowledge has accumulated. The time spent at $F$, $1/u_x(F)$, contributes $-\vartheta(F)$ times the state-risk part: faster research spends less time exposed at each level. The difference of the two integrals is $H_x$.

\emph{Step 4: the two directions.} The safety share enters only through $a(s)$, so $H_s$ is a single integral with a definite sign. When customers pay, the move is $(dx,ds,di)=(0,1,-1)$ and hazard changes by $H_s$. When research pays, the move is $(-1,1,0)$ and hazard changes by $H_s-H_x$. Research gives the larger cut exactly when $H_x>0$.

\begin{proposition}[Hazard effects of the two funding sources]\label{prop:funding}
In the constant-share environment, suppose differentiation under the integral is valid, and define
\[
\ell(F)=-\frac{\partial t_x(F)}{\partial x}
 =\frac{\varepsilon z(F)}{\gamma x[m+z(F)]},\qquad
\vartheta(F)=\frac{\partial\log u_x(F)}{\partial x}
 =\frac{\varepsilon m}{x[m+z(F)]}.
\]
Then
\begin{align}
H_x&=\underbrace{a(s)\int_0^1\ell(F)
       \left[\frac{h(F)}{u_x(F)}+k(F)\right]e^{-a(s)t_x(F)}\,dF}_{\text{less knowledge before later risk}}
\nonumber\\[-2pt]
&\hspace{8mm}-\underbrace{\int_0^1\frac{h(F)}{u_x(F)}\vartheta(F)
                    e^{-a(s)t_x(F)}\,dF}_{\text{less time exposed to state risk}},
\label{eq:Hx}\\
H_s&=-a'(s)\int_0^1t_x(F)
       \left[\frac{h(F)}{u_x(F)}+k(F)\right]e^{-a(s)t_x(F)}\,dF\leq0.
\label{eq:Hs}
\end{align}
One more unit of safety share changes hazard by $H_s$ when customers pay and by $H_s-H_x$ when capability research pays. Funding from research gives the larger hazard reduction exactly when $H_x>0$.
\end{proposition}

The proof is in Appendix~\ref{app:fundingproofs}.

Both sources buy the same direct knowledge gain; they differ through the research effect in \eqref{eq:Hx}. Here $\ell(F)$ is the delay, in months per unit of research share, with which capability $F$ arrives. Slower research gives knowledge that much more time to accumulate before each level is reached, but prolongs exposure to state risk at the level. The exposure term is the risk accounting of \citet{trammell2026}: slower passage lengthens exposure to state risk, while step risk proportional to speed does not depend on speed. The knowledge term is what a dedicated safety stock adds. Learning is necessary for the hazard advantage but not sufficient; the lower-productivity comparison gives a case in which knowledge accumulates and protects later risk, yet $H_x<0$.

At the comparison allocation, the hazard side favors funding from research. Here the knowledge term in \eqref{eq:Hx} is \alKnowledgeTerm{} and the exposure term \alExposureTerm{}, so $H_x=\alHx$ and $H_s=\alHs$. A unit of safety share funded by research cuts \alHazardAdvantage{} percent more hazard than one funded by customers, the ratio of $H_x$ to $-H_s$.

The hazard advantage needs high enough safety productivity. Scale safety productivity down with the hazard functions held fixed. At a tenth of baseline productivity, the knowledge term is \alKnowledgeTenth{} and the exposure term \alExposureTenth{}, so $H_x=\alHxTenth$. Knowledge still accumulates and protects later risk, and both sources still cut hazard, but customer funding cuts more. $H_x$ changes sign at \alYieldHxZero{} percent of baseline productivity.

The limiting cases separate the two effects. When neither existing nor additional safety compute produces learning, slower progress adds exposure. With only step risk, the added time helps knowledge accumulate. Appendix~\ref{app:fundingproofs} states these cases in Corollary~\ref{cor:failure} and distinguishes constant-share cuts from calendar-time pauses.

Hazard is half of the funding decision; services are the other half. Both depend on what happens after the takeoff period. Five steps give the service cost and the hazard reduction of each source, then compare them.

\emph{Step 1: customers pay.} Without spillover, the capability path does not depend on the customer share, so customer funding changes only services before completion. Write $T_x=t_x(1)$ for the completion date and
\[
J^{\mathrm{pre}}=\sqrt i\int_0^1\frac{e^{\kappa t_x(F)+\mu F/2}}{u_x(F)}\,dF
\]
for the discounted service value before completion. It is proportional to $\sqrt i$, so one more unit of share taken from customers costs $J^{\mathrm{pre}}/(2i)$: half the average value of their compute.

\emph{Step 2: research pays, in two pieces.} Let $V_c(T)$ be the discounted service value of the continuation when the takeoff period ends at $T$. Call $M(T)=-V_c'(T)$ the \emph{value of finishing}: what completing one month earlier adds to services. Total services are those before completion plus $V_c(T_x)$. Taking one unit of share from research moves every capability later and moves completion itself later, by $\ell(1)$ months.

\emph{Step 3: differentiate.} Each level arrives $\ell(F)$ months later, so customers receive its quality later; this quality loss is the term $Q_x$ below. Completion arrives $\ell(1)$ months later. During each extra month customers receive the slowed frontier's flow just before completion, $b_T=b(T_x,1,i)$, instead of the finished frontier's flow, $M(T_x)$. The \emph{endpoint term} $\ell(1)[M(T_x)-b_T]$ prices that difference.

\emph{Step 4: hazard after completion.} Let $\Phi e^{-S_T}$ be the hazard that remains after completion, which knowledge at completion, $S_T=a(s)T_x$, scales down as it scales down hazard during the takeoff period. Customer funding raises $S_T$ by $a'(s)T_x$. Funding from research raises it by the same amount and by $a(s)\ell(1)$ more, because knowledge has $\ell(1)$ more months to accumulate.

\emph{Step 5: compare.} Each source is judged by hazard reduction per unit of service cost. The cost of funding from research rises linearly in the value of finishing. So the comparison becomes a threshold on the finishing premium, $M(T_x)/b_T-1$: the proportional step up in service flow that finishing provides.

\begin{proposition}[The funding-source condition]\label{prop:source}
In the constant-share environment, let $0<i<\bar\imath$ and $H_s<0$. Let both funding directions be feasible for small moves (neither a control bound nor the deadline binds), and let the service horizon cover completion after such moves. One more unit of safety share costs discounted services
\begin{align}
c_C&=\frac{J^{\mathrm{pre}}}{2i}\qquad\text{when customers pay, and}
\label{eq:cC}\\
c_R&=\underbrace{\frac\mu2\sqrt i\int_0^1e^{\kappa t_x(F)+\mu F/2}\ell(F)\,dF}_{Q_x}
\;+\;\ell(1)\big[M(T_x)-b_T\big]
\qquad\text{when capability research pays.}
\label{eq:cR}
\end{align}
It cuts cumulative hazard, including the hazard after completion, by
\begin{align}
R_C&=-H_s+\Phi e^{-S_T}a'(s)T_x\qquad\text{when customers pay, and}
\label{eq:RC}\\
R_R&=H_x-H_s+\Phi e^{-S_T}\big[a'(s)T_x+a(s)\ell(1)\big]\qquad\text{when capability research pays.}
\label{eq:RR}
\end{align}
If $c_R>0$, research buys more hazard reduction per unit of service cost than customers exactly when $R_R/c_R>R_C/c_C$, that is, when the \emph{finishing premium} $M(T_x)/b_T-1$ lies below
\begin{equation}
\pi^\ast=\frac{c_C\,R_R/R_C-Q_x}{\ell(1)\,b_T}.
\label{eq:threshold}
\end{equation}
If $c_R<0$ and $R_R>0$, funding from research lowers hazard and raises services.
\end{proposition}

The proof is in Appendix~\ref{app:fundingproofs}.

Customer funding costs what the marginal unit of inference would have produced before completion. Funding from research costs the quality loss $Q_x$ plus the endpoint term, which is the completion delay $\ell(1)$ times the value of finishing, net of what the slowed frontier already provides. The endpoint term is a cost when the finished frontier serves customers better than the slowed one and a saving when it serves them worse. Delay pays for itself only if that saving exceeds the quality loss, so that $c_R<0$. On the hazard side, both sources add knowledge by completion, which protects whatever danger remains. Customer funding adds it through the larger safety share; funding from research adds it through that share and the extra time.

Each percentage point moved from research delays completion by \alDelayPerPoint{} months. At the comparison allocation customers keep $i=\alRdI$ and $J^{\mathrm{pre}}=\alRdJpre$. So customer funding costs $c_C=\alRdJpre/(2\times\alRdI)\approx\alCostCustomers$ index units per unit of share, or \alCostCustomersPoint{} index units per percentage point (\alCostCustomersPointPct{} percent of the budget unit). The \alBudgetPct{} percent budget is \alBudgetFiveIndex{} index units. Funding from research costs $Q_x=\alQuality$ in quality plus the endpoint term.

Under full redeployment, funding from research costs \alCostResearchHigh{} index units per unit of share. It cuts hazard by only \alPerCostResearchHigh{} per index unit, against \alPerCostCustomers{} for customer funding, so customers are cheaper at this allocation. Under no redeployment, funding from research costs \alCostResearchLow{} index units per unit of share and cuts hazard by \alPerCostResearchLow{} per index unit. Research is then the cheaper source.

The continuation determines what finishing provides. With the reference fixed, giving customers all the compute they can use after completion relaxes the budget without raising hazard. Lemma~\ref{lem:continuation} in Appendix~\ref{app:fundingproofs} establishes this result and the value of finishing, $M(T)=b(T,1,i_c)$.

After completion, hazard does not depend on how the remaining compute is split between customers and research that no longer advances capability. So the split matters for services alone. Holding the reference fixed allows the lemma to choose the policy's continuation without affecting the benchmark. Raising both paths' continuation share instead increases the reference's services between their completion dates, which makes delay costlier. As a result, full use of the compute that can serve customers is an assumption about the shared endpoint; the lemma shows that choosing it for any policy loses nothing when the reference is held fixed. The finishing premium is then $\sqrt{i_c/i}-1$, and the funding decision becomes a condition on the customer share after completion. At the \alBudgetPct{} percent customer-funded optimum, where $i=\alRdI$, the premium is $\sqrt{0.95/\alRdI}-1\approx\alRdPremRed$ under full redeployment and $\sqrt{0.45/\alRdI}-1\approx\alRdPremSat$ under no redeployment.

Three steps turn Proposition~\ref{prop:source} into a condition on the customer share after completion.

\emph{Step 1: the premium.} With the continuation of Lemma~\ref{lem:continuation}, the finished and the slowed frontier serve customers at the same date and quality. They differ only in the share, so the finishing premium is $\sqrt{i_c/i}-1$.

\emph{Step 2: square.} Research is the cheaper source when $\sqrt{i_c}<\sqrt i\,(1+\pi^\ast)$. When the right side is positive, squaring gives a threshold for $i_c=\min\{\bar\imath,1-s_0\}$; when it is not, no customer share satisfies the condition.

\emph{Step 3: continuing danger.} If a share $\xi$ of the endpoint state hazard persists after completion, decaying as knowledge grows at rate $\sigma$, it integrates to $\Phi e^{-S_T}$ with $\Phi=\xi h(1)/\sigma$. The ratio of hazard reductions $R_R/R_C$ then moves from the takeoff-period ratio $(H_x-H_s)/(-H_s)$ toward the ratio of the knowledge the two sources add by completion. It rises with $\xi$ exactly when the second ratio exceeds the first.

\begin{corollary}[A threshold in the finishing premium]\label{cor:threshold}
Under the continuation of Lemma~\ref{lem:continuation} and continuing danger $\xi$, so that $\Phi=\xi h(1)/\sigma$, capability research is the cheaper source exactly when
\begin{equation}
\min\{\bar\imath,1-s_0\}<\bar\imath^\ast(\xi)=\big[\max\{0,\sqrt i\,(1+\pi^\ast)\}\big]^2.
\label{eq:capthreshold}
\end{equation}
If $\sqrt i\,(1+\pi^\ast)\le0$, no customer share after completion makes research the cheaper source. Continuing danger raises the signed root $\sqrt i\,(1+\pi^\ast)$ if and only if $1+a(s)\ell(1)/[a'(s)T_x]>(H_x-H_s)/(-H_s)$. That is, the knowledge the two sources add by completion stands in a higher ratio than their hazard reductions during the takeoff period. The funding threshold is then nondecreasing in continuing danger and rises strictly on the branch where the signed root is positive; before that branch it can stay at zero. As continuing danger grows, $\bar\imath^\ast$ approaches its value with $R_R/R_C$ replaced by $1+a(s)\ell(1)/[a'(s)T_x]$.
\end{corollary}

The proof is in Appendix~\ref{app:fundingproofs}.

With no continuing danger, research is the cheaper source when the finishing premium is below \alPremiumThreshold{} percent. The customer share after completion must then be below $(\alRdSqrtI\times\alRdOnePlusPi)^2\approx\alRdThreshold$, or \alThreshold{} percent. Full redeployment lies far above the threshold (Figure~\ref{fig:source}, panel (a)): finishing raises the frontier's service flow by \alPremiumRedeployed{} percent. No redeployment lies below it: finishing adds only \alPremiumSaturated{} percent. The two continuations select different funding sources at the same comparison allocation.

Continuing danger barely moves the threshold. By completion, a unit of share moved into safety adds \alKnowledgeResearch{} units of knowledge when research pays and \alKnowledgeCustomer{} when customers pay. That ratio, \alKnowledgeRatio{}, is close to the ratio of takeoff hazard reductions, \alTakeoffRatio{}. So the threshold rises only to \alThresholdOne{} percent at $\xi=1$, when \alAfterShareOne{} percent of hazard falls after completion. However much danger continues, it is at most \alThresholdLimit{} percent (Figure~\ref{fig:source}, panel (b)). Full redeployment stays above the threshold in every case. Panel (c) shows the best constant rule across serving capacity. It cuts research to \alSatConstX{} percent at a serving capacity of 45 percent and to \alSweepSixty{} percent at 60 percent. It leaves research at 50 percent from a capacity of \alSweepSwitchZero{} percent, or from \alSweepSwitchOne{} percent when $\xi=1$.

\begin{figure}[htbp]\centering
\begin{minipage}{.325\textwidth}\begin{tikzpicture}\begin{axis}[width=\linewidth,height=5.4cm,xlabel style={align=center},tick label style={font=\scriptsize},label style={font=\footnotesize},title style={font=\small},grid=major,grid style={gray!15},xlabel={Continuation share\\(\%)},ylabel={Hazard cut per unit of cost},title={(a) Marginal funding sources},width=.9\linewidth,xmin=30,xmax=95,ymin=0,ymax=0.179,legend style={draw=none,font=\scriptsize,at={(0.98,0.98)},anchor=north east},scaled y ticks=false,yticklabel style={/pgf/number format/fixed,/pgf/number format/precision=2}]
\addplot[thick,red!75!black,no marks] coordinates {(30.00000,0.16237) (31.00000,0.15698) (32.00000,0.15203) (33.00000,0.14744) (34.00000,0.14319) (35.00000,0.13923) (36.00000,0.13554) (37.00000,0.13208) (38.00000,0.12884) (39.00000,0.12580) (40.00000,0.12293) (41.00000,0.12022) (42.00000,0.11766) (43.00000,0.11523) (44.00000,0.11293) (45.00000,0.11075) (46.00000,0.10867) (47.00000,0.10668) (48.00000,0.10479) (49.00000,0.10299) (50.00000,0.10126) (51.00000,0.09960) (52.00000,0.09802) (53.00000,0.09649) (54.00000,0.09503) (55.00000,0.09363) (56.00000,0.09227) (57.00000,0.09097) (58.00000,0.08972) (59.00000,0.08851) (60.00000,0.08734) (61.00000,0.08621) (62.00000,0.08512) (63.00000,0.08406) (64.00000,0.08304) (65.00000,0.08205) (66.00000,0.08109) (67.00000,0.08016) (68.00000,0.07926) (69.00000,0.07839) (70.00000,0.07754) (71.00000,0.07671) (72.00000,0.07591) (73.00000,0.07513) (74.00000,0.07437) (75.00000,0.07363) (76.00000,0.07291) (77.00000,0.07221) (78.00000,0.07152) (79.00000,0.07085) (80.00000,0.07020) (81.00000,0.06957) (82.00000,0.06895) (83.00000,0.06834) (84.00000,0.06775) (85.00000,0.06718) (86.00000,0.06661) (87.00000,0.06606) (88.00000,0.06552) (89.00000,0.06499) (90.00000,0.06448) (91.00000,0.06397) (92.00000,0.06348) (93.00000,0.06299) (94.00000,0.06252) (95.00000,0.06205)};
\addlegendentry{Research pays}
\addplot[thick,blue!70!black,dashed,no marks] coordinates {(30.00000,0.08353) (31.00000,0.08353) (32.00000,0.08353) (33.00000,0.08353) (34.00000,0.08353) (35.00000,0.08353) (36.00000,0.08353) (37.00000,0.08353) (38.00000,0.08353) (39.00000,0.08353) (40.00000,0.08353) (41.00000,0.08353) (42.00000,0.08353) (43.00000,0.08353) (44.00000,0.08353) (45.00000,0.08353) (46.00000,0.08353) (47.00000,0.08353) (48.00000,0.08353) (49.00000,0.08353) (50.00000,0.08353) (51.00000,0.08353) (52.00000,0.08353) (53.00000,0.08353) (54.00000,0.08353) (55.00000,0.08353) (56.00000,0.08353) (57.00000,0.08353) (58.00000,0.08353) (59.00000,0.08353) (60.00000,0.08353) (61.00000,0.08353) (62.00000,0.08353) (63.00000,0.08353) (64.00000,0.08353) (65.00000,0.08353) (66.00000,0.08353) (67.00000,0.08353) (68.00000,0.08353) (69.00000,0.08353) (70.00000,0.08353) (71.00000,0.08353) (72.00000,0.08353) (73.00000,0.08353) (74.00000,0.08353) (75.00000,0.08353) (76.00000,0.08353) (77.00000,0.08353) (78.00000,0.08353) (79.00000,0.08353) (80.00000,0.08353) (81.00000,0.08353) (82.00000,0.08353) (83.00000,0.08353) (84.00000,0.08353) (85.00000,0.08353) (86.00000,0.08353) (87.00000,0.08353) (88.00000,0.08353) (89.00000,0.08353) (90.00000,0.08353) (91.00000,0.08353) (92.00000,0.08353) (93.00000,0.08353) (94.00000,0.08353) (95.00000,0.08353)};
\addlegendentry{Customers pay}
\addplot[gray,densely dotted,no marks] coordinates {(63.52005,0.00000) (63.52005,0.17860)};
\end{axis}\end{tikzpicture}\end{minipage}\hfill
\begin{minipage}{.325\textwidth}\begin{tikzpicture}\begin{axis}[width=\linewidth,height=5.4cm,xlabel style={align=center},tick label style={font=\scriptsize},label style={font=\footnotesize},title style={font=\small},grid=major,grid style={gray!15},xlabel={Share of hazard\\after completion (\%)},ylabel={Threshold share (\%)},title={(b) Continuing danger},xmin=0,xmax=100,ymin=62,ymax=68]
\addplot[thick,black,mark=*,mark size=1pt] coordinates {(0.00000,63.52005) (2.26275,63.59193) (4.42536,63.66011) (6.49435,63.72488) (8.47565,63.78650) (10.37473,63.84517) (12.19660,63.90112) (13.94588,63.95452) (15.62682,64.00555) (17.24335,64.05436) (18.79910,64.10109) (31.64856,64.47817) (48.08037,64.93840) (53.65151,65.08915) (58.14290,65.20880) (64.93821,65.38673) (69.83531,65.51270) (73.53210,65.60657) (76.42170,65.67922) (78.74247,65.73712) (80.64732,65.78434) (82.23886,65.82360) (90.25392,66.01850) (95.85944,66.15212)};
\addplot[gray,densely dotted,no marks] coordinates {(0.00000,66.24943) (100.00000,66.24943)};
\end{axis}\end{tikzpicture}\end{minipage}\hfill
\begin{minipage}{.325\textwidth}\begin{tikzpicture}\begin{axis}[width=\linewidth,height=5.4cm,xlabel style={align=center},tick label style={font=\scriptsize},label style={font=\footnotesize},title style={font=\small},grid=major,grid style={gray!15},xlabel={Serving capacity\\(\%)},ylabel={Research share (\%)},title={(c) Best constant rule},xmin=45,xmax=100,ymin=38,ymax=52,legend style={draw=none,font=\scriptsize,at={(0.98,0.03)},anchor=south east}]
\addplot[thick,black,solid,mark=*,mark size=0.9pt] coordinates {(45.00000,40.74770) (50.00000,41.66921) (55.00000,45.35031) (60.00000,48.28507) (65.00000,50.00000) (70.00000,50.00000) (75.00000,50.00000) (80.00000,50.00000) (85.00000,50.00000) (90.00000,50.00000) (95.00000,50.00000) (100.00000,50.00000)};
\addlegendentry{$\xi=0$}
\addplot[thick,blue!70!black,dashed,mark=*,mark size=0.9pt] coordinates {(45.00000,40.74770) (50.00000,40.93365) (55.00000,44.62531) (60.00000,47.57765) (65.00000,49.96172) (70.00000,50.00000) (75.00000,50.00000) (80.00000,50.00000) (85.00000,50.00000) (90.00000,50.00000) (95.00000,50.00000) (100.00000,50.00000)};
\addlegendentry{$\xi=0.25$}
\addplot[thick,red!75!black,densely dotted,mark=*,mark size=0.9pt] coordinates {(45.00000,40.74770) (50.00000,40.59594) (55.00000,44.29466) (60.00000,47.25633) (65.00000,49.65065) (70.00000,50.00000) (75.00000,50.00000) (80.00000,50.00000) (85.00000,50.00000) (90.00000,50.00000) (95.00000,50.00000) (100.00000,50.00000)};
\addlegendentry{$\xi=1$}
\end{axis}\end{tikzpicture}\end{minipage}
\caption{Who should pay for safety compute. (a) Hazard reduction per index unit of service cost when one unit of share moves into safety at the customer-funded optimum (research 50\%, safety \alBoundaryS\%), by the continuation share $i_c$, the customer share a finished frontier provides; the dotted line marks the threshold of Corollary~\ref{cor:threshold}. (b) The threshold as continuing danger $\xi$ grows, plotted against the share of hazard that falls after completion at the same allocation; the dotted line is its limit. (c) Research share of the best constant rule by serving capacity $\bar\imath$, with the continuation of Lemma~\ref{lem:continuation}; at 50\% customers pay for all extra safety. All rules meet the \alBudgetPct\% budget and the 12-month deadline.}\label{fig:source}
\end{figure}
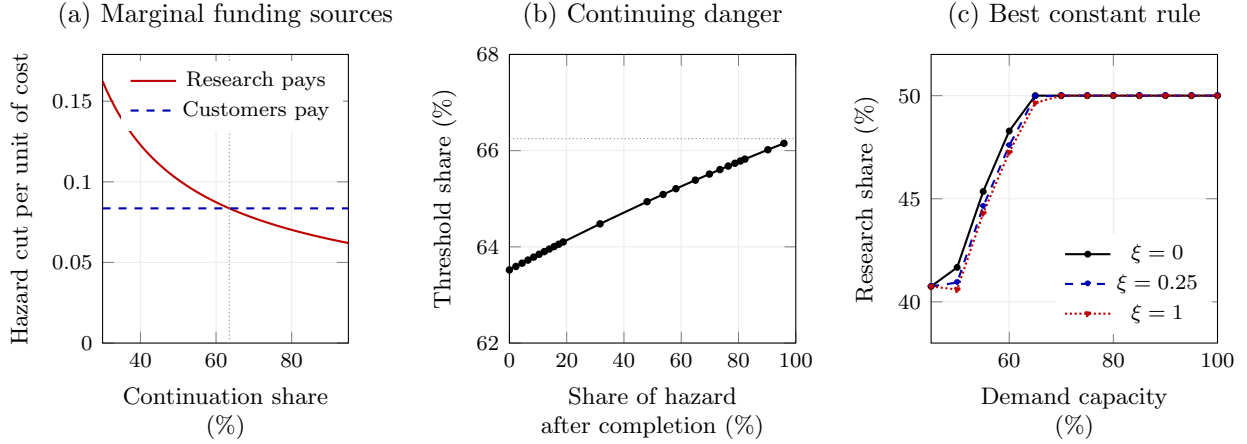

A larger service budget changes the marginal funding comparison because the customer-funded optimum leaves customers less compute. Corollary~\ref{cor:threshold} evaluates its threshold locally at that allocation. Along the customer-funded optimum, the threshold takes the following form.

\emph{Step 1: what does not depend on the customer share.} With research at $x_0$, the capability path is the reference path, so $t_x$, $u_x$, $\ell$, and $T_x=T_0$ do not move with $i$. Then $J^{\mathrm{pre}}$ and $Q_x$ are proportional to $\sqrt i$, and the slowed frontier's flow $b_T$ is $e^{\kappa T_0+\mu/2}\sqrt i$.

\emph{Step 2: insert.} Write $E=\ell(1)e^{\kappa T_0+\mu/2}$, so that $\ell(1)b_T=E\sqrt i$. The signed root $\sqrt i\,(1+\pi^\ast)$ becomes $\sqrt i+(c_C\,R_R/R_C-Q_x)/E$: a customer term that grows like $1/\sqrt i$, because $c_C$ does, and a quality term that shrinks like $\sqrt i$.

\emph{Step 3: a small customer share.} As $i\to0$ the quality term vanishes and the customer term grows without bound. A sufficient condition is that the hazard reduction from research funding stays bounded away from zero relative to that from customer funding.

\emph{Step 4: the slope.} Differentiating the root in $i$ gives the condition under which the threshold falls as customers keep more, that is, rises with the budget.

\begin{proposition}[The budget and the funding source]\label{prop:budget}
In the constant-share environment, consider allocations that keep research at $x_0$ and leave customers the share $i$, with the rest to safety. Let the conditions of Proposition~\ref{prop:source} hold. Where its root is positive, the threshold of Corollary~\ref{cor:threshold} satisfies
\begin{equation}
\sqrt{\bar\imath^\ast(i)}=(1-\theta_Q)\sqrt i+\theta_C\,\frac{\rho(i)}{\sqrt i},
\qquad \rho(i)=\frac{R_R}{R_C},
\label{eq:budgetthreshold}
\end{equation}
where $\theta_C=c_C\sqrt i/E$ and $\theta_Q=Q_x/(E\sqrt i)$, with $E=\ell(1)e^{\kappa T_0+\mu/2}$, are positive and do not depend on $i$.
\begin{itemize}
\item[(a)] If funding from research cuts hazard, with $\rho(i)\geq\underline\rho>0$ for small $i$, then $\bar\imath^\ast(i)\to\infty$ as $i\to0$. Whatever the continuation, capability research becomes the cheaper marginal source once the customer share is small enough.
\item[(b)] If $\rho$ is differentiable and the right side of \eqref{eq:budgetthreshold} is positive, $\bar\imath^\ast$ falls with $i$ exactly when ${(1-\theta_Q)\,i}<\theta_C\,[\rho(i)-2i\rho'(i)]$. In particular, it does so whenever $\rho'(i)\leq0$ and ${(1-\theta_Q)\,i}<\theta_C\,\rho(i)$.
\end{itemize}
\end{proposition}

The proof is in Appendix~\ref{app:fundingproofs}.

A larger budget leaves customers less compute at the customer-funded optimum. With square-root service value, the cost of taking one more unit from customers rises like $1/\sqrt i$ as their share shrinks. The quality term of funding from research shrinks like $\sqrt i$, and its endpoint term stays bounded.

Part (a) needs no hazard advantage. It assumes that research funding keeps reducing hazard by a positive fraction of the reduction from customer funding. Part (b) gives the condition under which the threshold rises with the budget. Diminishing returns to services do not by themselves guarantee it; the hazard ratio $\rho$ must not fall too fast as the customer share shrinks.

In the calibration, $\theta_C=\alThetaC$ and $1-\theta_Q=\alOneMinusThetaQ$. At the \alBudgetPct{} percent optimum, $i=\alRdI$ and $\rho=\alTakeoffRatio$, so \eqref{eq:budgetthreshold} reads $\alRdTermQ+\alRdTermC=\alRdSqrtThreshold$, the square root of the threshold \alRdThreshold. The customer term already carries most of the weight. The ratio $\rho$ rises from \alRhoLow{} at a budget of \alBudgetOne{} percent of $\bar J$ to \alRhoHigh{} at \alBudgetHigh{} percent. With $\rho$ rising, the threshold rises at every budget from \alBudgetOne{} to \alBudgetHigh{} percent (Section~\ref{sec:larger}).

Immediate redeployment is not essential to the comparison. When discounted service flows grow, finishing earlier brings a gradual expansion forward and makes delay costlier. Serving capacity that grows in calendar time can instead make finishing less valuable. Appendix~\ref{app:fundingproofs} states both cases in Corollary~\ref{cor:gradual} and gives the numerical example.

\section{How the allocation should change over time}\label{sec:timing}

On interior intervals, the baseline optimum front-loads safety relative to customer use. A shadow value measures what an extra unit of capability or knowledge contributes over the remaining path. For a regular optimum of \eqref{eq:planner},\footnote{A regular optimum is an optimal path with piecewise continuous controls, a nonzero multiplier on hazard, and finitely many junctions with the deadline constraint.} let $\nu\geq0$ be the service-budget multiplier and $q_F,q_S$ the shadow costs, or costates, of capability and knowledge in the minimization of $H-\nu J_\tau$. Write
\begin{equation}
A(t)=\int_t^T\lambda(v)\,dv,\qquad
W(t)=-[q_F(t)+k(F(t))e^{-S(t)}].
\label{eq:shadows}
\end{equation}
$A(t)$ is the hazard still ahead. One more unit of knowledge scales every later hazard down by the same proportion, so the costate equation gives $q_S=-A$ (Appendix~\ref{app:proof}). $W$ is the net value of a capability gain, including later services and the deadline, after the step hazard that the gain incurs. The continuation enters $W$, through the value of finishing, but not $q_S$.

At a date with no binding control restriction, the regulator maximizes
\begin{equation}
W u(t,F,x)+A\sigma(F)y(s)+\nu b(t,F,i)
\quad\text{over}\quad x+s+i=1.
\label{eq:pointwise}
\end{equation}
The three terms are the values of current research, current safety learning, and current services. If all three uses are interior,
\begin{equation}
W u_x=A\sigma(F)y'(s)=\nu b_i.
\label{eq:marginals}
\end{equation}
Here $u_x$ and $b_i$ are derivatives in the corresponding shares. The subscript in the constant-share pace $u_x(F)$ instead identifies a path. For a small increase in safety, research is the cheaper source if $Wu_x<\nu b_i$, and customers are cheaper otherwise. This is the pointwise counterpart of Proposition~\ref{prop:source}, with the shadow values carrying the future.

With positive value coefficients and no binding allocation constraints, the square-root technologies assign shares in proportion to squared coefficients (Proposition~\ref{prop:allocation}, Appendix~\ref{app:timingproofs}). Those values depend on the whole path. The timing result follows directly from equal marginal values for safety and customers, in four steps.

\emph{Step 1: the knowledge costate.} One more unit of knowledge scales every later hazard down in the same proportion, so its shadow value is the hazard still ahead: $q_S=-A$.

\emph{Step 2: equal marginal values.} Where safety and customers are both interior, \eqref{eq:marginals} equates the marginal value of safety, $A\sigma(F)y'(s)$, to that of customers, $\nu b_i$. Squaring gives the ratio $s/i$, whatever research does.

\emph{Step 3: differentiate in time.} Along the optimum $\dot A=-\lambda$ and the budget multiplier $\nu$ is constant, so log-differentiating the ratio leaves four terms: productivity, remaining hazard, compute growth net of discounting, and quality.

\emph{Step 4: signs.} With flat productivity, $\mu\geq0$, and $G\geq2r$, every term is nonpositive, so the ratio falls.

\begin{proposition}[Relative front-loading of safety]\label{prop:frontload}
Assume the baseline knowledge law with the learning multiplier \eqref{eq:multiplier} at $\alpha=1/2$, the square-root service flow, and a binding service budget ($\nu>0$). Consider an interval of a regular optimum on which safety and customer shares are positive, $i<\bar\imath$, and neither margin is otherwise constrained. Then
\begin{equation}
\frac{s(t)}{i(t)}=
\left[\frac{A(t)\sigma(F(t))(1-\psi)}
{\sqrt{s_0}\,\nu e^{\kappa t+\mu F(t)/2}}\right]^2,
\label{eq:ratio}
\end{equation}
and, for differentiable positive $\sigma$ and $A$,
\begin{equation}
\frac{d}{dt}\log\frac{s}{i}
=2\frac{d}{dt}\log\sigma(F(t))-
  2\frac{\lambda(t)}{A(t)}-G+2r-\mu\dot F(t).
\label{eq:ratioslope}
\end{equation}
If $\sigma$ is flat, $\mu\geq0$, and $G\geq2r$, the ratio is nonincreasing, and strictly decreasing whenever hazard is being incurred or service quality rises. Rising productivity reverses the direction exactly when
\begin{equation}
\frac{d}{dt}\log\sigma(F(t))>
\frac{\lambda(t)}{A(t)}+\frac G2-r+\frac\mu2\dot F(t).
\label{eq:risethreshold}
\end{equation}
\end{proposition}

The proof is in Appendix~\ref{app:timingproofs}.

Three forces push safety early relative to customer use. Knowledge bought early protects more of the remaining hazard, so its value falls at rate $\lambda/A$. A unit of customer share serves more as installed compute grows, at rate $G/2$, and as quality rises, at rate $\mu\dot F/2$. Only discounting, at rate $r$, pushes the other way.

If both margins stay interior for a full month, the baseline implies a decline of at least \alTimingMinPct{} percent in the safety/customer ratio. This bound uses flat safety productivity and only compute growth net of discounting: the ratio is at most $e^{-2\kappa}$ times its starting value. The hazard and quality terms strengthen the decline. This illustration does not extend across a binding floor, cap, or deadline junction.

Condition \eqref{eq:risethreshold} has a simple necessary form. Since $\lambda/A>0$ and $\mu\dot F\geq0$, a reversal requires $d\log\sigma/dt>G/2-r=\kappa$: safety productivity must grow faster than the discounted service flow does at fixed capability and share. In the calibration that is $\log2/12-0.03/12$, or \alKappaPct{} percent a month. If productivity rises with capability as $\sigma(F)=\sigma_0e^{\zeta F}$, the condition reads $(\zeta-\mu/2)\dot F>\lambda/A+G/2-r$. So $\zeta>\mu/2=\alMuHalf$ is necessary: safety productivity must rise more than $e^{\mu/2}=\alSqrtTen$-fold across the takeoff period. The threefold rise in the uncertainty exercise ($\zeta=\alRiseLog$) cannot reverse the decline at any date.

The direction rests on three modeling choices, and each matters for how the result compares with the literature. First, the service budget gives each unit of services the same weight at every date, so its multiplier $\nu$ stays constant even as the service flow grows. When a richer society values marginal consumption less, growth instead raises the value of safety relative to consumption \citep{jones2016,aschenbrenner2020,trammell2026}. With such a valuation, the direction of $s/i$ also depends on how the customer share moves (Appendix~\ref{app:proof}). Second, danger ends with the takeoff period, so remaining hazard $A$ shrinks to zero. This part of the result does not depend on how services are valued. Third, the condition $G\geq2r$ reflects share-based learning: the output of a safety share does not grow with installed compute, while the value of a customer share does. When the compute-using part of safety research scales with the square root of installed compute, the compute-growth terms cancel (Appendix~\ref{app:proof}).

Proposition~\ref{prop:frontload} restricts the ratio $s/i$, not the safety share alone, because research can move at the same time. It applies on intervals where both margins are interior, not across a binding safety floor, a binding serving capacity, or a junction with the deadline constraint. With knowledge that loses value at stronger capabilities, $-q_S$ becomes transfer-weighted remaining hazard, and Appendix~\ref{app:proof} gives the modified condition. With continuing danger, $q_S(T)=-\Phi e^{-S_T}$, so $A$ also counts the hazard that remains after completion. That slows the decline but, with flat productivity and $G\geq2r$, does not reverse it.

\section{Who pays and when: known technology}\label{sec:quantitative}

At the baseline budget, front-loaded customer-funded safety captures nearly all the gains from the tested joint allocations, and research cuts add value as the budget or calibration changes. The section first isolates customer-funded timing, which is the same in every continuation and serves as the benchmark for the rest. It then shows when research should also pay at the baseline budget, how larger budgets shift funding toward research, what moves the threshold, and which checks leave the funding answer unchanged.

\subsection{Timing paid by customers}

Customer-funded safety is the same in every continuation. With research at 50 percent, every capability arrives on the reference date, so the continuation cancels from the rule's service cost. At a constant share, the largest affordable safety share is \alBoundaryS{} percent, and the probability of catastrophe falls from \alFullRisk{} to \alBoundaryRisk{} percent.

Timing that customers pay for buys much more (Table~\ref{tab:classes}, panel A). Customer-funded rules lower the probability of catastrophe to \alCFStageRisk{} percent with two stages, and to \alCFSixRisk{} and \alCFTwelveRisk{} percent with six and twelve capability segments. The twelve-segment rule gains a further \alCFTimingGain{} points over the constant rule, with no delay in either continuation. The rule starts with more than a third of compute in safety and lowers the share in every segment until it reaches its floor at $F=\alCFTwelveFloorF$ (Figure~\ref{fig:paths}). Its safety-to-customer ratio falls throughout, as Proposition~\ref{prop:frontload} predicts, while research stays at its upper bound, which the proposition allows. Early safety is cheap in services because installed compute and model quality are low, while it protects almost all of the hazard of the takeoff period. These rules come from a restricted class, so they illustrate the timing result rather than test its interior condition.

At this budget, allowing research cuts adds little. Under full redeployment, the best rules found in every class keep research at or near 50 percent. The best constant and two-stage rules are customer-funded. The twelve-segment rule trims research by less than two percentage points from $F=\alRedTwelveCutFrom$ on, which gains \alRedTwelveAdd{} points. Under no redeployment, cutting research adds \alSatConstGain{} points to the constant rule, \alSatStageAdd{} to the two-stage rule, and \alSatTwelveAdd{} to the twelve-segment rule. So at this budget, among the tested policy classes, customer-funded timing is nearly all that the joint allocation achieves, in either continuation.

\begin{table}[htbp]\centering\small
\caption{Customer-funded rules and rules that may also cut research}\label{tab:classes}
\setlength{\tabcolsep}{5pt}
\begin{tabular}{@{}lrrrrr@{}}\toprule
& Customer-funded & \multicolumn{4}{c}{Research may also be cut}\\
\cmidrule(l){3-6}
& both continuations & \multicolumn{2}{c}{Full redeployment} & \multicolumn{2}{c}{No redeployment}\\
\cmidrule(lr){3-4}\cmidrule(l){5-6}
Budget and policy class & Risk (\%) & Risk (\%) & Delay & Risk (\%) & Delay\\\midrule
\addlinespace[3pt]\multicolumn{6}{@{}l}{\emph{A. Budget \alBudgetPct\% of $\bar J$}}\\
\quad Constant shares & 57.68 & 57.68 & 0.00 & 54.83 & 1.17\\
\quad Two stages & 52.14 & 52.14 & 0.00 & 51.06 & 0.46\\
\quad Six segments & 51.14 & 51.14 & 0.01 & 50.26 & 0.50\\
\quad Twelve segments & 51.01 & 51.01 & 0.01 & 50.10 & 0.49\\
\addlinespace[3pt]\multicolumn{6}{@{}l}{\emph{B. Budget 34.5\% of $\bar J$}}\\
\quad Constant shares & 40.91 & 40.84 & 0.19 & 34.08 & 3.29\\
\quad Twelve segments & 35.97 & 33.94 & 1.04 & 28.18 & 2.34\\
\addlinespace[3pt]\multicolumn{6}{@{}l}{\emph{C. Budget 69.1\% of $\bar J$}}\\
\quad Constant shares & 31.20 & 27.89 & 1.62 & 19.99 & 5.50\\
\quad Twelve segments & 30.03 & 21.11 & 2.82 & 14.50 & 5.06\\
\bottomrule\end{tabular}
\par\smallskip\begin{minipage}{\textwidth}\footnotesize
Risk is the probability of catastrophe; full speed gives 75.00\%. Delay is the largest milestone delay in months. Customer-funded rules keep research at 50\% throughout, so they delay no capability and give the same risk and service cost in both continuations. Two-stage rules switch shares once, at a milestone chosen from $\{0.05,0.1,0.2,\ldots,0.9\}$; segment rules hold shares constant between capability milestones at 0.1, 0.2, 0.3, 0.5, and 0.75 (six segments) or 0.05, 0.1, 0.15, 0.2, 0.25, 0.3, 0.4, 0.5, 0.6, 0.7, and 0.85 (twelve). Every rule meets the budget and the 12-month deadline. Entries other than full speed and the constant customer-funded rules are feasible multistart candidates within their policy classes, not certified global optima.
\end{minipage}\end{table}

\begin{figure}[htbp]\centering
\begin{minipage}{.325\textwidth}\begin{tikzpicture}\begin{axis}[width=\linewidth,height=5.0cm,tick label style={font=\scriptsize},label style={font=\small},title style={font=\small},grid=major,grid style={gray!15},xlabel={Capability},xmin=0,xmax=1,ymin=0,ymax=60,ylabel={Share (\%)},title={(a) Safety}]
\addplot[thick,black,solid,no marks] coordinates {(0.00000,5.00000) (0.05000,5.00000) (0.05000,5.00000) (0.10000,5.00000) (0.10000,5.00000) (0.15000,5.00000) (0.15000,5.00000) (0.20000,5.00000) (0.20000,5.00000) (0.25000,5.00000) (0.25000,5.00000) (0.30000,5.00000) (0.30000,5.00000) (0.40000,5.00000) (0.40000,5.00000) (0.50000,5.00000) (0.50000,5.00000) (0.60000,5.00000) (0.60000,5.00000) (0.70000,5.00000) (0.70000,5.00000) (0.85000,5.00000) (0.85000,5.00000) (1.00000,5.00000)};
\addplot[thick,blue!70!black,dashed,no marks] coordinates {(0.00000,35.41686) (0.05000,35.41686) (0.05000,29.82658) (0.10000,29.82658) (0.10000,24.68809) (0.15000,24.68809) (0.15000,20.19919) (0.20000,20.19919) (0.20000,16.39459) (0.25000,16.39459) (0.25000,13.22452) (0.30000,13.22452) (0.30000,9.55239) (0.40000,9.55239) (0.40000,5.96554) (0.50000,5.96554) (0.50000,5.00000) (0.60000,5.00000) (0.60000,5.00000) (0.70000,5.00000) (0.70000,5.00000) (0.85000,5.00000) (0.85000,5.00000) (1.00000,5.00000)};
\addplot[thick,green!45!black,dashdotted,no marks] coordinates {(0.00000,36.02167) (0.05000,36.02167) (0.05000,26.51310) (0.10000,26.51310) (0.10000,20.45103) (0.15000,20.45103) (0.15000,16.25719) (0.20000,16.25719) (0.20000,13.17506) (0.25000,13.17506) (0.25000,11.07132) (0.30000,11.07132) (0.30000,9.68749) (0.40000,9.68749) (0.40000,8.20349) (0.50000,8.20349) (0.50000,6.77757) (0.60000,6.77757) (0.60000,5.07230) (0.70000,5.07230) (0.70000,5.00000) (0.85000,5.00000) (0.85000,5.00000) (1.00000,5.00000)};
\end{axis}\end{tikzpicture}\end{minipage}\hfill
\begin{minipage}{.325\textwidth}\begin{tikzpicture}\begin{axis}[width=\linewidth,height=5.0cm,tick label style={font=\scriptsize},label style={font=\small},title style={font=\small},grid=major,grid style={gray!15},xlabel={Capability},xmin=0,xmax=1,ymin=0,ymax=60,ylabel={Share (\%)},title={(b) Research}]
\addplot[thick,black,solid,no marks] coordinates {(0.00000,50.00000) (0.05000,50.00000) (0.05000,50.00000) (0.10000,50.00000) (0.10000,50.00000) (0.15000,50.00000) (0.15000,50.00000) (0.20000,50.00000) (0.20000,50.00000) (0.25000,50.00000) (0.25000,50.00000) (0.30000,50.00000) (0.30000,50.00000) (0.40000,50.00000) (0.40000,50.00000) (0.50000,50.00000) (0.50000,50.00000) (0.60000,50.00000) (0.60000,50.00000) (0.70000,50.00000) (0.70000,50.00000) (0.85000,50.00000) (0.85000,50.00000) (1.00000,50.00000)};
\addplot[thick,blue!70!black,dashed,no marks] coordinates {(0.00000,50.00000) (0.05000,50.00000) (0.05000,50.00000) (0.10000,50.00000) (0.10000,50.00000) (0.15000,50.00000) (0.15000,50.00000) (0.20000,50.00000) (0.20000,50.00000) (0.25000,50.00000) (0.25000,50.00000) (0.30000,50.00000) (0.30000,50.00000) (0.40000,50.00000) (0.40000,50.00000) (0.50000,50.00000) (0.50000,50.00000) (0.60000,50.00000) (0.60000,50.00000) (0.70000,50.00000) (0.70000,50.00000) (0.85000,50.00000) (0.85000,50.00000) (1.00000,50.00000)};
\addplot[thick,green!45!black,dashdotted,no marks] coordinates {(0.00000,41.88588) (0.05000,41.88588) (0.05000,46.26082) (0.10000,46.26082) (0.10000,47.61412) (0.15000,47.61412) (0.15000,47.15810) (0.20000,47.15810) (0.20000,45.47391) (0.25000,45.47391) (0.25000,43.92868) (0.30000,43.92868) (0.30000,45.31251) (0.40000,45.31251) (0.40000,46.79651) (0.50000,46.79651) (0.50000,48.22243) (0.60000,48.22243) (0.60000,49.92770) (0.70000,49.92770) (0.70000,50.00000) (0.85000,50.00000) (0.85000,50.00000) (1.00000,50.00000)};
\end{axis}\end{tikzpicture}\end{minipage}\hfill
\begin{minipage}{.325\textwidth}\begin{tikzpicture}\begin{axis}[width=\linewidth,height=5.0cm,tick label style={font=\scriptsize},label style={font=\small},title style={font=\small},grid=major,grid style={gray!15},xlabel={Capability},xmin=0,xmax=1,ymin=0,ymax=60,ylabel={Share (\%)},title={(c) Customers},legend to name=pathlegend,legend columns=3,legend style={draw=none,font=\small,/tikz/every even column/.append style={column sep=8pt}}]
\addplot[thick,black,solid,no marks] coordinates {(0.00000,45.00000) (0.05000,45.00000) (0.05000,45.00000) (0.10000,45.00000) (0.10000,45.00000) (0.15000,45.00000) (0.15000,45.00000) (0.20000,45.00000) (0.20000,45.00000) (0.25000,45.00000) (0.25000,45.00000) (0.30000,45.00000) (0.30000,45.00000) (0.40000,45.00000) (0.40000,45.00000) (0.50000,45.00000) (0.50000,45.00000) (0.60000,45.00000) (0.60000,45.00000) (0.70000,45.00000) (0.70000,45.00000) (0.85000,45.00000) (0.85000,45.00000) (1.00000,45.00000)};
\addlegendentry{Reference}
\addplot[thick,blue!70!black,dashed,no marks] coordinates {(0.00000,14.58314) (0.05000,14.58314) (0.05000,20.17342) (0.10000,20.17342) (0.10000,25.31191) (0.15000,25.31191) (0.15000,29.80081) (0.20000,29.80081) (0.20000,33.60541) (0.25000,33.60541) (0.25000,36.77548) (0.30000,36.77548) (0.30000,40.44761) (0.40000,40.44761) (0.40000,44.03446) (0.50000,44.03446) (0.50000,45.00000) (0.60000,45.00000) (0.60000,45.00000) (0.70000,45.00000) (0.70000,45.00000) (0.85000,45.00000) (0.85000,45.00000) (1.00000,45.00000)};
\addlegendentry{Customer-funded, both continuations}
\addplot[thick,green!45!black,dashdotted,no marks] coordinates {(0.00000,22.09245) (0.05000,22.09245) (0.05000,27.22608) (0.10000,27.22608) (0.10000,31.93485) (0.15000,31.93485) (0.15000,36.58472) (0.20000,36.58472) (0.20000,41.35102) (0.25000,41.35102) (0.25000,45.00000) (0.30000,45.00000) (0.30000,45.00000) (0.40000,45.00000) (0.40000,45.00000) (0.50000,45.00000) (0.50000,45.00000) (0.60000,45.00000) (0.60000,45.00000) (0.70000,45.00000) (0.70000,45.00000) (0.85000,45.00000) (0.85000,45.00000) (1.00000,45.00000)};
\addlegendentry{No redeployment, research may be cut}
\end{axis}\end{tikzpicture}\end{minipage}
\par\smallskip\ref{pathlegend}
\caption{Best twelve-segment rules under the \alBudgetPct\% budget and the 12-month deadline, by capability. The customer-funded rule is the same in both continuations and delays no capability. Under full redeployment, the best rule that may also cut research is not shown: it trims research by less than two percentage points from $F=\alRedTwelveCutFrom$ on and otherwise differs by less than 0.1 percentage points in any share. Under no redeployment customers use at most 45\%.}\label{fig:paths}
\end{figure}
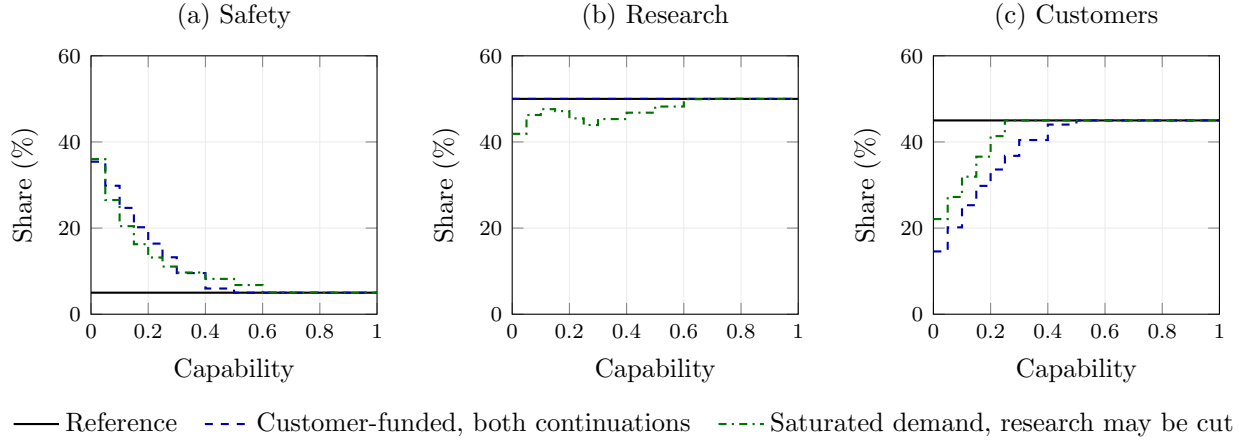

\subsection{When capability research should also pay}

\begin{table}[htbp]\centering\small
\caption{Known technology: who pays for safety compute, and when}\label{tab:known}
\setlength{\tabcolsep}{4pt}
\begin{tabular}{@{}lrrrrrrr@{}}\toprule
Policy & Research & Safety & Customers & Risk & Cost & Delay & Switch\\
& (\%) & (\%) & (\%) & (\%) & (\%) & (months) & ($F_c$)\\\midrule
Full speed & 50.0 & 5.0 & 45.0 & 75.00 & 0.00 & 0.00 & \\
\addlinespace[3pt]\multicolumn{8}{@{}l}{\emph{A. Full redeployment}}\\
\quad Customer-funded safety & 50.0 & 14.8 & 35.2 & 57.68 & 11.52 & 0.00 & \\
\quad Best constant rule & 50.0 & 14.8 & 35.2 & 57.68 & 11.52 & 0.00 & \\
\quad Best two-stage rule & 50.0/50.0 & 25.8/5.0 & 24.2/45.0 & 52.14 & 11.52 & 0.00 & 0.3\\
\addlinespace[3pt]\multicolumn{8}{@{}l}{\emph{B. No redeployment}}\\
\quad Customer-funded safety & 50.0 & 14.8 & 35.2 & 57.68 & 11.52 & 0.00 & \\
\quad Best constant rule & 40.7 & 14.3 & 45.0 & 54.83 & 11.52 & 1.17 & \\
\quad Best two-stage rule & 44.6/47.2 & 27.2/7.8 & 28.2/45.0 & 51.06 & 11.52 & 0.46 & 0.2\\
\addlinespace[3pt]\multicolumn{8}{@{}l}{\emph{C. Reference reset}}\\
\quad Customer-funded safety & 50.0 & 14.8 & 35.2 & 57.68 & 11.52 & 0.00 & \\
\quad Best constant rule & 37.0 & 13.4 & 49.6 & 54.51 & 11.52 & 1.73 & \\
\quad Best two-stage rule & 32.3/2.0 & 34.0/5.0 & 33.7/93.0 & 38.60 & 11.52 & 9.60 & 0.5\\
\bottomrule\end{tabular}
\par\smallskip\begin{minipage}{\textwidth}\footnotesize
Risk is the probability of catastrophe. Cost is discounted service loss over 60 months, in percent of the budget unit $\bar J$; every intervention row meets the same \alBudgetPct\% budget and the 12-month deadline. Two-stage entries show shares before/after the switch milestone $F_c$, chosen from $\{0.05,0.1,0.2,\ldots,0.9\}$. Research is bounded below by 2\% and safety by 5\%. The continuation cases differ only in what customers can use and what they receive after completion (Section~\ref{sec:model}); Appendix~\ref{app:reset} explains panel C. Rows other than full speed and customer-funded safety are feasible multistart candidates within their policy classes, not certified global optima.
\end{minipage}\end{table}

The continuation decides who pays, as Corollary~\ref{cor:threshold} predicts (Table~\ref{tab:known}). Under full redeployment, customers pay: the best constant rule is customer-funded safety itself. Delay would forgo what a finished frontier serving \alRedeployedShare{} percent of compute adds over the slowed one, so funding from research costs more per unit of hazard reduced.

Under no redeployment, research pays. The best constant rule keeps customers at 45 percent and moves research to \alSatConstX{} percent and safety to \alSatConstS{} percent. That cuts the probability of catastrophe by \alSatConstGain{} points to \alSatConstRisk{} percent. Customers could not use more compute, and a finished frontier would serve them no more than the slowed one does, so delay costs only quality. Panel (c) of Figure~\ref{fig:source} shows the best constant rule switching between the two answers as serving capacity passes the threshold.

Dynamic rules use both sources. Under no redeployment, the best two-stage rule switches at $F_c=\alSatStageSwitch$. Before the switch, customers give up most of the extra safety compute and research the rest; after it, safety falls to \alSatStageLateS{} percent. The twelve-segment rule cuts research most at the start of the takeoff period and not at all from $F=\alSatTwelveFullFrom$, with a largest milestone delay of \alSatTwelveDelay{} months. Under full redeployment, the best two-stage rule is customer-funded (Appendix~\ref{app:numerics} reports every switch milestone). Across serving capacity (Appendix~\ref{app:numerics}), the best two-stage rule beats the best constant rule by \alCapStageGainMin{} to \alCapStageGainMax{} points. From a serving capacity of \alCapNoCut{} percent it never cuts research. Panel C of Table~\ref{tab:known} reports the reference reset, whose gain Appendix~\ref{app:reset} traces to the continuation.

\subsection{Larger budgets: capability research pays too}\label{sec:larger}

The threshold of Corollary~\ref{cor:threshold} rises steeply with the budget, as Proposition~\ref{prop:budget} predicts (Figure~\ref{fig:budget}, panel (a)). It is \alThreshold{} percent at an \alBudgetPct{} percent budget and \alThresholdFifteen{} percent at \alBudgetMid{} percent. The customer share at the customer-funded optimum falls from \alBoundaryIFive{} percent at an \alBudgetPct{} percent budget to \alBoundaryIThirty{} percent at \alBudgetHigh{} percent. As it falls, the customer term of \eqref{eq:budgetthreshold} grows. The threshold reaches the full-redeployment share of 95 percent at a budget of \alCrossBudget{} percent, when customers keep \alCrossShare{} percent of compute. Above that budget, research is the cheaper marginal source even when a finished frontier gives customers all freed compute. The rise does not rest on a hazard advantage. At a tenth of baseline productivity, slowing research cuts less hazard than customer funding ($\rho$ between \alWeakRhoLow{} and \alWeakRhoHigh{}). Yet the threshold still rises, from \alWeakThreshFive{} percent at an \alBudgetPct{} percent budget to \alWeakThreshThirty{} percent at \alBudgetHigh{} percent. A threshold above 100 percent means research is the cheaper source whatever customers can use, since the customer share cannot exceed 100 percent.

\begin{figure}[htbp]\centering
\begin{minipage}{.325\textwidth}\begin{tikzpicture}\begin{axis}[width=\linewidth,height=5.4cm,xlabel style={align=center},tick label style={font=\scriptsize},label style={font=\footnotesize},title style={font=\small},grid=major,grid style={gray!15},xlabel={Service budget\\(\% of $\bar J$)},ylabel={Threshold share (\%)},title={(a) Threshold by budget},xmin=0,xmax=70,ymin=0,ymax=160]
\addplot[thick,black,no marks] coordinates {(2.30306,53.93917) (4.60612,56.11353) (6.90918,58.41961) (9.21224,60.88090) (11.51530,63.52005) (13.81835,66.36063) (16.12141,69.42815) (18.42447,72.75092) (20.72753,76.36083) (23.03059,80.29421) (25.33365,84.59273) (27.63671,89.30457) (29.93977,94.48576) (32.24283,100.20185) (34.54589,106.53003) (36.84895,113.56166) (39.15200,121.40568) (41.45506,130.19275) (43.75812,140.08079) (46.06118,151.26205) (48.36424,163.97256) (50.66730,178.50467) (52.97036,195.22396) (55.27342,214.59248) (57.57648,237.20099) (59.87954,263.81451) (62.18260,295.43788) (64.48566,333.41203) (66.78871,379.55840) (69.09177,436.40153)};
\addplot[blue!70!black,densely dotted,no marks] coordinates {(0.00000,95.00000) (70.00000,95.00000)};
\addplot[green!45!black,densely dotted,no marks] coordinates {(0.00000,45.00000) (70.00000,45.00000)};
\end{axis}\end{tikzpicture}\end{minipage}\hfill
\begin{minipage}{.325\textwidth}\begin{tikzpicture}\begin{axis}[width=\linewidth,height=5.4cm,xlabel style={align=center},tick label style={font=\scriptsize},label style={font=\footnotesize},title style={font=\small},grid=major,grid style={gray!15},xlabel={Service budget\\(\% of $\bar J$)},ylabel={Research share (\%)},title={(b) Best constant rule},xmin=0,xmax=70,ymin=15,ymax=62,legend style={draw=none,fill=none,font=\scriptsize,at={(0.97,0.97)},anchor=north east}]
\addplot[thick,blue!70!black,dashed,mark=*,mark size=0.9pt] coordinates {(5.75765,50.00000) (11.51530,50.00000) (17.27294,50.00000) (23.03059,50.00000) (28.78824,50.00000) (34.54589,48.32854) (46.06118,44.32628) (57.57648,40.80544) (69.09177,37.69054)};
\addlegendentry{Full redeployment}
\addplot[thick,green!45!black,dashdotted,mark=*,mark size=0.9pt] coordinates {(5.75765,45.02264) (11.51530,40.74771) (17.27294,37.04945) (23.03059,33.82893) (28.78824,31.00750) (34.54589,28.52207) (46.06118,24.99197) (57.57648,22.32533) (69.09177,20.07832)};
\addlegendentry{No redeployment}
\end{axis}\end{tikzpicture}\end{minipage}\hfill
\begin{minipage}{.325\textwidth}\begin{tikzpicture}\begin{axis}[width=\linewidth,height=5.4cm,xlabel style={align=center},tick label style={font=\scriptsize},label style={font=\footnotesize},title style={font=\small},grid=major,grid style={gray!15},xmode=log,xlabel={Safety productivity\\(multiple of baseline)},ylabel={Threshold share (\%)},title={(c) By productivity},xmin=0.05,xmax=2,ymin=0,ymax=100,xtick={0.1,0.25,0.5,1,2},xticklabels={0.1,0.25,0.5,1,2}]
\addplot[thick,black,no marks] coordinates {(0.05000,17.54721) (0.07500,30.40520) (0.10000,38.15189) (0.12500,43.22199) (0.15000,46.77823) (0.17500,49.40492) (0.20000,51.42243) (0.22500,53.01988) (0.25000,54.31578) (0.27500,55.38801) (0.30000,56.28984) (0.32500,57.05889) (0.35000,57.72249) (0.37500,58.30097) (0.40000,58.80973) (0.42500,59.26071) (0.45000,59.66323) (0.47500,60.02475) (0.50000,60.35125) (0.60000,61.39271) (0.70000,62.14449) (0.80000,62.71353) (0.90000,63.15993) (1.00000,63.52005) (1.10000,63.81720) (1.20000,64.06698) (1.30000,64.28020) (1.40000,64.46463) (1.50000,64.62596) (1.60000,64.76843) (1.70000,64.89529) (1.80000,65.00905) (1.90000,65.11165) (2.00000,65.20463)};
\addplot[green!45!black,densely dotted,no marks] coordinates {(0.05000,45.00000) (2.00000,45.00000)};
\end{axis}\end{tikzpicture}\end{minipage}
\caption{The funding threshold, the service budget, and safety productivity. (a) The funding threshold of Corollary~\ref{cor:threshold}, evaluated at the customer-funded optimum for each budget; below it, research is the cheaper marginal source, and a threshold above 100\% means research is cheaper whatever customers can use. Dotted lines mark the continuation shares of full redeployment (95\%) and no redeployment (45\%). (b) Research share of the best constant rule by budget. (c) The threshold at the \alBudgetPct\% customer-funded optimum when safety productivity is scaled with the hazard functions held fixed. All rules meet the 12-month deadline.}\label{fig:budget}
\end{figure}
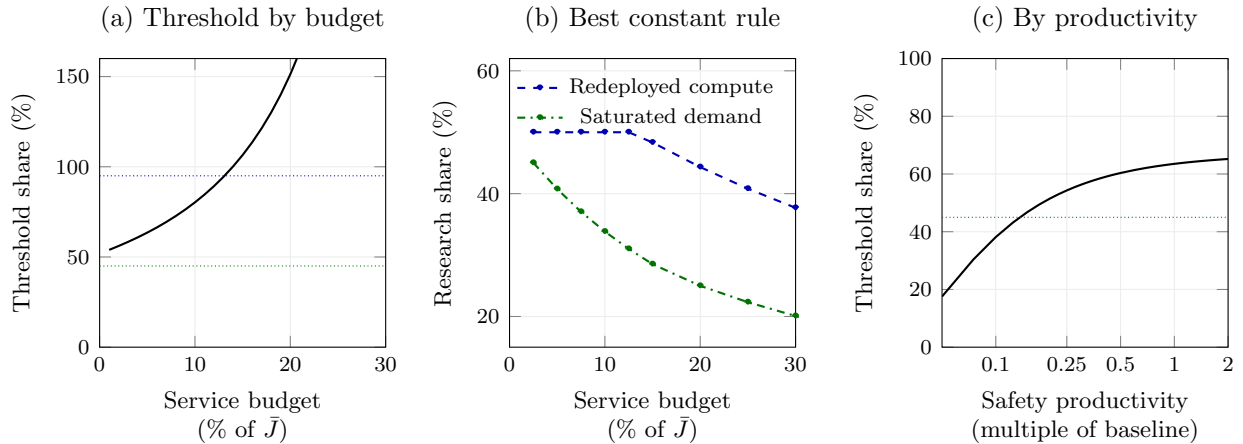

Under full redeployment, the best constant rule keeps research at 50 percent up to a budget of \alRedLastNoCut{} percent (panel (b)). From \alRedFirstCut{} percent on it cuts research, to \alRedXFifteen{} percent at a \alBudgetMid{} percent budget and \alRedXThirty{} percent at \alBudgetHigh{} percent. Under no redeployment it cuts research at every budget, from \alSatXTwoHalf{} percent at a \alBudgetLow{} percent budget to \alSatXThirty{} percent at \alBudgetHigh{} percent.

Richer timing enlarges the gain from cutting research and moves the cut to the start of the takeoff period (Table~\ref{tab:classes}, panels B and C). At a \alBudgetMid{} percent budget under full redeployment, the best twelve-segment rule cuts research most at the start and restores it to 50 percent from $F=\alRedTwelveBackFifteen$. Its largest milestone delay is \alRedTwelveDelayFifteen{} months. The rule lowers the probability of catastrophe \alRedTwelveAddFifteen{} points below the best customer-funded twelve-segment rule, against \alRedConstAddFifteen{} points among constant rules. At \alBudgetHigh{} percent the gain is \alRedTwelveAddThirty{} points; under no redeployment it is \alSatTwelveAddFifteen{} and \alSatTwelveAddThirty{} points at the two budgets. How deep the initial cut goes depends on how finely the rule can vary at the start; the gains barely depend on it (Appendix~\ref{app:numerics}). An early cut delays every later capability, giving knowledge more time before all of them, while it prolongs exposure only at levels where state risk is still low. Customer-funded safety still provides most of the reduction from full speed, but no longer nearly all of it (Table~\ref{tab:classes}, panel B).

Lower safety productivity moves the threshold the other way (panel (c)). With the hazard functions held fixed, the threshold at the \alBudgetPct{} percent customer-funded optimum is \alThresholdYieldQuarter{} percent at a quarter of baseline productivity and \alThresholdYieldTenth{} percent at a tenth. It falls below the no-redeployment share of 45 percent at \alYieldSatCross{} percent of baseline productivity. At a tenth of baseline productivity, the best constant rule under no redeployment keeps research at \alYieldTenthSatX{} percent: customers pay even though they cannot use more compute. Higher productivity raises the threshold only slightly, to \alThresholdYieldDouble{} percent at twice baseline productivity.

\subsection{What moves the threshold}\label{sec:drivers}

Three quantities move the threshold: the delay a point of research costs, how much extra knowledge safety compute buys, and how fast the value of customer services diminishes (Figure~\ref{fig:drivers}). Each panel of the figure varies one parameter and refits the calibration targets, as in the 36-month check of Appendix~\ref{app:takeoff36}. Start with the delay. With research at $x_0$, taking one unit of share from research delays completion by $\ell(1)=(1-e^{-\gamma T_0})/(Gx_0)$ months. The elasticity $\varepsilon$ and the share $x_0$ set how much a unit of share slows the pace. A point costs \alTdDelayBase{} months at baseline, \alTdDelayXHigh{} with a 65 percent research share and \alTdDelayEpsLow{} with $\varepsilon=0.3$. The cheaper delay raises the threshold to \alTdXHigh{} and \alTdEpsLow{} percent, above the full-redeployment share, so research should pay whatever customers can use.

Costlier delay reverses these comparisons. With $\varepsilon=1$ the threshold is \alTdEpsHigh{} percent, and with a 35 percent research share it is \alTdXLow{} percent against a no-redeployment share of 60 percent. In both, customers should pay even when they cannot use more compute.

The length of the takeoff period barely moves the threshold. The delay grows with the length of the takeoff period, damped by compute growth: a point costs \alTdDelayTNine{} months with a 9-month takeoff period and \alTdDelayTThirtySix{} with a 36-month one. But the budget is a share of the services over the takeoff period, so a longer period also raises what customer funding gives up, and every length buys the same customer-funded safety share. The threshold is \alTdTNine{} percent at 9 months, \alTdBase{} at 18 and \alTdTThirtySix{} at 36.

The second quantity is how much extra knowledge safety compute buys. When most safety learning needs no dedicated compute, customer-funded safety adds little knowledge, while slowing research still gives all learning more time. The threshold is \alTdPsiMid{} percent when three quarters of learning needs no compute and \alTdPsiHigh{} percent at nine tenths. Weak returns to safety compute and a larger reference safety share work the same way. With returns exponent $\alpha=0.2$ the threshold is \alTdAlphaLow{} percent, at the full-redeployment share. With an 18 percent safety share it is \alTdSafetyHigh{} percent, above that variant's full-redeployment share of 82 percent. 

The third quantity is the curvature of service value: faster diminishing returns make customer compute more costly to remove and favor research funding (Table~\ref{tab:curvature}). Replace $\sqrt Q$ by $Q^\eta$ and re-solve, holding the technology and the hazard functions fixed. With $\eta=1/4$, the threshold is \alCurvThrQuarter{} percent, and the best constant rule cuts research to \alCurvRedXQuarter{} percent even under full redeployment. With $\eta=3/4$, the threshold is \alCurvThrThreeQuarter{} percent, and the best constant rule keeps research at \alCurvSatXThreeQuarter{} percent even under no redeployment. So serving capacity decides who pays given how fast the value of customer services diminishes. That curvature is a choice about how services are valued, not an estimate; the table reports what each choice implies. Proposition~\ref{prop:frontload} assumes square-root safety and service technologies, and with other curvatures its closed form changes.

\begin{table}[htbp]\centering\small
\caption{The funding answer and the curvature of service value}\label{tab:curvature}
\begin{tabular}{@{}ccccccc@{}}\toprule
& Funding & \multicolumn{2}{c}{Full redeployment} & \multicolumn{2}{c}{No redeployment} \\ \cmidrule(lr){3-4}\cmidrule(lr){5-6}
Service value & threshold (\%) & Research (\%) & Risk (\%) & Research (\%) & Risk (\%) \\ \midrule
$Q^{0.25}$ & 239.9 & 37.9 & 45.17 & 29.4 & 35.48 \\
$Q^{0.5}$ & 63.5 & 50.0 & 57.68 & 40.7 & 54.83 \\
$Q^{0.75}$ & 42.9 & 50.0 & 62.12 & 50.0 & 62.12 \\
$Q^{1}$ & 36.5 & 50.0 & 64.74 & 50.0 & 64.74 \\
\bottomrule\end{tabular}
\par\smallskip\begin{minipage}{\textwidth}\footnotesize
\emph{Notes:} Service value $Q^\eta$ replaces $\sqrt Q$; the capability technology and the hazard functions are unchanged. Each row uses a service budget of \alBudgetPct{} percent of its own budget unit, the reference's discounted services over its takeoff period. The threshold is evaluated at the customer-funded constant rule (Corollary~\ref{cor:threshold}). The research and risk columns report the best constant rule, re-solved in each row; research at 50 percent means customers fund all extra safety. The row $Q^{0.5}$ is the paper's baseline.
\end{minipage}\end{table}

Under marginal pricing, the token-demand estimates imply rapidly diminishing marginal service value. Short-run price elasticities of demand for language-model tokens are just above one \citep{demirer2025}. With service value $Q^\eta$ priced at its margin, the elasticity is $1/(1-\eta)$, so an elasticity near one means $\eta$ near zero, and service value $Q^{0.1}$ puts the threshold at \alTdEtaTenth{} percent. These are short-run estimates, and longer-run responses may be larger.

Other parameters move the threshold less. A research speedup of 3 rather than 10 raises the threshold to \alTdSpeedLow{} percent, and one of 30 lowers it to \alTdSpeedHigh{} percent. The pause target, the state share and timing of hazard, compute growth, and the length of the takeoff period keep it between \alTdFlatMin{} and \alTdFlatMax{} percent, and the full-speed probability does not move it (Section~\ref{sec:robust}).

\begin{figure}[htbp]\centering
\begin{minipage}{.245\textwidth}\centering\begin{tikzpicture}\begin{axis}[scale only axis,width=.62\linewidth,height=2.9cm,tick label style={font=\scriptsize},xticklabel style={/pgf/number format/fixed,/pgf/number format/precision=2},label style={font=\footnotesize},title style={font=\small},grid=major,grid style={gray!15},ymin=30,ymax=130,ytick={40,60,80,100,120},clip=true,ylabel={Threshold (\%)},title={Curvature $\eta$},xtick={0.25,0.5,0.75}]
\addplot[thick,black,mark=*,mark size=1.1pt] coordinates {(0.25,239.947) (0.5,63.520) (0.75,42.920)};
\addplot[only marks,black,mark=*,mark size=2pt] coordinates {(0.5,63.520)};
\addplot[blue!70!black,densely dotted,no marks] coordinates {(0.25,95.000) (0.5,95.000) (0.75,95.000)};
\addplot[green!45!black,densely dotted,no marks] coordinates {(0.25,45.000) (0.5,45.000) (0.75,45.000)};
\end{axis}\end{tikzpicture}\end{minipage}\hfill
\begin{minipage}{.245\textwidth}\centering\begin{tikzpicture}\begin{axis}[scale only axis,width=.62\linewidth,height=2.9cm,tick label style={font=\scriptsize},xticklabel style={/pgf/number format/fixed,/pgf/number format/precision=2},label style={font=\footnotesize},title style={font=\small},grid=major,grid style={gray!15},ymin=30,ymax=130,ytick={40,60,80,100,120},clip=true,yticklabels={},title={Takeoff period $T_0$},xtick={9,18,36}]
\addplot[thick,black,mark=*,mark size=1.1pt] coordinates {(9,60.607) (12,62.646) (18,63.520) (24,62.928) (36,60.890)};
\addplot[only marks,black,mark=*,mark size=2pt] coordinates {(18,63.520)};
\addplot[blue!70!black,densely dotted,no marks] coordinates {(9,95.000) (12,95.000) (18,95.000) (24,95.000) (36,95.000)};
\addplot[green!45!black,densely dotted,no marks] coordinates {(9,45.000) (12,45.000) (18,45.000) (24,45.000) (36,45.000)};
\end{axis}\end{tikzpicture}\end{minipage}\hfill
\begin{minipage}{.245\textwidth}\centering\begin{tikzpicture}\begin{axis}[scale only axis,width=.62\linewidth,height=2.9cm,tick label style={font=\scriptsize},xticklabel style={/pgf/number format/fixed,/pgf/number format/precision=2},label style={font=\footnotesize},title style={font=\small},grid=major,grid style={gray!15},ymin=30,ymax=130,ytick={40,60,80,100,120},clip=true,yticklabels={},title={Elasticity $\varepsilon$},xtick={0.3,0.5,0.8,1}]
\addplot[thick,black,mark=*,mark size=1.1pt] coordinates {(0.3,98.198) (0.4,75.796) (0.5,63.520) (0.65,52.884) (0.8,46.517) (1,41.121)};
\addplot[only marks,black,mark=*,mark size=2pt] coordinates {(0.5,63.520)};
\addplot[blue!70!black,densely dotted,no marks] coordinates {(0.3,95.000) (0.4,95.000) (0.5,95.000) (0.65,95.000) (0.8,95.000) (1,95.000)};
\addplot[green!45!black,densely dotted,no marks] coordinates {(0.3,45.000) (0.4,45.000) (0.5,45.000) (0.65,45.000) (0.8,45.000) (1,45.000)};
\end{axis}\end{tikzpicture}\end{minipage}\hfill
\begin{minipage}{.245\textwidth}\centering\begin{tikzpicture}\begin{axis}[scale only axis,width=.62\linewidth,height=2.9cm,tick label style={font=\scriptsize},xticklabel style={/pgf/number format/fixed,/pgf/number format/precision=2},label style={font=\footnotesize},title style={font=\small},grid=major,grid style={gray!15},ymin=30,ymax=130,ytick={40,60,80,100,120},clip=true,yticklabels={},title={Research share $x_0$},xtick={0.35,0.5,0.65}]
\addplot[thick,black,mark=*,mark size=1.1pt] coordinates {(0.35,46.567) (0.5,63.520) (0.65,97.377)};
\addplot[only marks,black,mark=*,mark size=2pt] coordinates {(0.5,63.520)};
\addplot[blue!70!black,densely dotted,no marks] coordinates {(0.35,95.000) (0.5,95.000) (0.65,95.000)};
\addplot[green!45!black,densely dotted,no marks] coordinates {(0.35,60.000) (0.5,45.000) (0.65,30.000)};
\end{axis}\end{tikzpicture}\end{minipage}\par\smallskip
\begin{minipage}{.245\textwidth}\centering\begin{tikzpicture}\begin{axis}[scale only axis,width=.62\linewidth,height=2.9cm,tick label style={font=\scriptsize},xticklabel style={/pgf/number format/fixed,/pgf/number format/precision=2},label style={font=\footnotesize},title style={font=\small},grid=major,grid style={gray!15},ymin=30,ymax=130,ytick={40,60,80,100,120},clip=true,ylabel={Threshold (\%)},title={Learning $\psi$},xtick={0.25,0.5,0.9}]
\addplot[thick,black,mark=*,mark size=1.1pt] coordinates {(0.25,59.859) (0.5,63.520) (0.75,75.164) (0.9,115.983)};
\addplot[only marks,black,mark=*,mark size=2pt] coordinates {(0.5,63.520)};
\addplot[blue!70!black,densely dotted,no marks] coordinates {(0.25,95.000) (0.5,95.000) (0.75,95.000) (0.9,95.000)};
\addplot[green!45!black,densely dotted,no marks] coordinates {(0.25,45.000) (0.5,45.000) (0.75,45.000) (0.9,45.000)};
\end{axis}\end{tikzpicture}\end{minipage}\hfill
\begin{minipage}{.245\textwidth}\centering\begin{tikzpicture}\begin{axis}[scale only axis,width=.62\linewidth,height=2.9cm,tick label style={font=\scriptsize},xticklabel style={/pgf/number format/fixed,/pgf/number format/precision=2},label style={font=\footnotesize},title style={font=\small},grid=major,grid style={gray!15},ymin=30,ymax=130,ytick={40,60,80,100,120},clip=true,yticklabels={},title={Safety returns $\alpha$},xtick={0.2,0.5,0.7}]
\addplot[thick,black,mark=*,mark size=1.1pt] coordinates {(0.2,94.959) (0.5,63.520) (0.7,58.414)};
\addplot[only marks,black,mark=*,mark size=2pt] coordinates {(0.5,63.520)};
\addplot[blue!70!black,densely dotted,no marks] coordinates {(0.2,95.000) (0.5,95.000) (0.7,95.000)};
\addplot[green!45!black,densely dotted,no marks] coordinates {(0.2,45.000) (0.5,45.000) (0.7,45.000)};
\end{axis}\end{tikzpicture}\end{minipage}\hfill
\begin{minipage}{.245\textwidth}\centering\begin{tikzpicture}\begin{axis}[scale only axis,width=.62\linewidth,height=2.9cm,tick label style={font=\scriptsize},xticklabel style={/pgf/number format/fixed,/pgf/number format/precision=2},label style={font=\footnotesize},title style={font=\small},grid=major,grid style={gray!15},ymin=30,ymax=130,ytick={40,60,80,100,120},clip=true,yticklabels={},title={Safety share $s_0$},xtick={0.02,0.1,0.18}]
\addplot[thick,black,mark=*,mark size=1.1pt] coordinates {(0.02,58.981) (0.05,63.520) (0.1,72.129) (0.18,91.118)};
\addplot[only marks,black,mark=*,mark size=2pt] coordinates {(0.05,63.520)};
\addplot[blue!70!black,densely dotted,no marks] coordinates {(0.02,98.000) (0.05,95.000) (0.1,90.000) (0.18,82.000)};
\addplot[green!45!black,densely dotted,no marks] coordinates {(0.02,48.000) (0.05,45.000) (0.1,40.000) (0.18,32.000)};
\end{axis}\end{tikzpicture}\end{minipage}\hfill
\begin{minipage}{.245\textwidth}\centering\begin{tikzpicture}\begin{axis}[scale only axis,width=.62\linewidth,height=2.9cm,tick label style={font=\scriptsize},xticklabel style={/pgf/number format/fixed,/pgf/number format/precision=2},label style={font=\footnotesize},title style={font=\small},grid=major,grid style={gray!15},ymin=30,ymax=130,ytick={40,60,80,100,120},clip=true,yticklabels={},title={Research speedup},xtick={3,10,30},xmode=log,log ticks with fixed point]
\addplot[thick,black,mark=*,mark size=1.1pt] coordinates {(3,78.599) (5,71.580) (10,63.520) (20,57.053) (30,53.926)};
\addplot[only marks,black,mark=*,mark size=2pt] coordinates {(10,63.520)};
\addplot[blue!70!black,densely dotted,no marks] coordinates {(3,95.000) (5,95.000) (10,95.000) (20,95.000) (30,95.000)};
\addplot[green!45!black,densely dotted,no marks] coordinates {(3,45.000) (5,45.000) (10,45.000) (20,45.000) (30,45.000)};
\end{axis}\end{tikzpicture}\end{minipage}
\caption{What moves the funding threshold. Each panel varies one parameter, refits the calibration targets, and plots the threshold of Corollary~\ref{cor:threshold} at the \alBudgetPct{} percent customer-funded optimum; the large dot is the baseline. The upper dotted line is the full-redeployment share $1-s_0$: above it, research is the cheaper source whatever customers can use. The lower dotted line is the no-redeployment share $i_0$: below it, customers are cheaper even when they cannot use more compute. Service value $Q^{0.25}$ gives \alTdEtaQuarter{} percent, off the scale, and $Q^{0.1}$ gives \alTdEtaTenth{} percent (not shown). The full-speed probability, the pause target, the state share and shape of hazard, compute growth, and the takeoff length keep the threshold between \alTdFlatMin{} and \alTdFlatMax{} percent.}\label{fig:drivers}
\end{figure}

The answer flips inside the ranges the evidence allows. At an \alBudgetPct{} percent budget, research should pay whatever customers can use when the research share exceeds \alTdCrossXRedPct{} percent or when $\varepsilon$ is below \alTdCrossEpsRed{}. Customers should pay even under no redeployment when the research share is below \alTdCrossXSatPct{} percent or $\varepsilon$ exceeds \alTdCrossEpsSat{}. Estimated shares lie on both sides of the first flip point. OpenAI put about three quarters of its 2024 compute spending into research and development \citep{you2025}, Anthropic about 60 percent in 2025 \citep{emberson2026}, and OpenAI about half in 2025 \citep{epoch2026chips}. Anthropic's registration statement gives a bound from its financial statements. Compute for paying customers was at most 38 percent of its compute spending in 2025 and 27 percent in the first half of 2026, so at least 62 and 73 percent went to research, model development, internal use, and free users \citep{anthropic2026s1}. These shares count dollars and include some safety work. They also move fast, while the model holds reference shares fixed. The 2026 bound lies above the first flip point, so on those accounts the rule points to research. Extra safety compute should then come from slowing capability work rather than from customers who are already rationed. The 2025 bound lay just below it, where strong demand made customers the cheaper source. The bound counts free users and internal use with research, so a true research share below the flip point cannot be ruled out. Forecasts put the median takeoff period between 9 and 18 months \citep{lifland2026}; no takeoff length from 9 to 36 months flips the answer.

Where slowing research is cheap, rules that also cut research gain more. Under no redeployment, the best twelve-segment rule that may cut research beats the best customer-funded one by \alTdGainBase{} points at baseline. The gain is \alTdGainXHigh{} points with a 65 percent research share and \alTdGainEtaQuarter{} points with service value $Q^{1/4}$.

\subsection{Robustness}\label{sec:robust}

The tested changes in gradual redeployment and hazard construction preserve the funding answer. Let customers' use beyond the reference share expand over 6 to 24 months after completion. The budget at which research becomes the cheaper source then rises from \alCrossBudget{} to between \alExpandCrossSix{} and \alExpandCrossTwentyFour{} percent of $\bar J$. Across constructions of the hazard functions, the \alBudgetPct{} percent threshold stays between \alHzThrMin{} and \alHzThrMax{} percent. The reference probability moves risk levels only. At 25 percent the threshold stays at \alHrThr{} percent, the funding answer is the same, and the twelve-segment customer-funded risk is \alHrCFTwelveLo{} percent, against \alHrCFTwelveHi{} at 75 percent. Below the range, at 10 and 5 percent, the threshold is the same and that risk is \alHrCFTwelveTen{} and \alHrCFTwelveFive{} percent. Appendix~\ref{app:robust} gives both. A slower, 36-month takeoff period keeps the funding answer (Appendix~\ref{app:takeoff36}).

\section{Commitment before the technology is known}\label{sec:uncertainty}

The regulator does not know how safety research works when it commits. The \alStateCount{} technologies considered vary safety productivity and the transfer of knowledge to stronger models. They also vary whether safety compute scales with installed compute, whether it also advances capability, and whether danger continues after completion. The physical hazard functions stay fixed. The regulator commits to a constant rule that meets the service budget and the deadline in every technology. It picks the one that minimizes the largest excess probability of catastrophe over an informed choice. Appendix~\ref{app:states} states the problem and reports the rules.

A fixed rule costs little in the ten technologies where safety compute does not advance capability. In every continuation and budget considered, the chosen rule is within \alMaxNoSpill{} points of the informed choice on the same share grid in each of them. At the \alBudgetPct{} percent budget the bound is \alMaxNoSpillFive{} points (Appendix~\ref{app:states}). Across \alTdCommitCount{} parameter variants from Section~\ref{sec:drivers}, it is at most \alTdCommitWorst{} points, reached with a 9-month takeoff period. The regret from spillover is \alTdCommitSpillMin{} to \alTdCommitSpillMax{} points in every variant.

Without spillover, the affordable set is known in advance because a rule's capability path, service cost, and delay do not depend on the safety technology. Within that set, the ranking of rules by catastrophe probability barely moves with productivity, transfer, compute scaling, or continuing danger in the tested technologies. The first result is structural; the second is a numerical finding for the technologies without spillover among the \alStateCount{} considered, using constant rules.

The grid limits the interpretation of the bound. A feasible rule found off the grid can beat the grid's informed choice by more than the bound, so the bound measures the cost of not knowing the technology among grid rules rather than against an unrestricted optimum. The states hold fixed the hazard profile, which Appendix~\ref{app:hazard} varies for the known technology, as well as the returns to safety compute, the customer valuation, and the continuation. Their productivity also stays above the level at which the funding test of Section~\ref{sec:larger} changes sign under no redeployment. The deadline never excludes an affordable rule in the two consistent continuations, at any budget; the service budget is the binding commitment.

The remaining regret comes from spillover. When safety compute also advances capability, a rule delays less and costs fewer services, and an informed regulator would spend the slack on more safety. Under full redeployment, knowing whether spillover is present would cut the \alBudgetPct{} percent worst regret from \alRedRegret{} points to zero. Knowing whether safety productivity is weak would leave it at \alRedYieldSignal{}. In this exercise, free and perfect classification of spillover is worth more than classification of the productivity states, because spillover changes what a rule costs. The exercise has no acquisition cost, noisy signal, or dynamic rule, so it does not value measuring safety productivity in general.

The fixed rules follow the funding-source condition. Under full redeployment, the \alBudgetPct{} percent rule is customer-funded. Under no redeployment, it cuts research to \alSatRuleX{} percent, and its worst regret is \alSatRegret{} points, against \alSatBoundaryRegret{} for the best customer-funded rule. At larger budgets the fixed rules cut research in both continuations, as Section~\ref{sec:larger} predicts: under full redeployment the \alBudgetHigh{} percent rule cuts research to \alRedRuleXThirty{} percent. Robustness within this set does not mean low risk. Under full redeployment, the \alBudgetPct{} percent rule leaves the probability of catastrophe at \alRedTailRisk{} percent when danger continues after completion. In the combined weak-productivity, transfer, and spillover state it leaves \alRedAdverseRisk{} percent.

\section{Implementation and conditional predictions}\label{sec:implementation}

A joint allocation can be implemented as compute shares that change at an observable capability milestone. Enforcement needs a meaningful distinction between inference, safety research, and capability research. The spillover states show what happens when safety workloads also advance capability even though their allocation is observed correctly.

The White House Accord on Super Intelligence of September 2026 shows the gap. Six AI and chip companies committed to monitor the capabilities and alignment of their models. Internal teams make sure the controls work, independent auditors assess them, and a board committee oversees the results \citep{whitehouse2026}. The accord sets no compute for safety research and does not say whether customers or capability research should give it up. Auditing a rule from this paper would mean checking how compute is split among the three uses and whether safety work also advances capability.

Labs already pay for some safety by slowing research, and they do not yet keep the accounts an audit would need. Anthropic's registration statement says it has delayed releases and limited the initial availability of models so that safety measures keep pace with capabilities. It runs compute as one pool that it plans and assesses ``on a total basis rather than by individual workload'' \citep{anthropic2026s1}. Its estimate that about 6 percent of research compute went to safety relies on best-effort labels for one week \citep{anthropic2026measuring}. The registration statement also says that demand for compute has historically exceeded available capacity and that the lab is ``first and foremost a research lab'' whose strategy prioritizes frontier research \citep{anthropic2026s1}. Together these suggest that it sells less inference than customers would buy to keep compute in research, so it values a unit of research compute above the revenue that unit would earn serving customers. That valuation includes staying ahead of rivals. A regulator who also counts catastrophe risk would fund extra safety by trimming that research share, not by rationing customers further.

The model makes three conditional predictions. First, who pays for safety compute should track how much freed compute can serve customers when the frontier finishes and the size of the service budget. This holds given how customer services are valued and given the safety technology. In the baseline square-root calibration with a small budget, a frontier whose freed compute can serve customers at once should fund extra safety from customers. One whose freed compute cannot serve customers should also slow research (Corollary~\ref{cor:threshold}). Either comparison can reverse with the valuation of services, the share of compute in research, how fast the research pace responds to compute, or the productivity of safety research (Sections~\ref{sec:larger} and~\ref{sec:drivers}). Large enough service budgets should slow research in both, mostly early in the takeoff period. In the calibration, danger that continues after the takeoff period barely changes this.

Second, where both margins are interior, safety compute relative to inference should fall over a takeoff period. The exception is safety productivity that rises fast enough to satisfy condition \eqref{eq:risethreshold}.

Third, a rule that frees inference compute is worth what customers can use; idle capacity adds nothing to its service value.

Testing these predictions requires four measures. They are research compute, inference use and idle capacity, capability progress, and the rate at which particular safety work reduces assessed hazard. Knowledge here is measured in log-hazard reduction, not in publications or evaluation tasks.

The analysis is limited in six respects. First, the numbers come from one conditional scenario whose hazard functions are assumed. Second, services are an index, not revenue, so the model says nothing about who bears the monetary cost. Third, one frontier lab is regulated. The deadline stands for a rival whose response is not modeled, and the lab has no private objective. Fourth, installed compute grows exogenously. A rule changes compute use by category, and under no redeployment some capacity can sit idle, but it does not change the installed stock. Modeling how regulation changes compute demand would need investment, time to build, customer revenue, and a private compliance objective. Fifth, the rules come from restricted classes and are not certified global optima. Sixth, the results do not rank instrument families: the regulator minimizes hazard under a service budget and a deadline and imposes no risk ceiling.

\section{Conclusion}

The value of finishing determines the cheaper marginal source of safety compute under the model's common continuation. Above the threshold, customers should fund extra safety; below it, capability research should also pay. Customer funding gives up services during the takeoff period, while research funding also delays quality and completion. Delay costs, additional safety learning, and service valuation determine where the local threshold sits, and varying the calibration shows that they can reverse the funding answer within the ranges considered. One frontier lab's 2026 accounts put at least 73 percent of its compute spending outside paying customers. On those accounts the rule points to capability research, so extra safety compute should come from slowing capability work rather than from customers who are already rationed.

Larger budgets shift funding toward capability research in the calibration because each further unit taken from customers costs more. Research eventually becomes the cheaper source if its hazard reduction stays above a positive fraction of the reduction from customer funding; the numerical rules then cut research mostly early in the takeoff period. Customer-funded safety should come early relative to customer use under the baseline timing conditions, unless safety productivity rises fast. In the baseline square-root calibration at an \alBudgetPct{} percent budget, the best customer-funded rule is within \alSatTwelveAdd{} points of the best joint rule in the tested policy classes.

Within the tested constant-share rules, uncertainty about safety productivity changes the preferred allocation little without capability spillovers. Spillovers matter more because they change affordability. These funding and commitment comparisons need a common continuation: a model that lets a finished frontier serve customers less than a slowed one turns the endpoint term into a saving and credits safety research with gains that belong to the continuation.

\FloatBarrier
\appendix
\numberwithin{proposition}{section}
\numberwithin{corollary}{section}
\numberwithin{lemma}{section}
\renewcommand{\theHproposition}{appendix.\thesection.\arabic{proposition}}
\renewcommand{\theHcorollary}{appendix.\thesection.\arabic{corollary}}
\renewcommand{\theHlemma}{appendix.\thesection.\arabic{lemma}}
\section{Scenario construction}\label{app:calibration}

With $\phi(F)=\phi_0\cdot10^F$, dividing \eqref{eq:capability} by $\phi(F)$ and integrating the reference path from zero gives $(1-10^{-F})/(\phi_0\log10)=x_0^\varepsilon(e^{\gamma t_0(F)}-1)/\gamma$. Completion at $T_0$ fixes the constant, and the dates follow:
\begin{align}
\phi_0x_0^\varepsilon&=\frac{0.9\,\gamma}{\log10\,(e^{\gamma T_0}-1)},\\
t_0(F)&=\frac1\gamma\log\left[1+(e^{\gamma T_0}-1)
                   \frac{1-10^{-F}}{0.9}\right].
\end{align}
The first line fits the constant to completion at $T_0$. In the second, $(1-10^{-F})/0.9=z(F)/z(1)$: the share of the reference's research capacity over the whole takeoff period that it has accumulated on reaching $F$. More than half of it is needed for the first half of capability, where the research model is weakest. The full-speed pace at capability $F$ is $u_0(F)=\phi_0x_0^\varepsilon\cdot10^Fe^{\gamma t_0(F)}$. Put $H_0=-\log(1-0.75)$, $q=0.75$, $c=2\log[q/(1-q)]$, and $a_0=H_0c/(e^c-1)$. Full-speed hazard density per unit of capability is $a_0e^{cF}$, so that a share $q$ of it falls in the second half of capability development. With state share $\omega=1/2$, the unprotected primitives are
\begin{equation}
h(F)=\omega a_0e^{cF}u_0(F)e^{\sigma t_0(F)},\qquad
k(F)=(1-\omega)a_0e^{cF}e^{\sigma t_0(F)}.
\label{eq:primitives}
\end{equation}
The two primitives are measured in different units. The state part $h$ is a hazard per month, so it carries the full-speed pace $u_0$; the step part $k$ is a hazard per unit of capability. The factor $e^{\sigma t_0(F)}$ removes the protection that reference knowledge already provides. So the primitives reproduce the reference density when knowledge is $\sigma t_0(F)$. These functions are built once; later technology changes do not rebuild them.

The fitted productivity solves
\begin{equation}
\frac{h(0)}\sigma(1-e^{-12\sigma})+H_0e^{-12\sigma}=-\log(1-0.40).
\label{eq:sigmafit}
\end{equation}
The first term is state hazard during the stopped-clock pause; the second is hazard on the repeated path, protected by 12 additional months of knowledge. The fitted value is $\sigma=\alSigma$ per month. Matching the target does not validate the hazard functions, the state/step split, or the returns to safety compute.

\begin{table}[htbp]\centering\small
\caption{Judgments, assumptions, and fitted inputs}\label{tab:parameters}
\begin{tabular}{@{}p{4.6cm}p{10.8cm}@{}}\toprule
Input & Value and role\\\midrule
Source judgments & 18-month full-speed takeoff period; 75\% reference catastrophe probability, with 25\% as the lower case (Appendix~\ref{app:hazard}); 40\% under a safety-case regime, 10\% at 25\%; about a year of catch-up slack.\\
Mapping of the 40\% judgment & A 12-month initial pause with reference safety allocation and the compute clock stopped; a modeling assumption.\\
Installed compute & Growth $G=\log2/6$ per month; illustrative and exogenous.\\
Research technology & $\phi(F)=\phi_0\cdot10^F$; $\varepsilon=1/2$, so research capacity grows at $\gamma=\log2/12$; $\phi_0$ fitted to completion at month 18.\\
Reference allocation & $(x_0,s_0,i_0)=(0.50,0.05,0.45)$.\\
Safety technology & Learning multiplier \eqref{eq:multiplier} with $\psi=\alpha=1/2$; $\sigma=\alSigma$ per month, fitted using the stopped-clock pause; scaled by 0.05--2 in Section~\ref{sec:larger}, with the hazard functions fixed.\\
Hazard construction & State share $\omega=1/2$; 75\% of reference hazard in the second capability half; both varied in Appendix~\ref{app:hazard}, with the pause judgment and the reference probability.\\
Customer services & Square-root value; $\mu=\log10$; $r=0.03/12$ per month.\\
Continuation & Serving capacity $\bar\imath$, with $i_c=\min\{\bar\imath,1-s_0\}$ (Lemma~\ref{lem:continuation}): 1 for full redeployment (lead case), 0.45 for no redeployment, 0.45--1 in Figure~\ref{fig:source}. Continuing danger $\xi=0$, and up to 5 in Section~\ref{sec:funding}. Reference reset: $\bar\imath=1$, $i_c=0.45$. Gradual redeployment (Appendix~\ref{app:gradual}): use beyond 45\% expands over 6--24 months after completion, or capacity reaches 100\% at month 36--72.\\
Policy commitments & Maximum milestone delay $L=12$ months; service budgets $B=b\bar J$, with $\bar J$ the reference's services over its takeoff period and $b=\alBudgetPct\%$ in the main comparisons, 2.3--\alBudgetHigh\% in Section~\ref{sec:larger}, and $b\in\{\alBudgetPct,\alBudgetMid,\alBudgetHigh\}\%$ in Section~\ref{sec:uncertainty}; costs over 60 months.\\\bottomrule
\end{tabular}\end{table}

The growth rate describes an illustrative installed-compute path, not an estimate from observed laboratory investment. One frontier lab's spending grew at this rate. Anthropic's compute and infrastructure expense rose from \$1.39 billion in the second quarter of 2025 to \$5.56 billion a year later, a fourfold rise, or a doubling every six months \citep{anthropic2026s1}. Anthropic was among the fastest-growing labs over that year, and spending is not physical compute, so the rate is a fast case rather than an estimate. Evidence on the growth of training compute for milestone models is a different object \citep{sevilla2022}. At this rate the compute stock expands severalfold over the service horizon, which is why later service flows dominate discounted totals. The endpoint is also a modeling choice: danger ends in the baseline, and services continue at endpoint quality while compute grows. Neither assumption is a forecast of output or safety after the takeoff period.

\section{Proofs and extensions}\label{app:proof}

The proofs and supporting results follow the service budget and the two main questions: who pays, and when safety is supplied.

\subsection{The service budget}\label{app:budgetproof}

A budget stated as a percentage of services changes meaning with the horizon.

\begin{lemma}[Service cost after catch-up]\label{lem:budget}
Suppose the policy and the reference both complete by $T^\ast$ and thereafter have the same capability, compute stock, and customer share. Then $C_\tau=C_{T^\ast}$ for every $\tau\geq T^\ast$. If $J_{0,\tau}$ keeps increasing, a positive cost's percentage $C_\tau/J_{0,\tau}$ falls with $\tau$.
\end{lemma}

After catch-up, both paths deliver the same service flow, so extending the horizon adds nothing to their difference.

\begin{proof}
Differentiating $C_\tau$ in $\tau$ gives the reference flow minus the policy flow, which is zero after $T^\ast$.
\end{proof}

\subsection{Who should pay for safety compute}\label{app:fundingproofs}

\begin{proof}[Proof of Proposition~\ref{prop:funding}]
With constant $m$, the research clock \eqref{eq:clock} reads $m(e^{\gamma t_x(F)}-1)=z(F)$, which is the first part of \eqref{eq:constantpath}. By \eqref{eq:capability}, $u_x=\phi(F)x_0^\varepsilon me^{\gamma t_x}=\phi(F)x_0^\varepsilon[m+z(F)]$, and the reference pace is $u_0=\phi(F)x_0^\varepsilon[1+z(F)]$; their ratio gives the second part. Differentiating \eqref{eq:constantH}, the derivative of the protection factor in $x$ is $a\ell e^{-at_x}$, and the derivative of $h/u_x$ is $-(h/u_x)\vartheta$. The safety share enters only through $a(s)$. The two funding directions are $(dx,ds,di)=(0,ds,-ds)$ and $(-ds,ds,0)$.
\end{proof}

Two special cases follow from the two terms of \eqref{eq:Hx}. \emph{Step 1: no learning.} Set $a(s)=a'(s)=0$. The knowledge term vanishes, and $H_x$ is minus the exposure term, which is negative wherever state risk is positive. Zero marginal learning also gives $H_s=0$ in \eqref{eq:Hs}. So a marginal transfer from research to safety changes hazard by $H_s-H_x=-H_x>0$. \emph{Step 2: no state risk.} Set $h=0$. The exposure term vanishes, and $H_x$ is the knowledge term, which is nonnegative.

\begin{corollary}[A failure case]\label{cor:failure}
If $a(s)=a'(s)=0$ and state risk is positive on a set of positive measure, then $H_x<0$ and $H_s-H_x=-H_x>0$. A marginal transfer from research to safety increases hazard. If all hazard is step risk and $a(s)>0$, then $H_x\geq0$, strictly when learning changes exposure to positive later step risk.
\end{corollary}

\begin{proof}
For the first case, set $a=0$ in \eqref{eq:Hx} and $a'=0$ in \eqref{eq:Hs}. This gives $H_x<0$ and $H_s=0$, so $H_s-H_x>0$. For the second case, set $h=0$ in \eqref{eq:Hx}. The remaining integrand has the stated sign.
\end{proof}

The failure case concerns constant shares. A calendar-time pause instead builds a compute overhang, after which the frontier crosses dangerous levels faster, so an early pause can lower hazard even without learning \citep{gancommitments2026}. Without the overhang, the pause condition of \citet{binder2021} and \citet{bostrom2026}, written with the pause's own safety allocation in \citet[Proposition~4]{gancommitments2026}, applies.

\noindent\emph{Proof of Proposition~\ref{prop:source}.} 
Without spillover, the capability path does not depend on $i$, so customer funding changes only $J^{\mathrm{pre}}=\sqrt i\int_0^1e^{\kappa t_x+\mu F/2}u_x^{-1}\,dF$. Then $\partial J/\partial i=J^{\mathrm{pre}}/(2i)$ while $i<\bar\imath$, and the completion date does not move. For funding from research, write $J=\int_0^{T_x}b(t,F_x(t),i)\,dt+V_c(T_x)$, where $F_x(t)$ inverts $t_x(F)$. Holding $i$ fixed,
\[
\frac{\partial J}{\partial x}
=\int_0^{T_x}b_F\,\frac{\partial F_x(t)}{\partial x}\,dt
+\big[b_T-M(T_x)\big]\frac{\partial T_x}{\partial x}.
\]
Here $b_F=(\mu/2)b$, $\partial F_x(t)/\partial x=\ell(F)u_x(F)$ from $t_x(F_x(t))=t$, and $\partial T_x/\partial x=-\ell(1)$. Changing variables with $dt=dF/u_x$ gives $\partial J/\partial x=Q_x+\ell(1)[M(T_x)-b_T]$. So moving one unit of share from research to safety costs $c_R=\partial J/\partial x$, which is \eqref{eq:cR}. The same expression follows by differentiating the capability-domain form of $J$, using $d\ell/dF=\vartheta/u_x$ and $\ell(0)=0$.

Total hazard is $H(x,s)+\Phi e^{-S_T}$ with $S_T=a(s)T_x$. Customer funding raises $S_T$ by $a'(s)T_x$; funding from research raises it by $a'(s)T_x$ through the safety share and by $a(s)\ell(1)$ through the later completion date. With Proposition~\ref{prop:funding} these give \eqref{eq:RC} and \eqref{eq:RR}. Comparing hazard reduction per unit of service cost, research is the cheaper source when $R_R/c_R>R_C/c_C$. Because $c_R$ rises linearly in $M$ with slope $\ell(1)>0$, rearranging gives \eqref{eq:threshold}.

Proposition~\ref{prop:source} is the paper's funding condition. It needs two facts about the continuation: the value of finishing and the hazard that remains. The following result derives the first from the regulator's own problem, in three steps.

\emph{Step 1: hazard after completion is fixed.} After completion capability stays at one and safety stays at its reference share. So hazard no longer depends on how the remaining $1-s_0$ of compute is split between customers and research that no longer advances capability.

\emph{Step 2: services rise with the customer share.} The service flow after completion rises with the customer share up to serving capacity $\bar\imath$, and the share cannot exceed $1-s_0$. Holding the reference path fixed, a larger continuation share for the policy lowers its service cost, so it relaxes the budget without raising hazard.

\emph{Step 3: the value of finishing.} With the share $i_c$ after completion, the continuation is worth $\int_T^\tau e^{\kappa t+\mu/2}\sqrt{i_c}\,dt$. Finishing one month earlier adds one month of that flow at its start, so the value of finishing is the finished frontier's service flow at $T$.

\begin{lemma}[An optimal continuation]\label{lem:continuation}
Suppose that after completion capability stays at one, safety stays at the reference share, and customers use at most $\bar\imath$. Hold the reference path, and so $J_{0,\tau}$, fixed. Then for any policy, an optimal continuation can give customers $i_c=\min\{\bar\imath,1-s_0\}$ after completion: doing so never raises hazard and never tightens the budget. With this continuation shared by the reference and every policy, the value of finishing is $M(T)=b(T,1,i_c)=e^{\kappa T+\mu/2}\sqrt{i_c}$. A policy continuation that gives customers less, such as the reference reset when $\bar\imath>i_0$, provides fewer services without reducing hazard.
\end{lemma}

After completion, more customer compute provides extra services without affecting hazard, so it relaxes the budget when the reference stays fixed.

\noindent\emph{Proof of Lemma~\ref{lem:continuation}.} 
After completion, capability is fixed and the safety share is $s_0$. Hazard after completion then does not depend on how the remaining $1-s_0$ of compute is split between customers and research. The split affects only the service budget, and services after completion rise with $\min\{i,\bar\imath\}$. Given any path before completion, and with the reference and so $J_{0,\tau}$ held fixed, raising the policy's customer share after completion to $\min\{\bar\imath,1-s_0\}$ relaxes the budget without raising hazard. So an optimal continuation can use it without loss. A strict gain requires the service constraint to have positive marginal value, so that the services it frees can buy hazard reduction; with a slack budget the regulator is indifferent. Then $V_c(T)=\int_T^\tau e^{\kappa t+\mu/2}\sqrt{i_c}\,dt$, and $M(T)=-V_c'(T)=e^{\kappa T+\mu/2}\sqrt{i_c}$. The reference reset gives customers $i_0<\min\{\bar\imath,1-s_0\}$ whenever $\bar\imath>i_0$.

\noindent\emph{Proof of Corollary~\ref{cor:threshold}.} 
With $M(T)=b(T,1,i_c)$, the finishing premium is $\sqrt{i_c/i}-1$, so \eqref{eq:threshold} reads $\sqrt{i_c}<\sqrt i\,(1+\pi^\ast)$, which is \eqref{eq:capthreshold}. The continuing-danger closure gives hazard $\xi h(1)e^{-S_T-\sigma u}$ at $u$ months after completion, which integrates to $\Phi e^{-S_T}$ with $\Phi=\xi h(1)/\sigma$. Write $R_R/R_C=(A+\Phi B)/(C+\Phi D)$ with $A=H_x-H_s$, $C=-H_s$, $B=e^{-S_T}[a'(s)T_x+a(s)\ell(1)]$, and $D=e^{-S_T}a'(s)T_x$. Its derivative in $\Phi$ has the sign of $BC-AD$, which is positive exactly when $B/D>A/C$, and the ratio tends to $B/D$ as $\Phi$ grows. The threshold $\pi^\ast$, and with it the signed root $\sqrt i\,(1+\pi^\ast)$, rises with $R_R/R_C$. The clipped threshold $\bar\imath^\ast=[\max\{0,\sqrt i\,(1+\pi^\ast)\}]^2$ is nondecreasing in the signed root and rises strictly only where the root is positive.

\noindent\emph{Proof of Proposition~\ref{prop:budget}.} 
With $x=x_0$ the capability path does not depend on $i$, so neither do $t_x$, $u_x$, $\ell$, $T_x=T_0$, or $E$. By their definitions in Proposition~\ref{prop:source}, $J^{\mathrm{pre}}$ and $Q_x$ are proportional to $\sqrt i$. So $c_C=J^{\mathrm{pre}}/(2i)=\theta_C E/\sqrt i$ and $Q_x=\theta_Q E\sqrt i$ with positive constants $\theta_C$ and $\theta_Q$, and $b_T=e^{\kappa T_0+\mu/2}\sqrt i$. By \eqref{eq:threshold} and \eqref{eq:capthreshold}, $\sqrt{\bar\imath^\ast}=\sqrt i\,(1+\pi^\ast)=\sqrt i+(c_C\rho-Q_x)/E$, which is \eqref{eq:budgetthreshold}. (a) As $i\to0$ the first term vanishes and the second is at least $\theta_C\underline\rho/\sqrt i$. (b) Differentiating \eqref{eq:budgetthreshold},
\[
\frac{d\sqrt{\bar\imath^\ast}}{di}=\frac{(1-\theta_Q)\,i-\theta_C\,[\rho(i)-2i\rho'(i)]}{2i^{3/2}},
\]
and where $\sqrt{\bar\imath^\ast}>0$, $\bar\imath^\ast$ moves with its square root. If $\rho'(i)\leq0$, the bracket is at least $\rho(i)$.

Full redeployment hands customers the compute capability research no longer needs at the moment of completion. Proposition~\ref{prop:source} depends on that assumption only through the value of finishing, so it extends directly to compute that arrives gradually.

\emph{Step 1: the continuation with a rising share.} If customers' share after completion rises over time, the continuation is worth the integral of a flow whose share depends on the months since completion.

\emph{Step 2: finishing earlier moves the whole expansion.} Completing one month earlier brings every later step of the expansion one month forward, so the value of finishing is the initial flow plus the brought-forward gains.

\emph{Step 3: discounting with $\kappa>0$.} Each later step is weighted by $e^{\kappa u}>1$, so the brought-forward gains exceed the one-time step $\sqrt{i_c}-\sqrt{i_0}$ of immediate redeployment.

\begin{corollary}[Gradual redeployment]\label{cor:gradual}
Let the service horizon cover the expansion below.
\begin{itemize}
\item[(a)] Suppose customers return to $i_0$ at completion, and their use beyond it expands linearly to $i_c>i_0$ over $R>0$ months, the same for every path. At $u$ months after completion, $i_c(u)=i_0+(i_c-i_0)\min\{1,u/R\}$. Then
\[
M(T)=e^{\kappa T+\mu/2}\left[\sqrt{i_0}+\int_0^R e^{\kappa u}\,\frac{d\sqrt{i_c(u)}}{du}\,du\right],
\]
which exceeds $e^{\kappa T+\mu/2}\sqrt{i_c}$, its value under immediate redeployment, whenever $\kappa>0$.
\item[(b)] Suppose instead that serving capacity $\bar\imath(t)$ varies with calendar time, that $i<\bar\imath(t)$ before completion, and that after it customers receive the share $\min\{\bar\imath(t),1-s_0\}$. Then $M(T)=b(T,1,\min\{\bar\imath(T),1-s_0\})$, and Corollary~\ref{cor:threshold} holds with $\bar\imath$ replaced by $\bar\imath(T_x)$.
\end{itemize}
\end{corollary}

Finishing earlier brings the whole expansion forward. The discounted value of a unit of compute grows at rate $\kappa$, \alKappaPct{} percent a month in the calibration. So the later steps of a gradual expansion are worth more than an immediate step of the same size, and gradual redeployment makes delay costlier, not cheaper. In the scenario, spreading the expansion over 6 to 24 months after completion raises the finishing premium at the \alBudgetPct{} percent customer-funded optimum. It rises from \alPremiumRedeployed{} to between \alExpandPremiumSix{} and \alExpandPremiumTwentyFour{} percent. What can make delay cheap is serving capacity that grows in calendar time: finishing earlier then gives customers only what they can use at that date. Section~\ref{sec:robust} quantifies both.

\noindent\emph{Proof of Corollary~\ref{cor:gradual}.} 
(a) The continuation is worth $V_c(T)=\int_T^\tau e^{\kappa t+\mu/2}\sqrt{i_c(t-T)}\,dt$, so
\[
M(T)=-V_c'(T)=e^{\kappa T+\mu/2}\sqrt{i_c(0)}+\int_T^{T+R}e^{\kappa t+\mu/2}\,\frac{d\sqrt{i_c(u)}}{du}\bigg|_{u=t-T}\,dt,
\]
which is the stated expression since $i_c(0)=i_0$. When $\kappa>0$, $e^{\kappa u}>1$ on $(0,R]$, and $\sqrt{i_c(u)}$ rises there, so the integral exceeds $\sqrt{i_c}-\sqrt{i_0}$. The continuation depends on research only through the completion date, so the proof of Proposition~\ref{prop:source} goes through with this $M$. (b) Now $V_c(T)=\int_T^\tau e^{\kappa t+\mu/2}\sqrt{\min\{\bar\imath(t),1-s_0\}}\,dt$, whose integrand does not depend on $T$, so $M(T)$ is that integrand at $T$. With $i<\bar\imath(t)$ before completion, customer funding still costs $J^{\mathrm{pre}}/(2i)$, and the proof of Corollary~\ref{cor:threshold} applies with $\min\{\bar\imath(T_x),1-s_0\}$ as the continuation share.

\subsection{How the allocation should change over time}\label{app:timingproofs}

The minimized present-value Hamiltonian, omitting terms independent of the controls, is
\begin{equation}
\mathcal H=e^{-S}h(F)+[q_F+e^{-S}k(F)]u
                      +q_S\sigma(F)y(s)-\nu b.
\end{equation}
The service multiplier is constant because it belongs to one integral budget. The deadline constraint can change the capability costate and create junctions, but it does not depend on $S$, and in the baseline $\dot S$ does not depend on $S$. So
\[
\dot q_S=-\frac{\partial\mathcal H}{\partial S}
          =[h(F)+k(F)u]e^{-S}=\lambda.
\]
Baseline danger ends at completion and post-completion services do not depend on knowledge, so $q_S(T)=0$ and $q_S(t)=-\int_t^T\lambda(v)\,dv=-A(t)$. The continuation enters only the terminal value of capability. With continuing danger, the terminal hazard $\Phi e^{-S(T)}$ gives $q_S(T)=-\Phi e^{-S(T)}$ instead, so $-q_S$ also counts the hazard that remains after completion. Survival-weighted services or a knowledge-dependent terminal payoff would change $q_S$.

Three steps give the allocation at given shadow values.

\emph{Step 1: a sum of square roots.} Take $\varepsilon=1/2$, the learning multiplier \eqref{eq:multiplier} with $\alpha=1/2$, and the square-root service flow. Then the pointwise objective \eqref{eq:pointwise} is $b_R\sqrt x+b_S\sqrt s+b_I\sqrt i$ plus a constant. Its coefficients collect the shadow values and the technology.

\emph{Step 2: the bound.} By the Cauchy--Schwarz inequality, the sum is at most $\sqrt{b_R^2+b_S^2+b_I^2}$ when the shares sum to one.

\emph{Step 3: the maximizer.} Equality holds only when $(\sqrt x,\sqrt s,\sqrt i)$ is proportional to $(b_R,b_S,b_I)$, so each share is proportional to the square of its coefficient.

\begin{proposition}[Allocation at given shadow values]\label{prop:allocation}
Suppose $\varepsilon=1/2$, the learning multiplier is \eqref{eq:multiplier} with $\alpha=1/2$, $W,A,\nu>0$, all three uses are unrestricted within the simplex, and the serving cap does not bind at the optimum. Define
\begin{equation}
b_R=W\phi(F)e^{Gt/2},\qquad
b_S=\frac{A\sigma(F)(1-\psi)}{\sqrt{s_0}},\qquad
b_I=\nu e^{\kappa t+\mu F/2}.
\label{eq:coefficients}
\end{equation}
The unique pointwise maximizing shares are
\begin{equation}
(x,s,i)=\frac{(b_R^2,b_S^2,b_I^2)}{b_R^2+b_S^2+b_I^2}.
\label{eq:shares}
\end{equation}
This requires $b_I^2/(b_R^2+b_S^2+b_I^2)<\bar\imath$. If the cap binds, customers receive $\bar\imath$, and research and safety divide $1-\bar\imath$ in proportion to $b_R^2$ and $b_S^2$. If research's coefficient is nonpositive and the others are positive, the pointwise optimum assigns zero to research. An active deadline constraint or a lower bound on research can force a positive share.
\end{proposition}

\begin{proof}
Up to the term $A\sigma\psi$, the pointwise objective is $b_R\sqrt{x}+b_S\sqrt{s}+b_I\sqrt{i}$. Cauchy--Schwarz bounds it by $\sqrt{b_R^2+b_S^2+b_I^2}$, with equality only at \eqref{eq:shares}. A nonpositive research coefficient cannot improve on assigning that share to a positively valued use.
\end{proof}

Each use receives a share proportional to the square of its value coefficient. The coefficients in \eqref{eq:coefficients} combine the technology with the shadow values. Research's rises with the value of a capability gain, research productivity and compute. Safety's rises with the hazard still ahead and safety productivity. Customers' rises with compute, net of discounting, and with quality. With square-root technologies, the first units of any use are worth the most, so every use with positive value receives some compute. As an illustration, coefficients in the ratio $2:1:1$ for research, safety, and customers give the shares $4/6$, $1/6$, and $1/6$. The formula takes the shadow values as given; the deadline, the endpoint, and the service budget determine them. So it is not a closed form for the whole optimal path. The timing result in Proposition~\ref{prop:frontload} follows from the same marginal conditions without knowing $W$.

\begin{proof}[Proof of Proposition~\ref{prop:frontload}]
Equal marginal values for safety and customers imply \eqref{eq:ratio}, even if research is at a bound. Along the optimal path, $\dot A=-\lambda$ and the service-budget multiplier is constant. Log differentiation yields \eqref{eq:ratioslope}; its sign gives the remaining statements.
\end{proof}

\subsection{Alternative knowledge and service technologies}\label{app:timingextensions}

Other curvatures change the closed form. Take the learning multiplier \eqref{eq:multiplier} with general $\alpha$ and service value $Q^\eta$, as in Table~\ref{tab:curvature}. The service flow is then $B(t,F)i^\eta$, and equal safety and customer marginal values give
\[
A\sigma(F)\frac{(1-\psi)\alpha}{s_0^\alpha}s^{\alpha-1}
       =\nu B(t,F)\eta i^{\eta-1}.
\]
An explicit equation for $s/i$ follows when the two exponents coincide; the square-root case gives \eqref{eq:ratio}. With different exponents the path of the ratio also depends on the scale of one allocation, and the monotonicity result needs further restrictions.

With service value that diminishes faster, the direction of $s/i$ also depends on the path of the customer share. Value services as $e^{-rt}Q^\eta/\eta$ for $\eta\neq0$ and $e^{-rt}\log Q$ for $\eta=0$; the baseline square root is $\eta=1/2$. At a fixed customer share, the growth term of the customer margin, $\kappa+\mu\dot F/2$, becomes $\eta(G+\mu\dot F)-r$, which is negative for $\eta\le0$. With $\alpha=1/2$, equal safety and customer marginal values then give
\[
\frac{d}{dt}\log\frac si=2\frac{d}{dt}\log\sigma-2\frac{\lambda}{A}-2\eta(G+\mu\dot F)+2r+(1-2\eta)\frac{d}{dt}\log i,
\]
so away from $\eta=1/2$ a result like Proposition~\ref{prop:frontload} would need restrictions on how the customer share moves.

With knowledge loss $\dot S=a(t,F,s)-\delta\dot FS$, knowledge created at $t$ protects a later date $v$ with weight $e^{-\delta[F(v)-F(t)]}$, and
\[
-q_S(t)=\int_t^T e^{-\delta[F(v)-F(t)]}\lambda(v)\,dv.
\]
This weighted stock need not fall at rate $\lambda$, and a compute-scaling factor changes the safety coefficient in \eqref{eq:coefficients}.

\begin{corollary}[Transfer and absolute safety compute]\label{cor:extension}
Suppose there is no capability spillover and knowledge satisfies
\[
\dot S=\sigma(F)\left[\psi+(1-\psi)\sqrt{s/s_0}\,e^{\beta Gt}\right]
       -\delta\dot F S,
\qquad \delta,\beta\geq0.
\]
Under the conditions of Proposition~\ref{prop:frontload}, put $A_\delta(t)=\int_t^T e^{-\delta[F(v)-F(t)]}\lambda(v)\,dv$. Then
\begin{equation}
\frac{d}{dt}\log\frac{s}{i}
=2\frac{d}{dt}\log\sigma(F(t))-2\frac{\lambda}{A_\delta}
 +(2\beta-1)G+2r+(2\delta-\mu)\dot F,
\label{eq:generalratio}
\end{equation}
and the safety/customer ratio rises exactly when the right side is positive.
\end{corollary}

\begin{proof}
The knowledge costate satisfies $\dot q_S=\lambda+\delta\dot Fq_S$, so $q_S=-A_\delta$ and $\dot A_\delta=-\lambda+\delta\dot F A_\delta$. Equal safety and customer marginal values give \eqref{eq:ratio} with $A$ replaced by $A_\delta$ and the safety numerator multiplied by $e^{\beta Gt}$. Log differentiation yields \eqref{eq:generalratio}.
\end{proof}

With $\beta=1/2$, compute growth enters both margins equally and cancels. Knowledge loss adds $2\delta\dot F$, because knowledge created later travels a shorter capability distance before the hazards it protects. It also changes $A_\delta$, the path, and hazard, so a larger loss rate need not raise the realized ratio. In the uncertainty exercise, only the combined rise-and-transfer state has a positive capability coefficient, $2\zeta+2\delta-\mu=\alRiseTransferCoef$. So it is the only state in which the ratio can rise; whether it does depends on $\lambda/A_\delta$ along the path. Capability spillovers add a development effect to safety itself and fall outside the corollary.

\section{The reference reset}\label{app:reset}

The reference reset produces a large and spurious gain. Its best two-stage rule (Table~\ref{tab:known}, panel C) puts \alResStageEarlyS{} percent of compute into safety before $F_c=\alResStageSwitch$. It then cuts research to its 2 percent floor and gives customers \alResStageLateI{} percent of compute. It completes at month \alResStageEnd{} and reaches \alResStageRisk{} percent, \alResStageGainBoundary{} points below customer-funded safety.

The continuation pays for this. The rule's early safety costs \alResEarlyLoss{} percent of $\bar J$ before month 18, nearly twice the budget. Between month 18, when the reference finishes, and its own completion, the rule recovers \alResRepay{} percent. Once its capability passes \alResCrossF, a frontier serving \alResStageLateI{} percent of compute to customers out-serves a finished one serving 45 percent. Just before completion its service flow is \alResFlowRatio{} times the reference's. Delay pays for itself, so the rule wants more of it: lowering the research floor to 0.5 percent lets it use \alRelaxedDelay{} months of delay and reach \alRelaxedRisk{} percent. In the endpoint term of Proposition~\ref{prop:source}, the reset makes $\sqrt{i_c}-\sqrt i$ strongly negative late in the takeoff period.

None of this survives a consistent continuation. Under full redeployment the finished frontier serves customers better than the slowed one; under no redeployment the extra inference is worthless. The continuation is a primitive of the allocation problem, not a choice to be made for convenience.

\section{Robustness of the known-technology results}\label{app:robust}

At small budgets, who pays turns on how much compute customers can use when the frontier finishes. How fast freed compute reaches them afterward does not change the answer (Appendix~\ref{app:gradual}). Let customers' use beyond the reference share expand over 6 to 24 months after completion. The finishing premium at the \alBudgetPct{} percent customer-funded optimum then rises from \alPremiumRedeployed{} to between \alExpandPremiumSix{} and \alExpandPremiumTwentyFour{} percent, as Corollary~\ref{cor:gradual} predicts. The budget at which research becomes the cheaper marginal source rises from \alCrossBudget{} to between \alExpandCrossSix{} and \alExpandCrossTwentyFour{} percent of $\bar J$. If instead serving capacity grows in calendar time, what matters is the capacity reached when customer-funded rules finish. At the \alBudgetPct{} percent budget, research becomes the cheaper source once full use takes more than \alCriticalAbsorb{} months, though the gains from cutting it are small.

The alternative hazard constructions preserve the funding answer (Appendix~\ref{app:hazard}). The variants change the state share of hazard, how late it comes, the pause judgment that fixes safety productivity, and the reference probability. Across them the \alBudgetPct{} percent threshold stays between \alHzThrMin{} and \alHzThrMax{} percent. Under full redeployment, research becomes the cheaper source from budgets of \alHzCrossMin{} to \alHzCrossMax{} percent. Customer-funded timing remains most of the gain at the \alBudgetPct{} percent budget. Risk levels move more than these comparisons do.

\subsection{Freed compute that arrives gradually}\label{app:gradual}

Neither variation of Corollary~\ref{cor:gradual} changes customer-funded rules, which finish on the reference date; only the price of delay moves (Table~\ref{tab:gradual}).

Gradual expansion makes delay costlier. Let customers' use beyond the reference share expand over 6, 12, or 24 months after completion. The finishing premium at the \alBudgetPct{} percent customer-funded optimum then rises from \alPremiumRedeployed{} percent to \alExpandPremiumSix{}, \alExpandPremiumTwelve{}, and \alExpandPremiumTwentyFour{} percent. All three lie further above \alPremiumThreshold{} percent, the premium below which research is the cheaper source. Research becomes the cheaper marginal source only at larger budgets: from \alExpandCrossSix{}, \alExpandCrossTwelve{}, and \alExpandCrossTwentyFour{} percent of $\bar J$, against \alCrossBudget{} percent with immediate redeployment.

Serving capacity that grows in calendar time can make delay cheap. Let serving capacity rise from 45 percent at the start of the takeoff period to all installed compute at a later month. The test then compares the capacity reached by month 18, when customer-funded rules finish, with the \alThreshold{} percent threshold. Capacity at completion exceeds the threshold exactly when full use comes before month \alCriticalAbsorb. When full use comes at month 60 or 72, the best constant rule at the \alBudgetPct{} percent budget cuts research, to \alAbsorbXSixty{} and \alAbsorbXSeventyTwo{} percent. The gains are small: \alAbsorbGainSixty{} and \alAbsorbGainSeventyTwo{} points. At a \alBudgetMid{} percent budget it cuts research in every case shown, to between \alAbsorbXFifteenMin{} and \alAbsorbXFifteenMax{} percent.

\begin{table}[htbp]\centering\small
\caption{When freed compute reaches customers gradually}\label{tab:gradual}
\setlength{\tabcolsep}{4pt}
\begin{tabular}{@{}lrlrrr@{}}\toprule
& Finishing & Cheaper source & Research cheaper & \multicolumn{2}{c}{Research share (\%)}\\
\cmidrule(l){5-6}
Continuation & premium (\%) & at \alBudgetPct\% & from budget (\%) & \alBudgetPct\% budget & \alBudgetMid\% budget\\\midrule
Immediate redeployment & 64 & customers & 30.2 & 50.0 & 48.3\\
\addlinespace[3pt]\multicolumn{6}{@{}l}{\emph{Use beyond 45\% expands after completion over}}\\
\quad 6 months & 73 & customers & 34.2 & 50.0 & 49.9\\
\quad 12 months & 84 & customers & 38.6 & 50.0 & 50.0\\
\quad 24 months & 116 & customers & 48.4 & 50.0 & 50.0\\
\addlinespace[3pt]\multicolumn{6}{@{}l}{\emph{Capacity reaches 100\% at month (capacity at month 18)}}\\
\quad 36 (72.5\%) & 43 & customers & 18.3 & 50.0 & 43.4\\
\quad 48 (65.6\%) & 36 & customers & 13.2 & 50.0 & 41.1\\
\quad 60 (61.5\%) & 32 & research & 9.8 & 49.1 & 39.3\\
\quad 72 (58.8\%) & 29 & research & 7.2 & 47.7 & 38.0\\
\bottomrule\end{tabular}
\par\smallskip\begin{minipage}{\textwidth}\footnotesize
The finishing premium is the value of finishing relative to the slowed frontier's service flow at the \alBudgetPct\% customer-funded optimum (Proposition~\ref{prop:source}); research is the cheaper marginal source when it lies below \alPremiumThreshold\%. The fourth column gives the budget, in percent of $\bar J$, above which the threshold of Corollary~\ref{cor:threshold} exceeds the effective continuation share, so that research is the cheaper marginal source. Research shares are those of the best constant rule. Customer-funded rules finish on the reference date and are unchanged by either variation. Every rule meets the 12-month deadline.
\end{minipage}\end{table}

\subsection{The hazard construction}\label{app:hazard}

The hazard functions rest on four choices, each varied here one at a time (Table~\ref{tab:hazard}). Half of reference hazard is state risk, and three quarters of it falls in the second half of capability development. A 12-month pause brings the probability of catastrophe to 40 percent, which fixes safety productivity. The reference probability is 75 percent, the AI Futures Project's value, and \citet{greenblatt2025} also gives it for a takeoff period with almost no safety effort. The paper treats it as the upper end of a range and also runs 25 percent, the probability \citet{samotsvety2022} give for misaligned AI takeover by 2100. The model conditions on a full-speed takeoff period, so its probability should exceed unconditional estimates like that one; surveys report medians of 5 to 10 percent \citep{grace2025}, and Table~\ref{tab:hazard} also reports 10 and 5 percent. Each variation refits safety productivity and rebuilds the hazard functions, leaving the capability path and services unchanged.

\begin{table}[htbp]\centering\small
\caption{Key numbers under alternative hazard constructions}\label{tab:hazard}
\setlength{\tabcolsep}{3.5pt}
\begin{tabular}{@{}lrrrrrrr@{}}\toprule
& & \multicolumn{2}{c}{Customer-funded} & Gain from & Hazard & & Research\\
& & \multicolumn{2}{c}{risk (\%)} & research & advantage & Threshold & cheaper\\
\cmidrule(lr){3-4}
Hazard construction & Productivity & Constant & 12 segments & cuts & (\%) & (\%) & from (\%)\\\midrule
Baseline & 0.086 & 57.68 & 51.01 & 0.91 & 25 & 63.5 & 30.2\\
\addlinespace[3pt]\multicolumn{8}{@{}l}{\emph{State share of hazard (baseline 0.5)}}\\
\quad 0.25 & 0.085 & 57.97 & 51.40 & 0.82 & 28 & 65.3 & 28.5\\
\quad 0.75 & 0.088 & 57.39 & 50.62 & 1.02 & 22 & 61.8 & 31.8\\
\addlinespace[3pt]\multicolumn{8}{@{}l}{\emph{Share of hazard in the second half (baseline 0.75)}}\\
\quad 0.5 & 0.094 & 59.06 & 50.99 & 1.04 & 24 & 63.2 & 30.4\\
\quad 0.9 & 0.084 & 56.73 & 50.86 & 1.25 & 25 & 63.6 & 30.1\\
\addlinespace[3pt]\multicolumn{8}{@{}l}{\emph{Probability after the 12-month pause (baseline 40\%)}}\\
\quad 30\% & 0.117 & 51.69 & 43.16 & 1.06 & 26 & 64.5 & 29.3\\
\quad 50\% & 0.060 & 62.90 & 58.12 & 0.74 & 23 & 62.0 & 31.6\\
\addlinespace[3pt]\multicolumn{8}{@{}l}{\emph{Reference probability (baseline 75\%)}}\\
\quad 5\% & 0.086 & 3.13 & 2.61 & 0.07 & 25 & 63.5 & 30.2\\
\quad 10\% & 0.086 & 6.33 & 5.28 & 0.13 & 25 & 63.5 & 30.2\\
\quad 25\% & 0.086 & 16.34 & 13.76 & 0.33 & 25 & 63.5 & 30.2\\
\quad 60\% & 0.086 & 43.35 & 37.60 & 0.76 & 25 & 63.5 & 30.2\\
\quad 90\% & 0.086 & 76.03 & 69.43 & 0.94 & 25 & 63.5 & 30.2\\
\bottomrule\end{tabular}
\par\smallskip\begin{minipage}{\textwidth}\footnotesize
Each row refits safety productivity (per month) to the pause judgment and rebuilds the hazard functions; the capability path and services do not change. Risk is the probability of catastrophe at the \alBudgetPct\% budget with research at 50\%. The gain from research cuts is the further reduction in risk when twelve-segment rules may also cut research, under no redeployment. The hazard advantage is how much more hazard funding from research cuts than customer funding per unit of share at the \alBudgetPct\% customer-funded optimum; the threshold is that of Corollary~\ref{cor:threshold} there; the last column is the budget, in percent of $\bar J$, above which research is the cheaper marginal source under full redeployment. In the reference-probability rows the pause brings the probability to 1.9\%, 3.8\%, 10.1\%, 28.7\% and 57.2\%, preserving the baseline's proportional cut in cumulative hazard.
\end{minipage}\end{table}

At the \alBudgetPct{} percent customer-funded optimum, funding from research cuts between \alHzAdvMin{} and \alHzAdvMax{} percent more hazard per unit of share than customer funding. The threshold stays far below the full-redeployment share and above the no-redeployment one, and the budget at which research becomes the cheaper marginal source barely moves (Section~\ref{sec:robust}). Front-loading customer-funded safety lowers the probability of catastrophe by \alHzTimingMin{} to \alHzTimingMax{} points relative to the constant rule. Allowing research cuts adds \alHzSatAddMin{} to \alHzSatAddMax{} points under no redeployment. Under full redeployment it adds at most \alHzRedAddMax{} points, and any cut comes after $F=0.7$.

Risk levels move more than these comparisons. The pause judgment moves fitted productivity between \alHzSigmaMin{} and \alHzSigmaMax{} per month, and the twelve-segment customer-funded risk between \alHzCFMin{} and \alHzCFMax{} percent. Scaling the reference probability, with the pause's proportional cut in hazard held fixed, changes risk levels only. Hazard scales by a common factor, so productivity, the threshold, and the crossing budget, which depend on ratios of hazard, do not move. At 25 percent the gain from also cutting research under no redeployment falls from \alHrSatAddHi{} to \alHrSatAddLo{} points, and to \alHrSatAddTen{} and \alHrSatAddFive{} at 10 and 5 percent.

\subsection{A longer takeoff period}\label{app:takeoff36}

A 36-month takeoff period keeps the funding answer (Table~\ref{tab:takeoff36}). Thirty-six months is a slow case. In the AI Futures Model's published forecasts for this stretch, the medians run from 9 to 18 months and the 90th percentiles from 3.6 to 15.5 years \citep{lifland2026}. Only research speed, hazard functions and safety productivity are refitted to the 75 and 40 percent targets. The budget unit covers each column's own takeoff period, so the \alBudgetPct{} percent budget buys the same \tkBoundarySLong{} percent customer-funded safety share in both. The threshold moves from \tkThresholdBase{} to \tkThresholdLong{} percent: customers still pay under full redeployment, research under no redeployment. A point taken from research now delays completion by \tkDelayPerPointLong{} months rather than \tkDelayPerPointBase{}, but customer funding also gives up services over a longer period.

Twelve-segment rules gain a little more from also cutting research: \tkRedTwelveAddLong{} points under full redeployment and \tkSatTwelveAddLong{} under no redeployment, against \tkRedTwelveAddBase{} and \tkSatTwelveAddBase{} at baseline. The threshold reaches 95 percent at a budget of \tkCrossBudgetLong{} percent, against \tkCrossBudgetBase{} at baseline.

\begin{table}[htbp]\centering\small
\caption{Headline results with a 36-month takeoff period}\label{tab:takeoff36}
\begin{tabular}{@{}lrr@{}}\toprule
& 18 months & 36 months\\
& (baseline) & \\\midrule
\addlinespace[3pt]\multicolumn{3}{@{}l}{\emph{A. Calibration}}\\
\quad Full-speed takeoff period (months) & \tkTakeoffBase & \tkTakeoffLong\\
\quad Growth of installed compute over the takeoff period (fold) & \tkComputeFoldBase & \tkComputeFoldLong\\
\quad Catastrophe probability at full speed (\%) & \tkFullRiskBase & \tkFullRiskLong\\
\quad Fitted safety productivity $\sigma$ (per month) & \tkSigmaBase & \tkSigmaLong\\
\quad Budget unit $\bar J$ (index units) & \tkUnitBase & \tkUnitLong\\
\addlinespace[3pt]\multicolumn{3}{@{}l}{\emph{B. Budget of \alBudgetPct\% of $\bar J$}}\\
\quad Safety share of the customer-funded constant rule (\%) & \tkBoundarySBase & \tkBoundarySLong\\
\quad Delay of completion per point taken from research (months) & \tkDelayPerPointBase & \tkDelayPerPointLong\\
\quad Funding threshold (\%) & \tkThresholdBase & \tkThresholdLong\\
\quad Catastrophe probability, customer-funded twelve segments (\%) & \tkCFTwelveRiskBase & \tkCFTwelveRiskLong\\
\quad Gain from also cutting research, full redeployment (points) & \tkRedTwelveAddBase & \tkRedTwelveAddLong\\
\quad Gain from also cutting research, no redeployment (points) & \tkSatTwelveAddBase & \tkSatTwelveAddLong\\
\quad Threshold rise at twice baseline safety productivity (points) & \tkYieldDoubleRiseBase & \tkYieldDoubleRiseLong\\
\quad Threshold change when continuing danger dominates (points) & \tkTailRiseBase & \tkTailRiseLong\\
\addlinespace[3pt]\multicolumn{3}{@{}l}{\emph{C. Larger budgets, full redeployment}}\\
\quad Budget at which the threshold reaches 95\% (\% of $\bar J$) & \tkCrossBudgetBase & \tkCrossBudgetLong\\
\quad Funding threshold at a \alBudgetHigh\% budget (\%) & \tkThresholdThirtyBase & \tkThresholdThirtyLong\\
\quad Gain from also cutting research, \alBudgetMid\% budget (points) & \tkRedTwelveAddFifteenBase & \tkRedTwelveAddFifteenLong\\
\quad Gain from also cutting research, \alBudgetHigh\% budget (points) & \tkRedTwelveAddThirtyBase & \tkRedTwelveAddThirtyLong\\
\bottomrule\end{tabular}
\par\smallskip\begin{minipage}{\textwidth}\footnotesize
\emph{Notes:} The 36-month column lengthens the full-speed takeoff period and refits three inputs by the construction of Appendix~\ref{app:calibration}: the research constant $\phi_0$, so that full speed completes at month 36; the hazard functions, so that full speed still gives a 75\% catastrophe probability; and safety productivity $\sigma$, so that a 12-month stopped-clock pause at the start still gives 40\%. Everything else is held fixed: compute doubles every six months; research productivity and model quality rise tenfold by the endpoint; reference shares are 50/5/45; the hazard construction, the discount rate, and the continuation cases are unchanged. The deadline is still 12 months, so every rule completes by month 48, within the 60-month service horizon. Budgets are shares of each column's own budget unit $\bar J$, the reference's discounted services over its takeoff period, so the \alBudgetPct\% budget buys the same customer-funded safety share in both columns. Thresholds are evaluated at the customer-funded constant rule (Corollary~\ref{cor:threshold}); the continuing-danger row is the threshold's change in the limit where continuing danger dominates. Customer-funded rules keep research at 50\% throughout. A gain is the fall in catastrophe probability when twelve-segment rules may also cut research, against the best customer-funded twelve-segment rule. Segment rules are feasible multistart candidates, not certified global optima.
\end{minipage}\end{table}

\section{Uncertain technologies and state outcomes}\label{app:states}

\subsection{The commitment problem}

The uncertain technologies vary safety productivity, knowledge transfer to stronger models, the scaling of safety work with installed compute, capability spillovers, and danger after completion. The baseline and nine states that each change one feature are joined by three that combine adverse changes. The physical hazard functions $h$ and $k$ stay fixed, so full-speed risk differs across states. The states describe uncertainty about one physical frontier.

The commitment is chosen to minimize the worst regret. Let $\mathcal D$ be the constant-share grid. In technology $\theta$, let $\mathcal D_\theta(B)$ be the rules that meet the service budget and the deadline, and define the informed probability
\begin{equation}
P_\theta^\ast(B)=\min_{d\in\mathcal D_\theta(B)}P_\theta(d).
\label{eq:oracle}
\end{equation}
This is the lowest probability of catastrophe on the constant-share grid if the regulator knew the technology. The regulator commits to
\begin{equation}
\min_{d\in\bigcap_\theta\mathcal D_\theta(B)}
\max_\theta\{P_\theta(d)-P_\theta^\ast(B)\},
\label{eq:regret}
\end{equation}
the rule that minimizes the largest excess probability of catastrophe relative to an informed choice, with no prior over states. The rule must meet both commitments in every state (Table~\ref{tab:robust}). The informed benchmark uses the same grid and commitments.

\begin{table}[htbp]\centering\small
\caption{Fixed rules chosen before the safety technology is known}\label{tab:robust}
\setlength{\tabcolsep}{4pt}
\begin{tabular}{@{}rrrrrrr@{}}\toprule
Budget & Research & Safety & Worst & Largest regret & Customer-funded & Max delay\\
(\%) & (\%) & (\%) & regret & without spillover & worst regret & (months)\\\midrule
\addlinespace[3pt]\multicolumn{7}{@{}l}{\emph{A. Full redeployment}}\\
11.5 & 50.0 & 14.5 & 5.90 & 0.00 & 5.90 & 0.00\\
34.5 & 49.5 & 30.5 & 6.52 & 0.28 & 6.69 & 0.06\\
69.1 & 36.5 & 41.5 & 7.78 & 0.17 & 10.65 & 1.81\\
\addlinespace[3pt]\multicolumn{7}{@{}l}{\emph{B. No redeployment}}\\
11.5 & 41.0 & 14.0 & 4.86 & 0.15 & 8.92 & 1.13\\
34.5 & 29.5 & 27.0 & 7.37 & 0.10 & 13.37 & 3.09\\
69.1 & 21.5 & 41.0 & 8.36 & 0.17 & 23.58 & 5.06\\
\addlinespace[3pt]\multicolumn{7}{@{}l}{\emph{C. Reference reset}}\\
11.5 & 36.0 & 13.0 & 4.79 & 0.12 & 10.34 & 1.89\\
34.5 & 28.0 & 26.0 & 7.35 & 0.23 & 13.37 & 3.41\\
69.1 & 21.5 & 41.0 & 8.36 & 0.24 & 23.58 & 5.06\\
\bottomrule\end{tabular}
\par\smallskip\begin{minipage}{\textwidth}\footnotesize
Regret is in percentage points of catastrophe probability, relative to the best rule on the same grid chosen with knowledge of the technology. Budgets are shares of $\bar J$; service cost counts 60 months. Each rule must meet the budget and the 12-month deadline in all 13 technologies. The customer-funded column restricts the chosen rule to research at 50\% and keeps the same informed benchmark. Exhaustive share grid: spacing 0.005, research 2--50\%, safety 5--70\%, customers at least 1\%.
\end{minipage}\end{table}

\subsection{Knowledge laws of the technologies}

Each technology $\theta$ has its own knowledge law,
\begin{equation}
\dot S=\sigma f_\theta n_\theta e^{\zeta_\theta F}
\left[\psi+(1-\psi)\left(\frac{s}{s_0}\right)^\alpha
e^{\beta_\theta Gt}\right]-\delta_\theta\dot F S,
\label{eq:uncertainS}
\end{equation}
where $f_\theta$ scales safety productivity and $\zeta_\theta$ is the log of its rise across the takeoff period. The normalization $n_\theta$ makes rising productivity, at reference shares, add the same knowledge by the endpoint as a flat one:
\[
n_\theta=\frac{T_0}{\int_0^1 e^{\zeta_\theta F}/u_0(F)\,dF}.
\]
A transfer coefficient $\delta_\theta=\log2$ leaves half of an initial stock's protection at the endpoint. Safety research scaled by installed compute uses $\beta_\theta=1/2$. Capability spillover replaces the research multiplier by
\begin{equation}
\left(\frac{x+\chi_\theta s}{x_0+\chi_\theta s_0}\right)^\varepsilon,
\qquad \chi_\theta=0.2\ \text{in spillover states},
\label{eq:spill}
\end{equation}
which leaves the reference path unchanged but lets extra safety compute speed capability. Spillover is the only feature that changes a rule's capability path, and so its service cost and delay.

Continuing-danger states add the post-completion hazard
\begin{equation}
H_{\mathrm{tail}}=\frac{\xi_\theta h(1)e^{-S(T)}}{\sigma f_\theta n_\theta e^{\zeta_\theta}},
\qquad \xi_\theta\in\{0.25,1\},
\end{equation}
with knowledge growing at the endpoint productivity and the reference safety share after completion. In the 0.25 case, state hazard runs at one quarter of its endpoint rate right after completion. Hazard primitives are not rescaled when the tail is added, so these states carry very high probabilities of catastrophe.

\begin{table}[htbp]\centering\footnotesize
\caption{Technology states with fixed hazard primitives: the \alBudgetPct\% rules}\label{tab:states}
\setlength{\tabcolsep}{4pt}
\begin{tabular}{@{}lrrrrrrr@{}}\toprule
& Full speed & \multicolumn{2}{c}{Full redeployment} & \multicolumn{2}{c}{No redeployment} & \multicolumn{2}{c}{Reference reset}\\
\cmidrule(lr){3-4}\cmidrule(lr){5-6}\cmidrule(l){7-8}
Technology & risk (\%) & Risk (\%) & Regret & Risk (\%) & Regret & Risk (\%) & Regret\\\midrule
Baseline & 75.00 & 58.05 & 0.00 & 55.32 & 0.10 & 54.58 & 0.00\\
Quarter productivity & 97.92 & 96.75 & 0.00 & 96.69 & 0.04 & 96.73 & 0.08\\
Half productivity & 93.54 & 88.38 & 0.00 & 87.62 & 0.08 & 87.49 & 0.08\\
Higher productivity & 51.11 & 30.71 & 0.00 & 27.91 & 0.06 & 27.11 & 0.00\\
Rising productivity & 77.93 & 62.49 & 0.00 & 60.21 & 0.15 & 59.70 & 0.06\\
Partial transfer & 87.37 & 77.08 & 0.00 & 75.45 & 0.12 & 75.10 & 0.07\\
Absolute safety compute & 59.22 & 35.23 & 0.00 & 30.97 & 0.00 & 29.41 & 0.00\\
Capability spillover & 75.00 & 58.68 & 5.44 & 56.13 & 4.55 & 55.46 & 4.48\\
Small continuing hazard & 95.60 & 84.68 & 0.00 & 81.93 & 0.05 & 81.04 & 0.00\\
Continuing hazard & 99.98 & 99.25 & 0.00 & 98.80 & 0.00 & 98.62 & 0.00\\
Weak productivity, transfer, spillover & 98.64 & 98.06 & 0.24 & 98.08 & 0.17 & 98.13 & 0.22\\
Half productivity, rise, transfer & 96.71 & 94.46 & 0.00 & 94.28 & 0.06 & 94.33 & 0.12\\
Compute scaling and spillover & 59.22 & 36.18 & 5.90 & 32.13 & 4.86 & 30.64 & 4.79\\
\bottomrule\end{tabular}
\par\smallskip\begin{minipage}{\textwidth}\footnotesize
Each case uses its own minimax-regret rule at the \alBudgetPct\% budget (full redeployment, research 50.0\% and safety 14.5\%; no redeployment, research 41.0\% and safety 14.0\%; reference reset, research 36.0\% and safety 13.0\%). Regret is in percentage points. The individual variants change productivity to 0.25, 0.5, or 1.5 times baseline; raise productivity threefold; allow partial transfer; use absolute safety compute; add a 0.2 spillover; or retain 0.25 or 1 times the endpoint state hazard. The last three rows combine changes. No state is refitted to a common catastrophe probability.
\end{minipage}\end{table}

Probability regret compresses differences near one, so a hazard-regret criterion could rank rules differently in the tail states. Neither criterion is invariant to an arbitrary change in the hazard scale. Table~\ref{tab:states} reports the levels behind the regret calculation.

\section{Numerical construction and validation}\label{app:numerics}

For rules whose shares change at capability milestones, arrival dates have a closed form. Let $m_j$ be the research multiplier in segment $j$. If segment $j$ begins at $(F_j,t_j)$, the research clock \eqref{eq:clock} gives
\begin{equation}
e^{\gamma t(F)}=e^{\gamma t_j}+\frac{z(F)-z(F_j)}{m_j},
\qquad F\in[F_j,F_{j+1}].
\end{equation}
Within segment $j$ the clock runs at the constant multiplier $m_j$, starting where the previous segment stopped. So the date at which each capability level arrives has a closed form within every segment. Knowledge is integrated over capability with the integrating factor $e^{\delta_\theta F}$, and hazard as $(h/u+k)e^{-S}$. Discounted services are integrated until completion; the continuation flow is then added in closed form. Each segment has its own capability grid, and the node at each change of shares is duplicated so that no interval bridges the change.

Constant rules are optimized from six starting points and two-stage rules from eight, at each switch milestone in Table~\ref{tab:switch}. The best switch is then re-solved on 1601 capability nodes. In every case at least \alStartsAgree{} of 8 two-stage starts reach the same best rule. The finer mesh changes its probability of catastrophe by at most \alStageRefinement{} percentage points. Returned candidates are checked against the resource, budget, and deadline constraints with a small explicit service margin. Agreement across starts is evidence of numerical stability, not a proof of global optimality.

Six- and twelve-segment rules use the same formula on each segment, with the milestones listed under Table~\ref{tab:classes}. Each is optimized from three starting points, on 161 capability nodes per segment, and re-evaluated on 801 nodes per segment. At the \alBudgetPct{} percent budget, one start for the rules that may also cut research maps the best two-stage rule onto the segments. In every case at least \alSegmentAgree{} of 3 starts reach the same rule within $10^{-6}$ in cumulative hazard. Customer-funded rules fix research at 50 percent and choose only the safety shares. Customer-funded two-stage rules are optimized at every switch milestone of Table~\ref{tab:switch}. A calendar-time solver with an event at every milestone reproduces each segment rule within $\alSegmentCheck$ in cumulative hazard.

The control mesh is checked separately from the integration mesh. Every twelve-segment rule is re-solved with 16 segments, refined near the start of the takeoff period, and with 24 segments. Its probability of catastrophe changes by at most \alMeshRiskChange{} percentage points. The gain from allowing research cuts changes by at most \alMeshGainChange{} points, and the largest milestone delay by at most \alMeshDelayChange{} months. The milestones the text cites are the same on every mesh: where research is cut, where it returns to 50 percent, and where safety reaches its floor. Shares in the first segment change more, because the research cut grows deeper and shorter as that segment narrows. So the text reports when research is cut, not how deeply.

\begin{table}[htbp]\centering\small
\caption{Best two-stage rules by switch milestone}\label{tab:switch}
\begin{tabular}{@{}rrrr@{}}\toprule
Switch $F_c$ & Full redeployment & No redeployment & Reference reset\\\midrule
0.05 & 55.24 & 52.34 & 51.04\\
0.10 & 53.73 & 51.38 & 48.96\\
0.20 & 52.29 & 51.06 & 46.35\\
0.30 & 52.14 & 51.55 & 44.72\\
0.40 & 52.77 & 52.24 & 43.21\\
0.50 & 53.63 & 52.92 & 38.60\\
0.60 & 54.53 & 53.49 & 38.84\\
0.70 & 55.40 & 53.90 & 41.74\\
0.80 & 56.22 & 54.18 & 46.20\\
0.90 & 56.98 & 54.47 & 51.20\\
\bottomrule\end{tabular}
\par\smallskip\begin{minipage}{.72\textwidth}\footnotesize
Probability of catastrophe (\%) of the best two-stage rule at each switch milestone, under the \alBudgetPct\% budget and the 12-month deadline; 801 capability nodes. Table~\ref{tab:known} reports the best switch in each case, re-solved on 1601 nodes.
\end{minipage}\end{table}

Table~\ref{tab:capacity} repeats the known-technology comparison across serving capacity, with the continuation of Lemma~\ref{lem:continuation} and no continuing danger. The best constant rules there and in Figure~\ref{fig:source} use the same multistart routine; the figure's curve uses 801 capability nodes.

\begin{table}[htbp]\centering\small
\caption{Best rules across serving capacity}\label{tab:capacity}
\setlength{\tabcolsep}{4pt}
\begin{tabular}{@{}rrrrrrrr@{}}\toprule
Serving & \multicolumn{3}{c}{Best constant rule} & \multicolumn{4}{c}{Best two-stage rule}\\
\cmidrule(lr){2-4}\cmidrule(l){5-8}
capacity (\%) & Research & Safety & Risk & Research & Safety & Risk & Switch\\\midrule
45 & 40.7 & 14.3 & 54.83 & 44.6/47.2 & 27.2/7.8 & 51.06 & 0.2\\
55 & 45.4 & 13.9 & 57.22 & 47.5/40.0 & 26.7/5.0 & 50.56 & 0.3\\
65 & 50.0 & 14.8 & 57.68 & 50.0/42.7 & 26.4/5.0 & 51.78 & 0.3\\
75 & 50.0 & 14.8 & 57.68 & 50.0/49.9 & 25.8/5.0 & 52.14 & 0.3\\
85 & 50.0 & 14.8 & 57.68 & 50.0/50.0 & 25.8/5.0 & 52.14 & 0.3\\
100 & 50.0 & 14.8 & 57.68 & 50.0/50.0 & 25.8/5.0 & 52.14 & 0.3\\
\bottomrule\end{tabular}
\par\smallskip\begin{minipage}{.9\textwidth}\footnotesize
Shares and the probability of catastrophe in percent, under the \alBudgetPct\% budget and the 12-month deadline, with the continuation of Lemma~\ref{lem:continuation} and no continuing danger. Customer-funded safety gives 57.68\% at every capacity. Two-stage entries show shares before/after the switch milestone, chosen from the grid of Table~\ref{tab:switch} and re-solved on 1601 nodes.
\end{minipage}\end{table}

The uncertainty exercise searches an exhaustive grid. Research shares run from 0.02 to 0.50 and safety shares from 0.05 to 0.70 in steps of 0.005, with customers at least 0.01. That gives \alGridCount{} rules, each evaluated in all \alStateCount{} technologies and three continuations on 1601 capability nodes. Re-evaluating every rule on 401 and 801 nodes changes no chosen rule; the largest change in any probability is \alIntegrationGap{} percentage points.

Independent methods check the integrator. Adaptive quadrature over capability reproduces the constant rules' hazard. A calendar-time solver with events at the switch and the endpoint reproduces the two-stage paths. Both agree within $\alPathCheck$ in cumulative hazard. The derivatives in Propositions~\ref{prop:funding} and~\ref{prop:source} agree with centered differences of independently integrated hazard and service cost within a relative error of $\alDerivativeCheck$. They also agree under the gradual continuations of Corollary~\ref{cor:gradual}. The square-root allocation \eqref{eq:shares} agrees with a constrained numerical maximization within $\alSqrtGap$ in shares. These checks test paths, derivatives, and optimization; they do not validate the assumed hazard or safety productivity.

\setlength{\bibsep}{2pt plus 1pt}
\section{Notation}\label{app:notation}

The table follows first appearance in the main text. The allocation coefficients $b_R,b_S,b_I$ are defined in Appendix~\ref{app:timingproofs}.

\begin{table}[htbp]\centering\small
\caption{Notation}\label{tab:notation}
\begin{tabular}{@{}>{$}l<{$}>{\raggedright\arraybackslash}p{12.0cm}@{}}
\toprule
\multicolumn{1}{@{}l}{Symbol} & Meaning\\
\midrule
F,\ T & Capability, from 0 to the endpoint 1; completion date\\
K(t)=e^{Gt} & Installed compute, growing at rate $G$\\
x,\ s,\ i & Shares for capability research, safety research, and customer inference\\
x_0,\ s_0,\ i_0 & Reference (full-speed) shares\\
\phi(F),\ \varepsilon,\ \gamma=\varepsilon G & Research productivity of a model of capability $F$; research elasticity; growth of research capacity\\
t_0(F),\ u_0(F) & Arrival date and pace at $F$ on the reference path\\
m(t),\ z(F) & Research multiplier $(x/x_0)^\varepsilon$; research capacity the reference has accumulated on reaching $F$, $e^{\gamma t_0(F)}-1$\\
S,\ \sigma(F),\ y(s) & Knowledge (log-hazard reduction); safety productivity; learning multiplier\\
h(F),\ k(F),\ \lambda,\ H,\ P & State and step hazard; hazard rate; cumulative hazard; probability of catastrophe\\
\mu,\ r,\ \kappa=G/2-r & Quality elasticity; monthly discount rate; growth of the discounted service flow\\
\bar\imath,\ i_c,\ \xi & Serving capacity; customer share after completion; continuing danger\\
b,\ J_\tau,\ C_\tau,\ \bar J,\ B & Service flow; discounted services and service cost over $\tau$; budget unit; service budget\\
L & Largest allowed milestone delay\\
\phi_0,\ \psi,\ \alpha & Scenario forms: $\phi(F)=\phi_0\cdot10^F$; $y(s)=\psi+(1-\psi)(s/s_0)^\alpha$\\
t_x(F),\ u_x(F),\ a(s) & Arrival date and pace at constant research share $x$; knowledge per month $\sigma y(s)$\\
\ell(F),\ \vartheta(F) & Delay in months per unit of research share; proportional pace change per unit of research share\\
H_x,\ H_s & Hazard effects of the research and safety shares (Proposition~\ref{prop:funding})\\
J^{\mathrm{pre}},\ b_T,\ M(T),\ \Phi & Services before completion; service flow just before completion; value of finishing; continuing-hazard coefficient before knowledge protection\\
c_C,\ c_R,\ Q_x & Service cost of a unit of safety share when customers or research pay; quality part of $c_R$\\
R_C,\ R_R,\ \pi^\ast & Hazard reductions per unit share when customers or research pay; threshold finishing premium\\
\bar\imath^\ast,\ \rho,\ \theta_C,\ \theta_Q,\ E & Funding threshold; $R_R/R_C$; coefficients of Proposition~\ref{prop:budget}\\
\nu,\ q_F,\ q_S,\ A(t),\ W(t) & Budget multiplier; costates; hazard still ahead; net value of a capability gain\\
\zeta & Growth of safety productivity with capability, $\sigma(F)=\sigma_0e^{\zeta F}$\\
b_R,\ b_S,\ b_I & Coefficients on the three square-root terms in Proposition~\ref{prop:allocation}\\
\bottomrule
\end{tabular}
\end{table}

\FloatBarrier

\end{document}